\documentclass[11pt, oneside]{article}   	
\usepackage{geometry}                		
\usepackage{graphicx, subfigure}				
\usepackage{amsthm}
\usepackage{mathrsfs}
\usepackage{amsmath}
\usepackage{amsfonts}
\usepackage{amssymb}
\usepackage{comment}
\usepackage{natbib}
\usepackage{hyperref}
\usepackage{xcolor}
\usepackage{bbm}

\usepackage{amssymb}

\newcommand{\oo}{\mathscr{O}}
\newcommand{\1}{\textbf{\textup{1}}}

\newcommand{\Ru}{R_U}
\newcommand{\Rvh}{R_{\widehat V}}
\newcommand{\Ruh}{R_{\widehat U}}

\newcommand{\A}{\mathscr{A}}

\newcommand{\vv}{\mathcal{V}}

\newcommand{\R}{\mathbb{R}}
\newcommand{\E}{\mathbb{E}}
\newcommand{\pr}{\mathbb{P}}
\newcommand{\N}{\mathbb{N}}

\newcommand{\pp}{{\mathcal{P}}}

\newtheorem{lemma}{Lemma}[section]
\newtheorem{theorem}{Theorem}[section]
\newtheorem{definition}{Definition}

\newtheorem{corollary}{Corollary}[section]
\newtheorem{prop}{Proposition}[section]
\newtheorem{remark}{Remark}[section]

\newtheorem{assumption}{Assumption}

\title{Vintage Factor Analysis with Varimax Performs Statistical Inference}
\author{Karl Rohe and Muzhe Zeng}

\begin{document}
\maketitle

\begin{abstract}

Psychologists developed Multiple Factor Analysis to decompose multivariate data into a small number of interpretable factors without any \textit{a priori} knowledge about those factors  \citep{thurstone1935vectors}. 
In this form of factor analysis, the Varimax ``factor rotation'' is a key step to make the factors interpretable \citep{kaiser}. Charles Spearman and many others objected to factor rotations because the factors seem to be rotationally invariant
 \citep{thurstone1947, anderson1956statistical}.  These objections are still reported in all contemporary multivariate statistics textbooks.  
This is an engima because this vintage form of factor analysis has survived and is widely popular because, empirically, the factor rotation often makes the factors easier to interpret. 
We argue that the rotation makes the factors easier to interpret because, in fact,  the Varimax factor rotation performs  statistical inference.
We show that Principal Components Analysis (PCA) with the Varimax rotation
provides a unified spectral estimation strategy for a broad class of modern factor models, including the Stochastic Blockmodel and a natural variation of Latent Dirichlet Allocation (i.e., ``topic modeling'').
In addition, we show that Thurstone's widely employed sparsity diagnostics implicitly assess a key ``leptokurtic'' condition that makes the rotation statistically identifiable in these models. 
Taken together, this shows that the know-how of Vintage Factor Analysis performs statistical inference, reversing nearly a century of statistical thinking on the topic. 
With a sparse eigensolver, 
PCA with Varimax is both fast and stable.
Combined with Thurstone's straightforward diagnostics, 
this vintage approach is suitable for a wide array of  modern  applications.

\end{abstract}

\vspace{.05in} \hspace{.1in}

\vspace{.1in} 

Outside the language of mathematical statistics, Louis~Leon Thurstone, Henry Kaiser, and other psychologists developed the first forms of Multiple Factor Analysis, or what is referred to herein as Vintage Factor Analysis 
\citep{thurstone1935vectors, thurstone1947, kaiser}.  There are two simultaneous aims of Vintage Factor Analysis.  The first aim is to provide a low dimensional approximation of the observed data; in this sense, it is like Principal Components Analysis (PCA).\footnote{PCA 
is not the preferred approach in Vintage Factor Analysis.  See Remark \ref{remark:pcavsfactor} for a further discussion.}
 The second aim is to ensure that each factor in the lower dimensional representation, each coordinate, represents a ``scientifically meaningful category'' \citep{thurstone1935vectors}.
 A Varimax rotation of the principal components is a simple and popular way to find such meaningful dimensions \citep{kaiser, jolliffe2002principal}.

 For example, suppose $n$ students take an exam with $d$ questions, producing a $d$ dimensional vector of data for each individual.  
Principal components analysis with $k$=2 dimensions will roughly approximate the students' $d$ dimensional data; this is the first aim of factor analysis.
In order to make those two dimensions more interpretable, these principal components are rotated with the Varimax rotation.  In other words, Varimax provides a different coordinate basis for the two dimensional space.  Selecting the basis does not change the quality of the lower dimensional approximation.  However, after inspecting the $k$=2 Varimax coordinates, an analyst might find that one coordinate represents ``linguistic intelligence'' and the other coordinate represents ``logical-mathematical intelligence.''  This form of data analysis is often called ``exploratory'' because the factor dimensions are computed from the data without needing any hypothesis that specifies them.

Factor analysis is an enigma.  The key source of the controversy in Vintage Factor Analysis is the second aim, producing coordinates that correspond to ``scientifically meaningful categories.''
\cite{anderson1956statistical} formalized the concern by showing that under the Gaussian factor model, all rotations achieve the same fit.
This result implies that under the Gaussian factor model, the individual Varimax coordinates  cannot estimate anything meaningful. 
\cite{cosma} gives the conventional interpretation of this result,
 ``If we can rotate the factors as much as we like without consequences, how on Earth can we interpret them?'' 
Contemporary multivariate analysis textbooks all discuss the result from \cite{anderson1956statistical}, but then go on to report the empirical benefits of the factor rotation.  \cite{ramsay2007applied} says ``It is well known in classical multivariate analysis that an appropriate rotation of the principal components can, on occasion, give components ... more informative than the original components themselves.''  \cite{johnson2007applied} says ``A rotation of the factors often reveals a simple structure and aids interpretation.''  \cite{bartholomew2011latent} says ``Rotation assumes a very important role when we come to the interpretation of latent variables.''  \cite{jolliffe2002principal} says  ``The simplification achieved by rotation can help in interpreting the factors or rotated PCs.''  
These empirical findings appear to be inconsistent with the results of \cite{anderson1956statistical} that are described in those same textbooks.  


Varimax is the most popular way of computing a factor rotation \citep{kaiser}.  It is discussed in all of the textbooks cited in the previous paragraph. 
\cite{fda} describes Varimax as an ``invaluable tool in multivariate analysis.''  
It is contained in the base \texttt{R} packages, akin to \texttt{kmeans}, and is so popular that it is often not properly cited.
 Given an $n \times k$ matrix $U$, with columns that form an orthonormal basis (e.g. as in PCA), Varimax 
finds a $k\times k$ orthogonal matrix $R$ 
to maximize the following function over the set of $k \times k$ orthonormal matrices
\begin{equation}\label{eq:Varimax}
v(R, U) = 
\sum_{\ell=1}^k \frac{1}{n} \sum_{i=1}^n \left([UR]_{i \ell}^4 - 
\left(\frac{1}{n} \sum_{q=1}^n [UR]_{q \ell}^2\right)^2\right).
\end{equation}
\cite{kaiser} suggests normalizing each row of $U$ such that each row has a sum of squares equal to one. For simplicity, we do not use this normalization herein.\footnote{In \texttt{R}, the function \texttt{varimax} has a default argument \texttt{normalize = TRUE}.  Note that when $U$ has orthogonal columns (as is the case for PCA) and normalization is not used, then the second term in Varimax is a  constant function of the matrix $R$.  In such cases, this term can be ignored without changing the optimum.}


%


Factor rotations have survived for nearly a century because a rotation often makes the factors more interpretable. Yet the classical theoretical results do not explain how or why.  
Maxwell's Theorem  resolves the enigma
[\cite{maxwell} and III,4 in \cite{feller}].
It characterizes the multivariate  Gaussian distribution as the only distribution of independent random variables that is rotationally invariant.  This implies that
the rotation is partially identifiable, so long as the factors are independent and come from \textit{any} non-Gaussian distribution.  As such, 
if the independent latent factors are generated from a non-Gaussian distribution, then
the factor rotation has the potential to identify these factors as ``scientifically meaningful categories.'' 
See Figure \ref{fig:rotationallyInvariant} for an example in $k=2$ dimensions.\footnote{A common point of confusion is to presume that the factors must be Gaussian if we are using PCA; see Section \ref{sec:intuition} and Remark \ref{remark:pca} to see how PCA performs with non-Gaussian factors.}


 \begin{figure}[h] 
   \centering
   \includegraphics[width=5.5in]{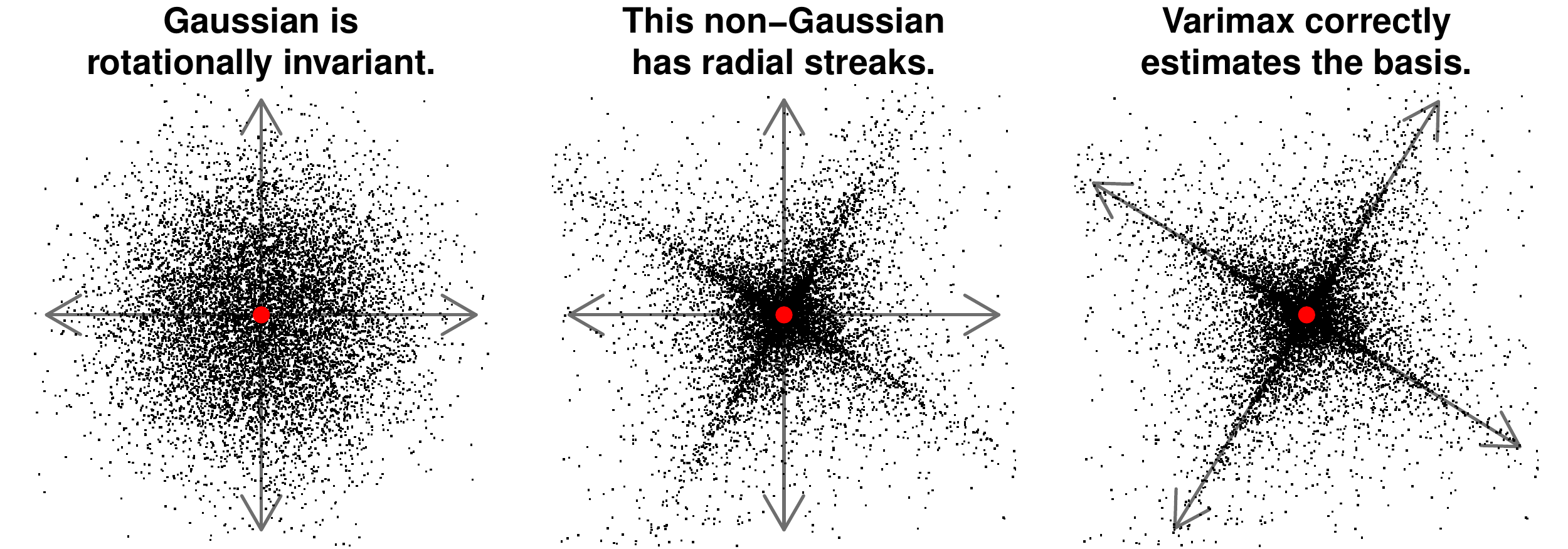} 
   \caption{
   Maxwell's Theorem characterizes the multivariate Gaussian distribution (left panel) as the only rotationally invariant distribution of independent variables. 
   The center panel and the right panel give the same data; the only difference is that the right panel gives the basis that is well estimated by Varimax. 
   }
   \label{fig:rotationallyInvariant}
\end{figure}

\begin{figure}[h]
\centering
\textbf{In this data example, the principal components (left) have radial streaks.
\\ Varimax aligns the streaks with the axes (right). \\
Varimax rotated PCA is Vintage Sparse PCA, \texttt{vsp}.
}

\subfigure[Principal Components]{\label{fig:scatterRotateA} \includegraphics[width = 2.8in]{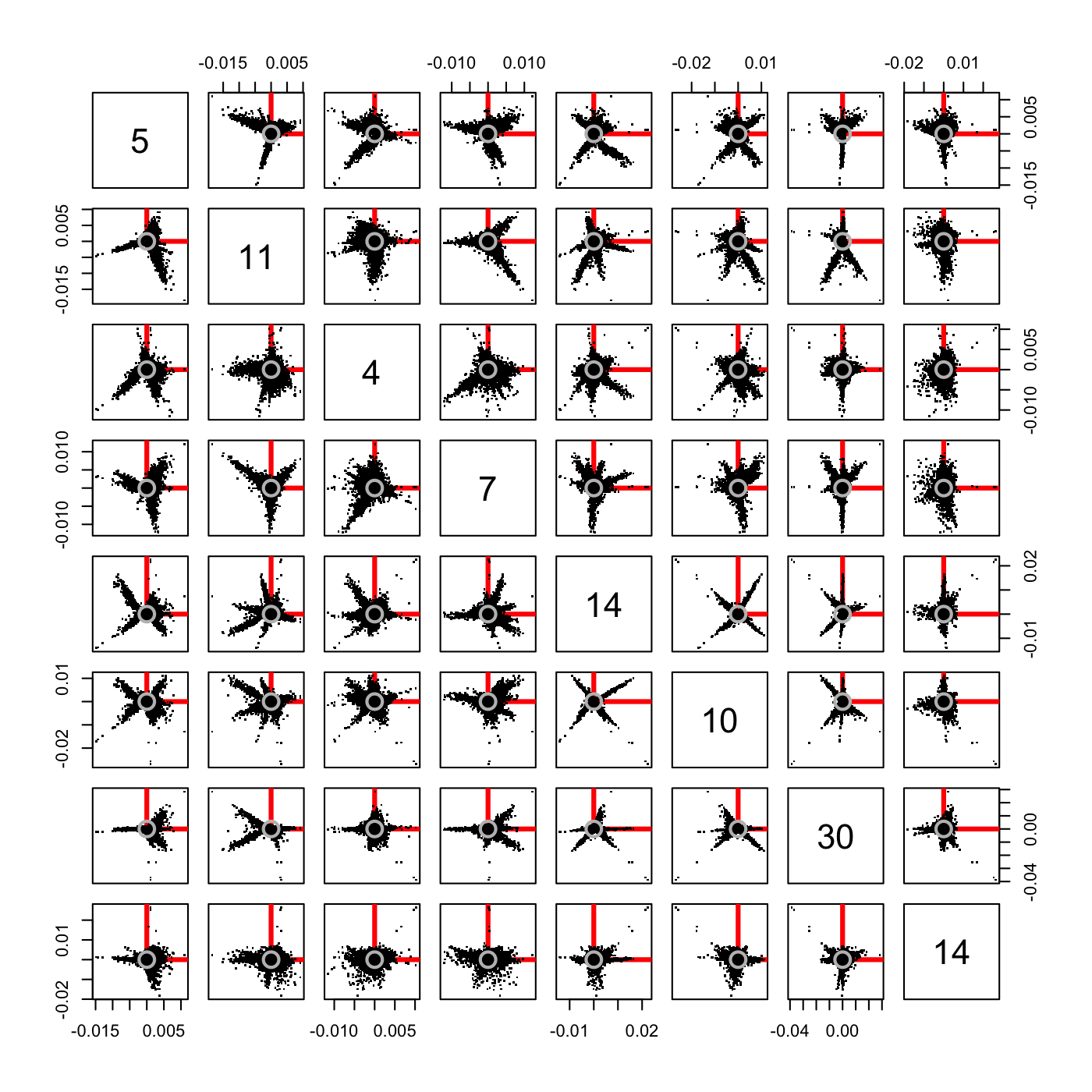} }
\subfigure[After Varimax rotation]{\label{fig:scatterRotateB} \includegraphics[width = 2.8in]{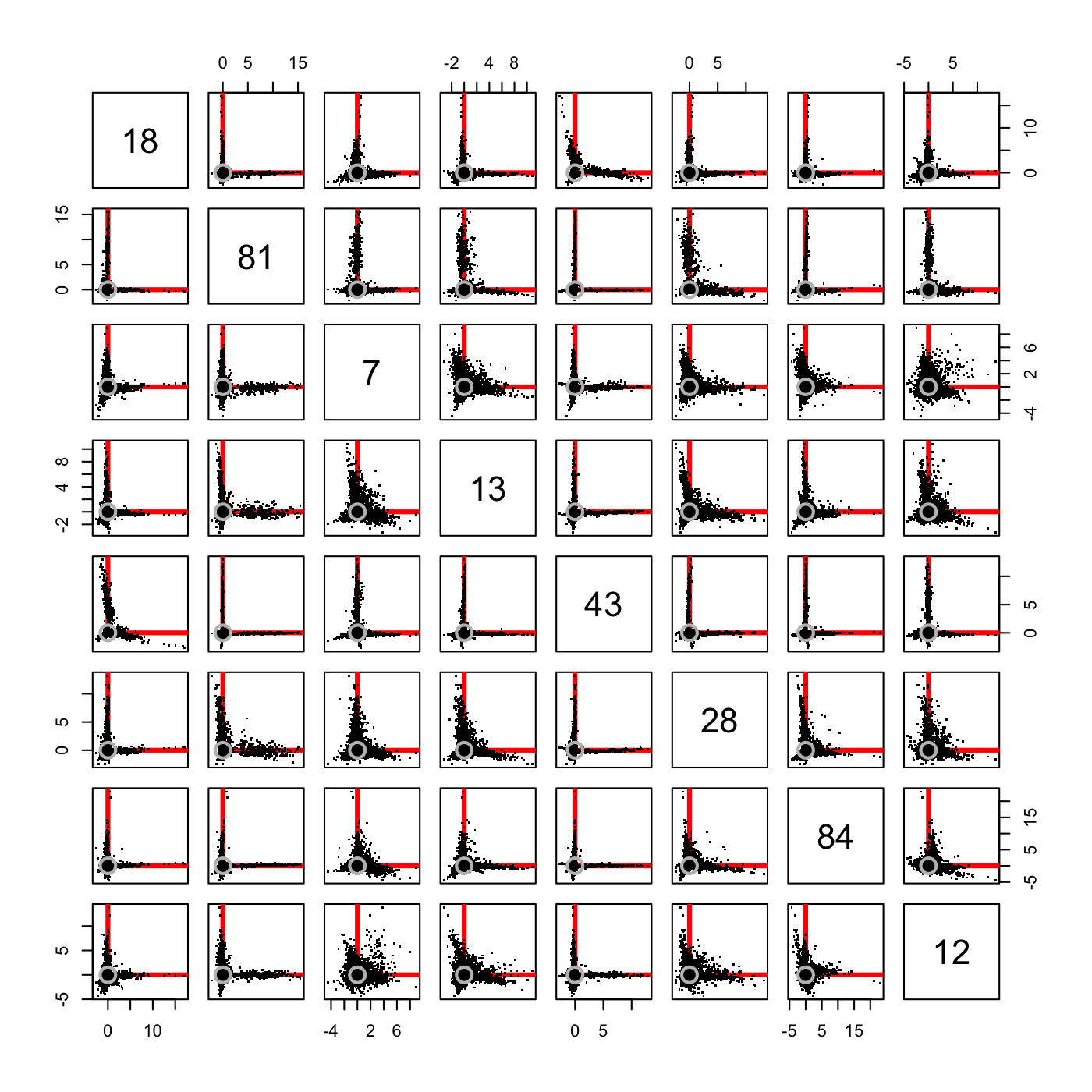} }
\caption{In this example, the data is a $300,000 \times 102,660$ document-term matrix of $300,000$ New York Times articles. Each small panel on the left is a scatter plot of two principal components. Each small panel on the right is a scatter plot of two Varimax rotated components.   The numbers down the diagonal give the sample kurtosis of the corresponding component. 
See Section \ref{sec:nyt} for more details. 
}

\label{fig:nytFactors}
\end{figure}


Maxwell's theorem and some of the core factor analysis methodologies have been rediscovered and further developed in the literature on Independent Components Analysis (ICA) \citep{ica}.  More recently, \cite{anandkumar2014tensor} showed how a tensor decomposition can estimate a broad class of factor models that is closely related to the class studied below.  
This paper demonstrates that an old approach with historical precedence to ICA is sufficient; tensor methods are not required.  This old approach comes with a suite of know-how and diagnostic practices that are described in Section \ref{sec:diagnostics}.  This old approach provides a unified spectral estimation strategy and diagnostic practices that can be applied to many different problems in multivariate statistics.  It relates Projection Pursuit, Independent Components Analysis, Non-Negative Matrix Decompositions, Latent Dirichlet Allocation, and Stochastic Blockmodeling.

Figure \ref{fig:nytFactors} shows a motivating data example with 
a set of 300,000 New York Times articles \citep{nyt}.  
In this example, the data matrix $A$ is a $300,000 \times 102,660$ document-word matrix, where $A_{ij} \in \{0,1\}$ indicates if document $i$ contains word $j$.
Figure \ref{fig:scatterRotateA} plots nine of the leading principal components.  Figure \ref{fig:scatterRotateB} plots these components after a Varimax rotation.  
  Section \ref{sec:alg} describes this procedure in more detail.  See Section \ref{sec:nyt} for further details on the data analysis in Figure \ref{fig:nytFactors}.  


All the plots in Figure \ref{fig:nytFactors} display ``radial streaks,'' a phrase used in \cite{thurstone1947} to diagnose factor rotations.  
In Figure \ref{fig:scatterRotateB}, the Varimax rotation aligns the streaks with the coordinate axes.  This is precicely the desired outcome of a factor rotation, to  make the rotated components approximately sparse.  
For this reason, this paper refers to Varimax rotated PCA as Vintage Sparse PCA (\texttt{vsp}). Modern notions of Sparse PCA (e.g. \cite{d2005direct}) presume that the principal components are themselves sparse.  In the vintage notion of sparse PCA, it is presumed that there exists a set of sparse basis vectors for the principal component subspace. That is, perhaps the principal components are not sparse, but they become sparse after a rotation.  These are two distinct notions of subspace sparsity.  
\cite{vince} referred to the vintage notion of sparsity as \textit{column-wise sparsity}.  


Theorem \ref{thm:main} shows that, under certain conditions, \texttt{vsp}  estimates the following semi-parametric factor model that generalizes the Stochastic Blockmodel and Latent Dirichlet Allocation.

\begin{definition}\label{def:model}  
Let $Z \in \R^{n \times k}$ and $Y \in \R^{d \times k}$ be  latent factor matrices.  
Under the \textbf{semi-parametric factor model}, we observe $A \in \R^{n \times d}$ which has independent elements and has expectation 
\begin{equation}\label{eq:model}
\E(A| Z, Y) = ZBY^T, \quad  \mbox{where } B \in \R^{k \times k} \mbox{ is not necessarily diagonal.}
\end{equation}
\end{definition}
Importantly, in the semi-parametric factor model, the columns of $Z$ are not the principal components.  
However, if the elements of $Z$ are independently generated from a ``leptokurtic'' distribution, then a Varimax rotation of the principal components estimates the columns of $Z$.  This leptokurtic condition this is the key identifying assumption for Varimax and \texttt{vsp}.
\begin{definition} \label{def:leptokurtic} 
For a random variable $X \in \R$ with four finite moments, let $\eta = \E(X)$ and define the $j$th centered moment as $\eta_j = \E(X - \eta)^j$ for $j = 2, 4$.  The kurtosis of $X$ is $\kappa =  \eta_4 / \eta_2^2 $. The random variable $X$ and its distribution are \textbf{leptokurtic} if $\kappa>3$.
\end{definition}
\noindent
Kurtosis was originally named and used by Pearson around 1900 to measure whether a symmetric distribution was Gaussian \citep{fiori2009karl}. 
 For any Gaussian random variable, $\kappa=3$.  As such, $\kappa \ne 3$ indicates a non-Gaussian distribution.  
Roughly speaking, when $\kappa >3$, the distribution has a heavier tail than Gaussian.  

Section \ref{sec:alg} describes the \texttt{vsp} algorithm and some variations on the algorithm.  Section \ref{sec:diagnostics} reinterprets the sparsity diagonstics developed in \cite{thurstone1935vectors, thurstone1947} to show that they implicitly assess the key identifying assumption for \texttt{vsp} to estimate the semi-parametric factor model (i.e., whether the factors in the columns of $Z$ appear leptokurtic).
In particular, Section \ref{sec:thurstone} discusses Thurstone notion of ``simple structure'' (a form of sparsity), his conjecture that simple structure resolves the rotational invariance, and his sparsity diagnostics that are described in modern textbooks, built into the base \texttt{R} packages for factor analysis, and used routinely in practice.  
Then, Theorem \ref{thm:sparsity} shows that any random variable $X$ that satisfies $P(X = 0) > 5/6$ (i.e., it is sparse) is necessarily leptokurtic.  
%
%
%
In this way, Thurstone's sparsity diagonstics and the know-how of Vintage Factor Analysis 
can be reinterpreted as assessing a key identifiability assumption for Varimax.

Sections \ref{sec:intuition} gives intuition for why \texttt{vsp} can estimate the latent factors by giving population results.  The first results show that the column space of the principal components of $\A = E(A|Z,Y) = ZBY^T$ equals to the column space of $Z$. 
Then, a Varimax rotation of the principal components specifies a new set of basis vectors for that column space.
Under the identifying assumption where the elements of $Z$ are  generated independently from a leptokurtic distribution and the entire distribution of $Z$ is known (i.e., infinite sample size), Theorem \ref{thm:varimax1} shows that each of the new basis vectors is estimating an individual column of $Z$ (up to a sign change).  
%
If the elements of $Y$ satisfy the same conditions required for $Z$, then $Y$ and $B$ can also be estimated, \textit{even when $B$ is not diagonal}.  Section \ref{sec:theoremstatement} gives the main theoretical result, Theorem \ref{thm:main}, which shows that \texttt{vsp} can estimate $Z$, when the matrix $A$ is high dimensional and random. This result allows for $A$ to be sparse and is enabled by recent technical developments that provide ``element-wise'' eigenvector bounds for random graphs \citep{erdHos2013spectral, cape2019signal, mao2018overlapping}.  
%
%
%
%
%
%
Section \ref{sec:examples} describes how the broad class of semi-parametric factor models includes the Stochastic Blockmodel, several of its generalizations, and a natural extension of Latent Dirichlet Allocation. Corollaries \ref{corollary:dcsbm} and \ref{corollary:pois} extend Theorem \ref{thm:main} to these models.

\vspace{.1in}\noindent
\textbf{Key Notation:} Let $\oo(k) = \{R \in \R^{k \times k}: R^T R= R R^T = I_k\}$ denote the set of $k\times k$ orthonormal matrices.  Let $\1_a \in \R^a$ be a column vector of ones.  Let $I_d$ denote the $d \times d$ identity matrix.  For $x \in \R^d$, let $diag(x) \in \R^{ d\times d}$ be a diagonal matrix with $diag(x)_{ii} = x_i$.  For $M \in \R^{a \times b}$, define $M_i \in \R^b$ as the $i$th row of $M$ and $\|M\|_{p\rightarrow \infty} = \max_i \|M_i\|_p$, for $p\ge 1$ and $\ell_p$ norm for vectors $\|\cdot\|_p$. Let $\|M\|_F$ be the Frobenius norm, $\|M\|$ be the spectral norm, $\|M\|_{\infty}$ be the maximum absolute row sum of $M$, and $\|M\|_{\max}$ be the maximum element of $M$ in absolute value. 
For sequences $x_n,y_n \in \R$, define $x_n \asymp y_n$ to mean that $x_n \rightarrow \infty$ and $y_n \rightarrow \infty$  and there exists an $N, \epsilon,$ and $c$ all in $(0, \infty)$ such that $x_n/y_n\in (\epsilon, c)$ for all $n>N$.  Define  $x_n \succeq y_n$ to mean that for any $\epsilon\in (0, \infty)$, there exists an $N < \infty$ such that for all $n>N$, $x_n / y_n > \epsilon>0$.   Define $[k] = \{1, \dots, k\}$.

\section{\texttt{vsp}: Vintage Sparse PCA}\label{sec:alg}

This section describes the methodological details of 
Vintage Sparse PCA (\texttt{vsp}).
First, the algorithm is stated.  Then, 
Remarks \ref{remark:scalingstep}, \ref{remark:powerMethod}, and \ref{remark:uncentering} describe ways in which  \texttt{vsp} can be modified for certain settings; Table \ref{table:options} summarizes these settings.   Section \ref{sec:nyt} illustrates the algorithm with a corpus of New York Times articles; this is the analysis that generated  Figure \ref{fig:nytFactors} above.  



\begin{enumerate}
\item[] Algorithm: \texttt{vsp}
\item[-]  Input $A \in \R^{n \times d}$ and desired number of dimensions $k$. 
%
%
%
\item \textbf{Centering} (optional).  Define row, column, and grand means, 
\[\widehat \mu_r = A\1_d/d \in \R^n, \quad 
\widehat \mu_c = \1_n^T A/n \in \R^d, \quad
\widehat \mu_. = \1_n^T A \1_d /(nd) \in \R.\] 
Here $\widehat \mu_r$ is a column vector and $\widehat \mu_c$ is a row vector.  Define 
\begin{equation} \label{eq:centeredA}
\widetilde A = A - \widehat \mu_r\1_d^T - \1_n \widehat \mu_c + \widehat \mu_. \1_n \1_d^T \in \R^{n \times d}.
\end{equation}
If $A$ is large and sparse, step 1 and 2 can be accelerated. See Remark \ref{remark:powerMethod}.

\item \textbf{SVD}.  If centering is being used, then compute the top $k$ left and right singular vectors of $\widetilde A$, $\widehat U\in \R^{n \times k}$ and $\widehat V\in \R^{d \times k}$. These are the principal components and their loadings.  Let $\widehat D \in \R^{k \times k}$ be a diagonal matrix  containing the corresponding singular values.  So, $\widetilde A \approx \widehat U\widehat D\widehat V^T$.  If centering is not being used, then use the original input matrix $A$ instead of $\widetilde A$.
\item \textbf{Varimax}.  
Compute   the orthogonal matrices that maximize Varimax, 
$v(R, \widehat U)$ and $v(R, \widehat V)$.  Define them as $\Ruh, \Rvh \in \oo(k)$ respectively.

\item[-]  Output:
\begin{equation}
\widehat Z = \sqrt{n} \widehat U \Ruh, \quad \widehat Y = \sqrt{d} \widehat V \Rvh, \quad \mbox{ and } \quad \widehat B =  R_{\widehat U}^T\widehat D R_{\widehat V} / \sqrt{nd}
\end{equation}
\end{enumerate}


In modern applications where the row sums (or column sums) of $A$ are highly heterogeneous, the scaling step in the next remark is often considered before a spectral decomposition. One can apply this step before \texttt{vsp} and input the scaled matrix $L$ into \texttt{vsp}. This step does not appear as part of Vintage Factor Analysis.  Rather, it has emerged from recent work on spectral clustering \citep{chaudhuri2012spectral, amini2013pseudo}. 

\begin{remark} \label{remark:scalingstep} [Optional scaling step]   Define the row ``degree'', the row regularization parameter, and the diagonal degree matrix as
\[deg_r = A\1_d \in \R^n, \quad
\tau_r = \1_n^Tdeg_r /n \in \R,\quad
D_r = diag(deg_r + \tau_r \1_n) \in \R^{n \times n}.\]
Similarly, define the column quantities $deg_c, \tau_c, D_c$ with $deg_c = \1_n^T A \in \R^d$ and $\tau_c = deg_c \1_d / d$. 
%
Define the scaled (or normalized) adjacency matrix as $L = D_r^{-1/2}AD_c^{-1/2}.$  Then, input $L$ to \texttt{vsp} (instead of $A$).  When using $L$, \texttt{vsp} estimates a scaled version of $Z$ and $Y$.  To undo this, the output of \texttt{vsp} could be ``rescaled'' as $D_r^{1/2} \widehat Z$ and $D_c^{1/2} \widehat Y$.  Even when $L$ is used, this paper never rescales the output.
\end{remark}


Normalizing the adjacency matrix with the regularizer $\tau$ improves the statistical performance of spectral estimators derived from a sparse random matrix  \citep{le2017concentration}.  
In many empirical examples, the $\tau_r$ and $\tau_c$ prevent large outliers in the elements of the singular vectors  that are created as an artifact of noise in sparse matrices \citep{zhang2018understanding}.   
In this paper, the scaling step is used for the analysis of the New York Times data, but it is not studied in the main theorem.


\begin{remark}\label{remark:powerMethod} [Fast computation for sparse data matrices] 
In many contemporary applications, $A$ is sparse (i.e., most elements $A_{ij}$ are zero).  In this case, the SVD step should be computed with power methods.  These methods are faster and require less memory  because they only require matrix-vector multiplication.  Moreover, if step 1 is being used, then the centered matrix $\widetilde A$ should not be explicitly computed.  
Instead,
the matrix-vector multiplications can be computed as the right hand side of the following equality,
\begin{equation}\label{eq:fastmatrixvector}
\widetilde A x =  Ax - \widehat \mu_r(\1_d^Tx) - \1_n (\widehat \mu_cx) + \widehat \mu_. \1_n (\1_d^T x),
\end{equation}
and similarly for $y \widetilde A$.
When computed naively, the left hand side of Equation \eqref{eq:fastmatrixvector} requires $O(nd)$ operations.  
However, the right hand size requires $O(\texttt{nnz})$ operations, 
where $\texttt{nnz}$ is the number of nonzero elements in $A$.  In the New York Times example displayed in Figure \ref{fig:nytFactors}, $\texttt{nnz}$ is three orders of magnitude smaller than $nd$.
Using Equation \eqref{eq:fastmatrixvector} also dramatically reduces the amount of memory required to store the matrices.  This can be used in conjunction  with the scaling step in Remark \ref{remark:scalingstep}. This is implemented in an \texttt{R} package available on GitHub \citep{githubCode} 
 using the \texttt{R} packages Matrix and rARPACK \citep{Matrix, rARPACK}. 
\end{remark}

The optional centering step (step 1 of \texttt{vsp}) plays a surprising role.  In particular, 
Proposition \ref{prop:centering} in Section \ref{sec:intuition} shows that if $A$ is centered in step 1, then \texttt{vsp} estimates the centered factors in the semi-parametric factor model (i.e., $Z - \E(Z)$). See Remark \ref{remark:centering} for more discussion.  To estimate $Z$, recenter  $\widehat Z$ as follows.

\begin{remark}  \label{remark:uncentering} [Optional recentering step]
After running \texttt{vsp} with the centering step, it is possible to use the quantities already computed to  recenter the estimated factors $\widehat Z$ and $\widehat Y$ as a post-processing step. This enables \texttt{vsp} to estimate $Z$ instead of $Z - \E(Z)$.
Define
\begin{equation} \label{eq:muhat}
\widehat \mu_Z=  \sqrt{n} \widehat \mu_c \widehat V \widehat D^{-1} \Ruh, \quad \mbox{ and } \quad \widehat \mu_Y=  \sqrt{d} \widehat \mu_r^T \widehat U \widehat D^{-1} \Rvh
\end{equation}
and recenter the estimated factors as follows: $\widehat Z  +  \1_n\widehat\mu_Z$ and  $\widehat Y +  \1_d\widehat\mu_Y$.  If the rescaling in Remark \ref{remark:scalingstep} is also used, then recenter before rescaling.  Section \ref{sec:intuition} and Appendix \ref{app:recentering} justify the estimator $\widehat \mu_Z$.
\end{remark}

Table \ref{table:options} below lists the variations of \texttt{vsp} that are defined above and discussed in this paper. 

\vspace{.2in}
\noindent
\begin{table}[h]
\begin{center}
\begin{tabular}{|l|l|}
\hline
\textbf{Option} & \textbf{Motivated when ...} \\
\hline 
Centering   & factor modeling, topic modeling, soft-clustering. \\
& See Remarks \ref{remark:uncentering} and \ref{remark:centering}, Theorem \ref{thm:main}, Section \ref{sec:lda}, Corollary \ref{corollary:pois} \\
\hline
Recentering & the factor means are desired. \\ 
& See Theorem \ref{thm:main}, Remark \ref{remark:centering}, Section \ref{app:recentering}.\\
\hline
Avoid centering & hard-clustering, Stochastic Blockmodeling. \\
&  See Section \ref{sec:block}, Corollary \ref{corollary:dcsbm}. \\
\hline
Scaling  & heterogeneous column sum or row sums in $A$. \\
& Used in the data example. \\
\hline
Rescaling  & we want to estimate the distribution of the factors $Z$. \\
& See Remark \ref{remark:scalingstep}.\\
\hline
\end{tabular}
\caption{The motivation for each of the optional steps in \texttt{vsp}.} \label{table:options}
\end{center}
\label{defaulttable}
\end{table}

\subsection{Data example} \label{sec:nyt}

In Figure \ref{fig:nytFactors}, the data matrix $A$ is a  $300,000 \times 102,660$ document-term matrix from a collection of $300,000$ New York Times articles. 
In this example, the row and column sums of $A$ are highly heterogeneous, ranging several orders of magnitude. As such, the matrix $A$ was scaled as in Remark \ref{remark:scalingstep} and \texttt{vsp} was given $L$. 
In \texttt{vsp}, the centering step (step 1) and the recentering step (\ref{remark:uncentering}) were used.
Given that the signs of the principal components and the factors are arbitrary, the sign of each principal component and each Varimax factor was chosen to make the third sample moment (i.e., skew) positive.  

After computing the leading $k=50$ principal components, twelve were removed because they localized on a relatively few number of articles (i.e., these twelve principal components were dominated by a few outliers) \citep{zhang2018understanding}.  Figure \ref{fig:nytscree} shows the screeplot of the remaining 38 singular values and a gap at $k=8$.  The Varimax rotation for these $k=8$ principal components was recomputed.  These are the eight principal components and eight Varimax factors displayed in Figure \ref{fig:nytFactors}. Each panel should display 300,000 points. To prevent overplotting, the display only shows a sample of 5000 points. The inclusion probability for point $i$ is proportional to $\|\widehat Z_i\|_2$, where $\widehat Z_i$ is the $i$th row of $\widehat Z$.

With  $k=50$ dimensions \texttt{vsp} takes roughly two minutes in \texttt{R} on a 3.5 GHz 2017 MacBook Pro with the packages \texttt{Matrix} for sparse matrix calculations and \texttt{rARPACK} for sparse eigencomputations \citep{Matrix, rARPACK}.  Recomputing the Varimax rotation for the leading $k=8$ principal components takes roughly two seconds. This example is documented at \url{github.com/RoheLab/vsp-paper}. The  \texttt{R} package is available at \url{github.com/RoheLab/vsp} \citep{githubCode}.


\section{Rotational invariance, simple structure, Thurstone's diagnostics, kurtosis, and sparsity} \label{sec:diagnostics}

\begin{quote}

Any rotation of the factors fits the data equally well; this is what is meant by ``rotational invariance.''
Thurstone proposed using sparsity to remove this invariance. His sparsity diagnostics are still used routinely in practice.  Theorem \ref{thm:sparsity} shows that sparsity implies the key leptokurtic condition that is sufficient for Varimax to identify the rotation.  In this way, Vintage Factor Analysis performs statistical inference.
\end{quote}

Step 2 of \texttt{vsp} approximates $\widetilde A$ with the leading $k$ singular vectors, $\widetilde A \approx \widehat U \widehat D \widehat V^T$. Step 3 computes the Varimax rotations of $\widehat U$ and $\widehat V$.  However, for any rotation matrices $R_1, R_2\in \oo(k)$, rotating $\widehat U$ and $\widehat V$ does not change the approximation to $\widetilde A$,
\[\widehat U \widehat D \widehat V^T = (\widehat U R_1) (R_1^T \widehat DR_2)  (\widehat VR_2)^T,\]
where the rotated factor matrices $\widehat U R_1$ and $\widehat VR_2$ still have orthonormal columns. As such, no rotation can improve the approximation to $\widetilde A$. Many have interpreted this to imply that we can never estimate factor rotations from data.  This is the misunderstanding of rotational invariance.  

In an attempt to resolve the rotational invariance, Thurstone developed a new type of data analysis to find rotations $R_{\widehat U} \in \oo(k)$ such that $\widehat U R_{\widehat U}$ is sparse \citep{thurstone1935vectors, thurstone1947}.  He developed a suite of tools and diagnostics to assess this sparsity and many of these remain in use today.  They are described in modern textbooks, built into the base \texttt{R} packages for factor analysis, and used routinely in practice.  Section \ref{sec:thurstone} describes these diagnostic practices. Section \ref{sec:leptokurtosis} and Theorem \ref{thm:sparsity} show how these diagnostics can be reinterpreted as assessing whether the factors come from a leptokurtic distribution which is a key condition for Varimax to be able to identify the correct factor rotation in Theorems \ref{thm:varimax1} and \ref{thm:main}.

%
%
%
%

\subsection{Thurstone's simple structure and diagnostics} \label{sec:thurstone}

\cite{thurstone1935vectors} and \cite{thurstone1947} propose using sparsity to remove the rotational invariance. 
``In numerical terms this is a demand for the [rotation which provides] the smallest number of non-vanishing entries in each row of the ... factor matrix.  It seems strange indeed, and it was entirely unexpected, that so simple and plausible an idea should meet with a storm of protest from the statisticians'' [p333 \cite{thurstone1947}].  Thurstone refers to this sparsity in the rotated factor matrix as \textit{simple structure}.  Thurstone's use of sparsity is analogous to the modern use of sparsity in high dimensional regression and underdetermined systems of linear equations.  In these more modern problems, without any sparsity constraint, there is a large space of plausible solutions. However, under certain conditions, the sparse solution is unique.
This intuition is analogous to Thurstone's intuition  for resolving rotational invariance. 

Thurstone implemented techniques to find rotations which produce sparse solutions, but he struggled to find any assurance that the computed solution is the sparsest solution. ``When [a solutions has] been found which produces a simple structure, it is of considerable scientific interest to know whether the simple structure is unique...  The necessary and sufficient conditions for uniqueness of a simple structure need to be investigated. In the absence of a complete solution to this problem, five criteria will here be listed which probably constitute sufficient conditions for the uniqueness of a simple structure''  [p334 \cite{thurstone1947}].    
Thurstone's five conditions motivate his ``radial streaks'' diagnostic, illustrated in Figure \ref{fig:nytFactors}.  In the quote below, Thurstone's original mathematical notation has been replaced with the notation in this paper.  

  \vspace{.1in}
  \noindent
\framebox{%
\begin{centering}
   \begin{minipage}{5.6in}
     
     \vspace{-.15in}
     
\textbf{\center{ Five rules for simple factor structure; quoted from \cite{thurstone1947} p335}}
     
     \vspace{.05in}   
     
We shall describe five useful criteria by which the $k$ reference vectors [i.e., the columns of $R_{\widehat U}$] can be determined. These are as follows:
\begin{enumerate}
\item Each row of the .... matrix $\widehat UR_{\widehat U}$ should have at least one zero.  
\item For each column $\ell$ of the factor matrix $\widehat UR_{\widehat U}$ there should be a distinct set of $k$ linearly independent [rows] whose factor loadings $[\widehat UR_{\widehat U}]_{j\ell}$ are zero. [sic\footnote{There cannot be $k$ linearly independent vectors in a $k-1$ dimensional hyperplane.}]
\item For every pair of columns of $\widehat UR_{\widehat U}$ there should be several [rows] whose entries $[\widehat UR_{\widehat U}]_{jp}$ vanish in one column but not in the other.
\item For every pair of columns of $\widehat UR_{\widehat U}$, a large proportion of the tests should have zero entries in both columns. This applies to factor problems with four or five or more common factors. 
\item For every pair of columns there should preferably be only a small number of [rows] with non-vanishing entries in both columns.
\end{enumerate}
When these [five] conditions are satisfied, the plot of each pair of columns shows (1) a large concentration of points in two radial streaks, (2) a large number of points at or near the origin, and (3) only a small number of points off the two radial streaks. For a configuration of $k$ dimensions there are $\frac{1}{2}k(k-1)$ diagrams. When all of them satisfy the three characteristics, we say that the structure is `compelling,' and we have good assurance that the simple structure is unique. In the last analysis it is the appearance of the diagrams that determines, more than any other criterion, which of the hyperplanes of the simple structure are convincing and whether the whole configuration is to be accepted as stable and ready for interpretation.*

------------------------------------------

\small{*Ever since I found the simple-structure solution for the factor problem, I have never attempted interpretation of a factorial result without first inspecting the diagrams. [footnote original to text]}
 \end{minipage}
 \end{centering}
  }
  \vspace{.1in}

An example of the diagrams (i.e., plots) that Thurstone proposes are given in Figure \ref{fig:nytFactors}  for the New York Times data.  Each of those plots displays radial streaks.  After the Varimax rotation, those radial streaks align with the coordinate axes, making the rotated factors approximately sparse.

If the diagnostic plots do not show radial streaks, Thurstone suggests that one should proceed more cautiously.    
  A few pages after the quote above, Thurstone gives a diagram with points evenly spaced inside a circle (i.e., rotationally invariant) and explains what happens when you have loadings that appear to come from a rotationally invariant distribution. ``A figure such as [this] leaves one unconvinced, no matter where the axes are drawn, unless an interpretation can be found that seems right.  Random configurations like this seldom yield clear interpretations, but they are not, of course, physically impossible.''

The current paper creates a statistical theory around Thurstone's key ideas by presuming that the factors are generated as random variables from a statistical model and using the Varimax estimator. 
 Thurstone does not presume the latent factors are generated from a probability distribution per se, and as such, does not cite or recognize the importance of Maxwell's Theorem.  Moreover, Thurstone computed rotations by hand and human judgement.  Only after Thurstone's death in 1955 did it become popular to compute factor rotations such as Varimax on ``electronic computers'' with numerical optimization techniques.

\subsubsection{Simple structure in contemporary multivariate statistics}

Contemporary textbooks on multivariate statistics still suggest that the rotated factors or the rotated principal components should be inspected to see if they are sparse  \citep{mardia1979multivariate,  jolliffe2002principal, johnson2007applied, bartholomew2011latent}.  
These textbooks all share the empirical observation that it is often easier to interpret factors which have been rotated for sparsity. The given reason is that sparse factors are ``simpler.'' While this appears to use Thurstone's word, these texts do not discuss whether or not this simple structure might resolve the problem of rotational invariance.  Rather, it is an empirical observation that sparse and simple solutions are easier to interpret.    
For example, 
``The simplification achieved by rotation can help in interpreting the factors or rotated PCs'' \citep{jolliffe2002principal}. Similarly,  ``A rotation of the factors often reveals a simple structure and aids interpretation''  \citep{johnson2007applied}.  
The notion that the data analyst should inspect the factors for sparsity is built into the \texttt{print} function for factor loadings in the base \texttt{R} packages; if a loading is less than the  \texttt{print} argument \texttt{cutoff}  then instead of printing a number, it appears as a whitespace.


This paper shows that sparsity does not merely make the factors simpler;  sparsity enables statistical identification and inference.  \textit{Sparsity} and ``\textit{radial streaking}'' are two distinctively non-Gaussian patterns.  As such, Thurstone's visualizations  and diagnostics can be reinterpreted as assessing whether the factors are generated from a non-Gaussian distribution and thus, by Maxwell's theorem, whether the rotation is statistically identifiable. The next section shows that if a distribution is sufficiently sparse, then it is leptokurtic. 

\subsection{Kurtosis and sparsity} \label{sec:leptokurtosis}



The next theorem shows that sparsity implies leptokurtosis.
In this way, Thurstone's sparsity diagnostics can be reinterpreted as assessing an identifying assumption for Varimax.  Moreover, sparsity can replace leptokurtosis in the identifying assumptions for Varimax.  



\begin{theorem} \label{thm:sparsity}
Any random variable $X$ that satisfies $\frac{5}{6} < \pr(X=0) <1$ and has four finite moments is leptokurtic.
\end{theorem}

This theorem does not make any parametric assumptions and the moment assumptions are only so that kurtosis is defined.  See Section \ref{appendix:sparsityproof} in the Appendix for a proof.  This theorem assumes  ``hard sparsity''  (i.e., $\pr(X=0)>0$) for technical convenience.  See Appendix \ref{app:sparsitylepto} for a discussion about softer forms of sparsity.


%
%
%
%
%
%
%
%
%
%
%





\section{Gaining intuition for \texttt{vsp} with the population results} \label{sec:intuition}

\begin{quote}
This section studies each of the three steps in \texttt{vsp} by studying their population behavior.  
Statistical convergence around the population quantities is rigorously treated in Theorem \ref{thm:main} in Section \ref{sec:theoremstatement}.
\end{quote}

The semi-parametric factor model is a latent variable model with two sequential layers of randomness.  In the first layer of randomness, the latent variables $Z$ and $Y$ are generated.  In the second layer, the observed matrix $A$ is generated, conditionally on the latent variables.  
To parallel these two layers, there are two types of population results given in this section. 

The first two steps of \texttt{vsp} compute the principal components. 
Propositions \ref{prop:centering}  and  \ref{prop:svd} study these steps
applied to the
 population matrix 
\begin{equation}\label{eq:factorpopulation}
 \A = \E(A|Z, Y) = ZBY^T,
 \end{equation}
instead of $A$.    These propositions imply that the population principal components can be expressed as  $\widetilde ZR$, where $\widetilde Z \in \R^{n \times k}$ is $Z$ after column centering and  $R \in \R^{k \times k}$ is defined below.  If the $nk$ many random variables in $Z \in \R^{n \times k}$ are mutually independent, then $R$ converges to a rotation matrix.    These results allows for the randomness in $Z$ and $Y$, but they remove the second layer of randomness by using $\A$ instead of $A$. 
Then, Theorem \ref{thm:varimax1} studies the population version of the Varimax step.  To do this, take the expectation of the Varimax objective function, evaluated at the population principal components (i.e., $\widetilde ZR$), where the expectation is over the distribution of $Z$. This expectation removes the randomness in $Z$.
Under the identification assumptions for Varimax defined below, Theorem \ref{thm:varimax1} shows that the rotation that maximizes this function is $R^T \in \oo(k)$. So, rotating the population principal components with the population Varimax rotation yields the original factors, $(\widetilde ZR)R^T = \widetilde Z$.

Define $\bar Z \in \R^{n \times k}$ such that $\bar Z_{ij}$ equals the sample mean of the $j$th column of $Z$. Similarly for $\bar Y \in \R^{d \times k}$.  Define 
\begin{equation}\label{eq:tz}
\mbox{$\widetilde Z = Z- \bar Z \quad$ and $\quad \widetilde Y = Y- \bar Y$.}
\end{equation}

\begin{prop} \label{prop:centering} [Step 1 of \texttt{vsp}]
Centering $\A$ to get $\widetilde \A$ as in Equation \eqref{eq:centeredA}, has the effect of centering $Z$ and $Y$.
\[\widetilde \A = \widetilde Z B  \widetilde Y^T\]
This does not require any distributional assumptions on $Z$ or $Y$.  
\end{prop}  
A proof is given in Appendix \ref{app:popproofs}. The next proposition gives the SVD of $\widetilde \A = \widetilde Z  B \widetilde Y^T$.  
Define 
\[\widehat \Sigma_Z = \widetilde Z^T \widetilde Z / n, 
\quad \widehat \Sigma_Y = \widetilde Y^T \widetilde Y / d,\] 
and define  $ \widetilde R_U, \widetilde R_V \in \oo(k),$ and  diagonal matrix $\widetilde D$ to be the SVD of $\widehat \Sigma_Z^{1/2} B \widehat \Sigma_Y^{1/2} \in \R^{k \times k}$,
\[\widehat \Sigma_Z^{1/2} B \widehat \Sigma_Y^{1/2} = \widetilde R_U^T \widetilde D \widetilde R_V.\]
The next proposition shows that the rotation matrices $\widetilde R_U$ and  $\widetilde R_V$ convert the factor matrices $\widetilde Z$ and $\widetilde Y$ into the principal components and loadings $U$ and $V$.



\begin{prop} \label{prop:svd} [Step 2 of \texttt{vsp}]  Define the following matrices, 
\begin{equation} \label{eq:svdA}
U =  n^{-1/2} \widetilde{Z}\ \widehat \Sigma_Z^{-1/2} \widetilde R_U^T, \quad D = \sqrt{nd} \widetilde D, \quad V = d^{-1/2} \widetilde Y \widehat \Sigma_Y^{-1/2} \widetilde R_V^T.
\end{equation}
Then, $\widetilde \A = U D V^T$, where $U$ and $V$ contain the left and right singular vectors of $\widetilde \A$ and $D$ contains the singular values of $\widetilde \A$.  This does not require any distributional assumptions on $Z$ or $Y$.  
\end{prop}

The proof requires demonstrating the equality $\widetilde \A = U D V^T$ and showing that $U$ and $V$ have orthonormal columns. Substituting in the definitions reveals this result.  
Taken together, Propositions \ref{prop:centering} and \ref{prop:svd} show that  the first two steps of \texttt{vsp} on $\A$ compute $U \propto \widetilde{Z}\ \widehat \Sigma_Z^{-1/2} \widetilde R_U^T$; these are the principal components of $\A$.

\begin{remark} \label{remark:pca}[Relationship between PCA and the factors]
Proposition \ref{prop:svd} relates PCA on the population matrix $\A$ to the factors $Z$. 
This is because the population principal components are the columns of the matrix 
\begin{equation}\label{eq:pca}
U =  n^{-1/2} \widetilde{Z}\ \widehat \Sigma_Z^{-1/2} \widetilde R_U^T.
\end{equation}
So, the principal components are the centered latent factors $\widetilde Z$, ``whitened'' with $\widehat \Sigma_Z^{-1/2}$, and rotated by a $k \times k$ nuisance matrix $\widetilde R_U^T$.  
Despite the fact that PCA is typically considered a second order technique, this result implies that the principal components themselves do not retain any first or second order information about the latent factors, but retain all other distributional information. With Maxwell's Theorem, this suggests that higher order techniques such as Varimax hold the possibility of identifying the nuisance matrix.  In fact, Theorem \ref{thm:varimax1} below shows that Varimax can identify the nuisance matrix. 
\end{remark}

The Varimax problem applied to the population principal components $U$ in Equation \eqref{eq:pca} is 
\begin{equation}\label{eq:poppcavarimax}
\arg \max_{R \in \oo(k)} v(R, \widetilde{Z}\ \widehat \Sigma_Z^{-1/2} \widetilde R_U^T).
\end{equation}
Despite the fact that these are the population principal components, this is still a sample quantity because $Z$ is random.  This randomness is from the first stage of randomness in the semi-parametric factor model. 
The next theorem  gives a population result for the M-estimator in \eqref{eq:poppcavarimax} by studying the expected value of $v$ over $Z$, to show that it can identify $\widetilde R_U$.
Assumption \ref{assumption:Varimax} gives the identification assumptions on the distribution of $Z$ that will be used in both the population result for Varimax (Theorem \ref{thm:varimax1}) and the main theorem (Theorem \ref{thm:main}).

 \begin{assumption} \label{assumption:Varimax} [The identification assumptions for Varimax]
The matrix $Z \in \R^{n \times k}$ satisfies the identification assumptions for Varimax if all of the following conditions hold on the rows  $Z_i \in \R^k$ for $i=1, \dots, n$:
\begin{enumerate}
\item[i)] the vectors $Z_1, Z_2, \dots, Z_n$ are iid, 
\item[ii)] each vector $Z_i \in \R^k$ is composed of $k$ independent random variables (not necessarily identically distributed), 
\item[iii)] $Var(Z_{ij}) = 1$ for all $j$,\footnote{The third assumption in Varimax is not restrictive because the matrix $B$ can absorb a rescaling of the variables.  That is, let $Z^{rescaled} \in \R^{n \times k}$ satisfy the first two conditions and presume that $\A = Z^{rescaled} B^{rescaled} Y^T$.  Define $\Sigma_Z = Cov(Z^{rescaled}_i)$,  $Z = Z^{rescaled} \Sigma^{-1/2}$, and $B = \Sigma^{1/2}B^{rescaled}$.  Because $Z^{rescaled}$ satisfies the second condition, $\Sigma_Z$ is diagonal.   So, $Z = Z^{rescaled} \Sigma^{-1/2}$ retains independent components and now satisfies the third condition.  Moreover, $\A = Z B Y^T$.} and 
\item[iv)] the elements of $Z_i$ are leptokurtic.
\end{enumerate}
\end{assumption}

%
%
%
%
%
%
 Let $\widetilde Z_1$ be the first row of $\widetilde Z$.  Define $Z^o = Z_1 -\E(Z_1) \in \R^k$.  
Theorem \ref{thm:varimax1} shows that  the rotation matrix $R$ that maximizes the expected Varimax objective function, $\E(v(R, Z^o \widetilde R_U^T))$, is $\widetilde R_U$.  In this formulation, several quantities from the sample maximization problem \eqref{eq:poppcavarimax} have been replaced.  First, 
 the sample objective function $v$ in Equation \eqref{eq:Varimax} has been replaced with its expectation over the distribution of $Z$.  Then,  $\bar Z$ has been replaced by $\E(Z_1)$ and $\Sigma_Z^{-1/2}$ has been replaced with its limiting quantity under Assumption \ref{assumption:Varimax} (i.e., the identity matrix).


Because the Varimax objective function does not change if the estimated factors are reordered or if some of the estimated factors have a sign change, the maximizer of Varimax is actually an equivalence class that allows for these operations.  Define the set
\begin{equation} \label{eq:pp}
\pp(k) = \{ P \in \oo(k): P_{ij} \in \{-1,0,1\}\}.
\end{equation}
It is the full set of matrices that allow for column reordering and sign changes.  

  \begin{theorem}\label{thm:varimax1}[step 3]
Suppose that $Z \in \R^{n \times k}$ satisfies the identification assumptions for Varimax (Assumption \ref{assumption:Varimax}).  Let $Z_1 \in\R^k$ be the first row of $Z$.  
Define $Z^o = Z_1 - \E(Z_1)$. 
For any nuisance rotation matrix $\tilde R \in \oo(k)$, 
%
\begin{equation}\label{eq:popvarimax}
\arg \max_{R \in \oo(k)} \E(v(R, Z^o \tilde R^T)) = \{\widetilde R P: P\in \pp(k)\}
\end{equation}
\end{theorem}
%
%
%
%
The output step of \texttt{vsp}  right multiplies the principal components $\sqrt{n} U \approx \widetilde Z \widetilde R_U^T$ with a matrix which maximizes Varimax.  In the population results, this matrix is $\widetilde R_UP$.  Thus, the Varimax rotation reveals the unrotated factors, $(\widetilde Z \widetilde R_U^T)\widetilde R_U P = \widetilde ZP$. 

Remark \ref{remark:uncentering} describes a method to recenter the factors $\widetilde Z$ to get $Z$. Section \ref{app:recentering} in the appendix gives a population justification for this recentering step.

\begin{remark} \label{remark:centering} [The role of centering]
A version of Proposition \ref{prop:svd} still holds for the SVD of $\A$ (without centering) by replacing 
$\widehat \Sigma_Z$ with $Z^T Z / n$ and replacing $\widehat \Sigma_Y$ with $Y^T Y/d$ in Equation \eqref{eq:svdA}.   Even if the elements of the matrix $Z$ are independent and have unit variance, then the columns of $Z$ will not be asymptotically orthogonal (unless $\E(Z) = 0$).  As such, right multiplying 
$U = Z (Z^T Z / n)^{-1/2} \widetilde R_U^T$ with an orthogonal rotation (i.e., the one estimated by Varimax) cannot reveal $Z$.
This highlights the role of centering in \texttt{vsp}; centering $\A$ has the effect of centering the latent variables, which in turn makes the latent factors asymptotically orthogonal under the assumption of independence.  This allows Varimax to unmix them with an orthogonal matrix. 
\end{remark}

\begin{remark} \label{remark:pcavsfactor}
PCA is not the standard approach in Vintage Factor Analysis.
To see why, define $\A = \E(A|Z,Y) = ZBY^T$ and notice that the diagonal elements of $n^{-1}\A\A^T$ 
are less than or equal to the diagonal elements of the expected sample covariance matrix $n^{-1} \E(A^TA|Z,Y)$. 
PCA does not adjust for this excess along the diagonal of the sample covariance matrix and this makes PCA biased. 
However, as Theorem \ref{thm:main} shows in the next section, the estimates from PCA with a Varimax rotation converge to the desired quanties.  Thus, in the asymptote studied herein, PCA with a Varimax rotation is asymptotically unbiased.  It is possible that a different approach could have increased statistical efficiency, but this is not studied in this paper.  
\end{remark}

%
%

\section{The main theorem}\label{sec:theoremstatement}


\begin{quote}
Theorem \ref{thm:main} is the main result for this paper.  This theorem does not presume a parametric form for the random variables in $Z$ or $A$.  Instead, it uses the identifying assumptions for Varimax (Assumption \ref{assumption:Varimax}) and two further assumptions on the tails of these distributions. 
\end{quote}

 Recall that $\widehat \mu_Z$   estimates the column means of $Z$ defined in Remark \ref{remark:uncentering}.  
 Let $\widehat Z_i$ be the $i$th row of $\widehat Z$.  Theorem \ref{thm:main}  shows that for \textit{every} $i \in 1, \dots, n$, $\widehat Z_i +\widehat \mu_Z$ converges to $Z_i$ (after allowing for a permutation and sign flip).



\begin{assumption}\label{assumption:ztail}  
Each column of $Z$ and $Y$ is generated from a distribution that does not change asymptotically and has a moment generating function in some fixed $\epsilon >0$ neighborhood around zero. 
\end{assumption}
Let $\A$ be defined in Equation \eqref{eq:factorpopulation}. Define the mean and maximum of $\A$ as
\begin{equation} \label{eq:rhodef}
\rho_n = \frac{1}{nd} \sum_{i,j} \A_{ij} \quad \mbox{ and }  \quad \bar \rho_{n} = \underset{i,j}{\max}{| \A_{ij}|}. 
\end{equation}
Theorem \ref{thm:main} allows for $A$  to contain mostly zeros 
by assuming that as $n$ and $d$ grow,  $B_n = \rho_n B$ for some fixed matrix $B \in \R^{k \times k}$.  If $\rho_n \rightarrow 0$, then $A$ is sparse.   This is analogous to the asymptotics in \cite{bickel2009nonparametric} for the Stochastic Blockmodel.  

\begin{assumption}\label{assumption:atail} For any valid subscripts $i$ and $j$, eventually in $n$, $$ \mathbb{E} [(A_{ij} - \A_{ij})^m ] \leq \max\{(m-1)!(\bar{\rho}_n)^{m/2},\bar{\rho}_n\} ,  \ \mbox{ for all } m \geq 2, $$ where this expectation is conditional on $Z, Y$.
%
 \end{assumption} 
 
 Assumption \ref{assumption:atail} controls the tail behavior of the random variables in the elements of $A$.  This assumption is more inclusive than sub-Gaussian. For example, this assumption is satisfied when $A$ contains Poisson random variables, as happens in Latent Dirichlet Allocation in Section \ref{sec:lda}.  This assumption is also satisfied if $A$ contains Bernoulli random variables, as happens in Stochastic Blockmodeling.  See Sections \ref{sec:assumption3bern} and \ref{sec:assumption3pois} in the Appendix for further discussion.

The quantity 
\[\Delta_n = n \rho_n\]
controls the asymptotic rate in Theorem \ref{thm:main}.  So, it is helpful to have some sense for it.  For example, suppose that (i) $A$ contains Bernoulli elements, (ii) each row and column sum of $\A$ grows at a similar rate,  (iii) $n \asymp d$, and (iv) $\rho_n\rightarrow 0$, then $\Delta_n$ is roughly the expected number of ones in each row and column of $A$. 


\begin{theorem}\label{thm:main}  
Suppose that $A \in \R^{n \times d}$ is generated from a semi-parametric factor model that satisfies Assumptions \ref{assumption:Varimax}, \ref{assumption:ztail}, and \ref{assumption:atail}. Presume that asymptotically,  $\A = \rho_n ZBY^T$ for some fixed and full rank matrix $B$.  In the asymptotic regime where $n \asymp d$  and $  \Delta_{n} \succeq \log^{ 11.1 } n$, 
\begin{equation}\label{Zconvergence}
||(\widehat{Z}+\1_n\widehat{\mu}_{Z}) - ZP_n||_{2\to \infty}  = O_{p}(\Delta_n^{-.24}\log^{2.75}n),
\end{equation}
where  $\widehat{Z}$ is the estimate produced by \texttt{vsp} (with step 1) applied to $A$ and $\widehat{\mu}_{Z}$ is the estimate defined in Equation \eqref{eq:muhat}.
\end{theorem}

Theorem \ref{thm:main}  shows convergence in $2 \rightarrow \infty$ norm.  This means that every row of $\widehat{Z}+\1_n\widehat{\mu}_{Z}$ converges to the corresponding row of $Z$ in $\ell_2$. The $P_n$ matrix accounts for the fact that we do not attempt to identify the order of the columns in $Z$, or their sign.  If $\widehat{Z}$ is used without recentering by $\1_n\widehat{\mu}_{Z}$, then a similar result holds for estimating $\widetilde Z$.
By symmetry, if $Y$ satisfies the identification assumptions for Varimax, then \texttt{vsp} can also estimate $Y$.  If both $Z$ and $Y$ satisfy the identification assumptions for Varimax, then $B$ can also be recovered, even when it is not diagonal.  The proof for Theorem \ref{thm:main} begins in Appendix \ref{appendix:main}.

\section{Modern factor models as semi-parametric factor models}\label{sec:examples}

\begin{quote}
The semi-parametric factor model is related to ICA, the Stochastic Blockmodel, and Latent Dirichlet Allocation.  
Corollaries \ref{corollary:dcsbm} and \ref{corollary:pois} show that with some slight variations on the preprocessing of $A$, \texttt{vsp} can estimate the Stochastic Blockmodel and Latent Dirichlet Allocation.
\end{quote}

\subsection{Relationship to Independent Components Analysis}\label{sec:ica}


Independent Components Analysis (ICA) uses a type of semi-parametric factor model that is motivated by blind-source separation in signal processing.  In the typical formulation of ICA, we observe a multivariate time series $\A_t = Z_t M \in \R^k$ for $t = 1, \dots, n$, where $Z_t \in \R^k$ contains independent and non-Gaussian random variables.  The aim is to estimate $M^{-1}$, to unmix the observed signals in $\A_t$, and reveal the independent components  $Z_t$.  There are multiple ICA results that share some similarities to Theorems \ref{thm:varimax1} and \ref{thm:main}  (e.g.  \cite{comon1994independent, ica, chen2005consistent,chen2006efficient, ieeeICA, miettinen2015fourth, samworth2012independent}).  
To see the connection to the current paper, let $M \in \R^{k \times d}$ be potentially rectangular and defined as $M = BY^T$.  
To enable the regime $d=k$, the results for ICA typically presume that $\A = ZM$ is observed with little or no noise. 
In contrast, Theorem \ref{thm:main} covers situations where (i) $d$ grows at the same rate as $n$, (ii) there is an abundance of noise in $A$,  and (iii) $A$ is mostly zeros (i.e., sparse).  This allows the theorem to cover  the contemporary factor models in Section \ref{sec:examples}.  

\subsection{Tensor decompositions}\label{sec:ica}
Motivated in part by the issue of rotational invariance of PCA, \cite{kruskal1977three} showed how a tensor decomposition called the CP decomposition is unique; it decomposes a tensor into a set of factors that are not rotationally invariant. In Section 4, Kruskal discusses how this three way decomposition does not suffer from the same problem of rotation that ``consumes considerable attention and effort'' in factor analysis.  In an elegant formulation, \cite{anandkumar2014tensor} showed how these tensor spectral methods could be applied to  estimate the latent factors in a model class similar to the semi-parametric factor model. 
Where the principal components of $A$ are the eigenvectors of a matrix that contains the second order moments, $n^{-1} \E(\tilde A^T \tilde A)_{uv} = \E( \tilde A_{iu } \tilde A_{iv})$,
the elements of this higher order tensor contain the third order (or higher) moments; for example,  $T \in \R^{d \times d \times d}$ with $T_{u,v,w} = \E(A_{iu} A_{iv} A_{iw})$.  Then, for various formulations of $T$ and latent variable models, the CP tensor decomposition of $T$ has components that are equal to the latent factors \citep{janzamin2019spectral}.



The issue of rotational invariance  motivates for the extension from matrices to tensors. For example, in a recent book on using tensors for latent variable modeling, \cite{janzamin2019spectral} writes in the abstract ``PCA and other spectral techniques applied to matrices have several limitations. By limiting to only pairwise moments, they are effectively making a Gaussian approximation on the underlying data.'' 
However, despite the fact that PCA is typically imagined as a second order technique,
the principal components of $A$ retain the higher-order distributional properties of the latent variables (see Remark \ref{remark:pca}). 
As such, we need not consider the higher order moments of the manifest variables $A$ in the tensor $T$.  
\texttt{vsp} uses the higher order moments of the principal components themselves, by applying Varimax directly to the principal components. 
Given our heuristics around rotational invariance, it is surprising that this can work.

\subsection{Stochastic Blockmodels}\label{sec:block}

In social network analysis, $A \in \{0,1\}^{n \times n}$ is the adjacency matrix of a graph on $n$ people. 
\[A_{ij}  =  \left\{\begin{array}{cl}1 & \mbox{$i$ friends with $j$} \\0 & o.w.\end{array}\right.\]
The Stochastic Blockmodel \citep{holland} is a semi-parametric factor model for generating a random adjacency matrix. Under this model,  each individual $i$ is assigned to a single block $z(i) \in \{1, \dots, k\}$ and the probability that $i$ and $j$ are friends is 
\[\pr(A_{ij} = 1 | z(i),z(j)) = B_{z(i), z(j)}, \mbox{ where } B \in [0,1]^{k \times k}.\]  
Define $\A = E(A| Z, B, Y)$. To express $\A$ in the factor model as $ZBZ^T$, define $Z \in \{0,1\}^{n \times k}$ such that $Z_{ij}= 1$ when $z(i) = j$ and $Z_{ij} = 0$ otherwise.  When friendships are symmetric, so is $A$; in this setting $Y=Z$ and the elements above the diagonal of $A$ are independent. 
There are four popular generalizations of the Stochastic Blockmodel that have the structure $ZBZ^T$, and are thus other types of semi-parametric factor models.  The Degree-Corrected Stochastic Blockmodel includes an additional degree parameters $\theta_{i, z(i)}>0$ for each individual $i$.  The probability of friendship becomes $\A_{ij} = \theta_{i, z(i)} \theta_{j, z(j)} B_{z(i), z(j)}$ \citep{karrer2011stochastic}.  To express this model as $ZBZ^T$,
define $Z_{ij} = \theta_{i, z(i)} \mathbb{I}\{z(i)=j\}$, where $\mathbb{I} \in \{0,1\}$ is the indicator function.
In the Overlapping Stochastic Blockmodel, $Z \in \{0,1\}^{n \times k}$  is sparse \citep{latouche2011overlapping}.\footnote{The original paper paper on the overlapping Stochastic Blockmodel is not exactly the factor model used here because it includes a logistic link function, $\pr(A_{ij} = 1) = logit(Z_i B Z_j^T)$.}
In the mixed-membership Stochastic Blockmodel, each row of $Z$ is an independent sample from the Dirichlet distribution \citep{airoldi2008mixed}.  Later, \cite{zhang2014detecting} and \cite{jin2017sharp} generalized these models to only presume that $Z_i \in \R^k$ is element-wise non-negative.  
Table \ref{tab:sbm} summarizes all of these models.
While this discussion focuses on unipartite and undirected graphs, graphs that are ``two-way,'' ``bipartite,''  or ``directed,'' can also be modeled in the form $\A = ZBY^T$ \citep{pnas}.

\begin{table}[htbp]
   \centering
\begin{tabular}{|l|l|l|}
\hline SBM & the vector $Z_i \in \R^k$  contains & distribution of $Z_i$\\
\hline 0) Standard SBM &  a single one, the rest zeros  & multinomial\\
1) Degree-Corrected &  a single positive entry, the rest zeros & not specified\\
2) Overlapping  &   a mix of $1$s and $0$s & independent Bernoulli\\
3) Mixed Membership  & non-negative entries that sum to one & Dirichlet \\
4) Degree-Corrected, & & \\
\ \ \  \ \ Mixed Membership & non-negative entries &  not specified\\
\hline \end{tabular}
   \caption{Restrictions on the factor matrix $Z$ create variations on the Stochastic Blockmodel (SBM).  
   There are further differences between these models that are not emphasized by this table.}
   \label{tab:sbm}
\end{table}

\vspace{.1in}
\noindent
\textbf{Estimating the Degree-Corrected Stochastic Blockmodel with \texttt{vsp}.}
Under the Stochastic Blockmodel and the Degree Corrected version, each node $i$ belongs to exactly one cluster.  In such ``hard clustering'' models, the elements in the same row of $Z$ cannot be independent.  This implies that $Z$ cannot satisfy Assumption \ref{assumption:Varimax} of Theorem \ref{thm:main}. 
The next corollary shows that \texttt{vsp} \textit{without the centering step} can estimate these models. 

%
%


Let $\pi \in \R^k$ be a probability  distribution on $[k]$.  Suppose that $z(1), \dots, z(n) \sim Multinomial(\pi)$, independently.  For each block $j$, suppose that $\theta_{1,j}, \dots, \theta_{n,j} \in \R$ are independent random variables generated from a bounded probability distribution $f_j$.  The scale of this distribution is unidentifiable; so for technical convenience, it is presumed that $\E(Z_{ij}^2) = 1$, or equivalently, that $\E(\theta_{i, j}^2) = 1/\pi_j$.  This is akin to the third assumption in the Varimax assumption.
This scaling ensures that $\E(Z^T Z)/n$ (i.e. without centering) converges to the identity matrix.   If each $f_j$ is a point mass, then this model is equivalent to the SBM. 

\begin{corollary} \label{corollary:dcsbm}
Suppose that $A_n \in \R^{n \times n}$ is generated from the Degree Corrected Stochastic Blockmodel with $\E(A_n|Z_n) = Z_n B_n Z_n$, where $Z_n$ is generated as described in the proceeding paragraph. Suppose that the probability distributions $f_j$ for $j \in [k]$ are bounded.  Define $\rho_n$ as in Equation \eqref{eq:rhodef} and suppose that there exists a fixed matrix $B \in \R^{k \times k}$ such that $B_n = \rho_n B$. 

Define $\widehat Z \in \R^{n \times d}$ as the output of \texttt{vsp} without centering (i.e. skip step 1).  In the asymptotic regime where $\Delta_n \succeq \log^{11.1} n$, 
there exists a sequence permutation and sign-flip matrices $P_n \in \mathcal{P}(k)$ such that 
\begin{equation}\label{Zconvergence}
||\widehat{Z} - ZP_n||_{2\to \infty}  = O_{p}(\Delta_n^{-.24}\log^{2.75}n).
\end{equation}

%
%
\end{corollary}
A proof is contained in Appendix \ref{sec:dcsbm}.

To see why the centering step creates bias for \texttt{vsp} under a hard clustering model, note that \texttt{vsp} with the centering step (step 1) estimates $\widetilde Z$ (i.e., $Z$ after centering). By construction, $\widehat Z$ contains orthogonal columns.  However, under the Stochastic Blockmodel, $\widetilde Z$ does not.  Interestingly, $Z$ \textit{without centering}  does contain orthogonal columns and \texttt{vsp} without centering can estimate it.


%
%
%

\vspace{.1in}
\noindent
\textbf{Overlapping and Mixed Membership.}  Under the Overlapping SBM, $$Z_{ij} \sim Bernoulli(p_j)$$ independently for all $i$ and $j$. This will satisfy the identification assumptions for Varimax so long as $p_j \not \in [1/2 \pm 1/\sqrt{12}]$ for $j = 1, \dots, k$.  This rather strange condition ensures that $Z_{ij}$ is leptokurtic and thus Varimax can identity the rotation. If Varimax were replaced with an alternative rotation from the ICA literature, then one could remove the awkward condition on the $p_j$'s. 


Under the Mixed Membership SBM, $Z_i$ is on the simplex.  As such, its elements must sum to one and cannot be statistically independent.  This restriction to the simplex also limits the ability of the Mixed Membership model to create a large amount of degree heterogeneity, a common property in empirical networks.  As discussed in Section \ref{sec:lda}, this  problem also arrises for  Latent Dirichlet Allocation (LDA).  Section \ref{sec:lda} discusses a natural generalization of LDA that allows for more heterogeneous document lengths.  A similar generalization could be applied to the Mixed Membership SBM.  This would create a ``Degree-Corrected Mixed Membership model.''  Under such a model, a result analogous to Corollary \ref{corollary:pois} could be derived.


%


\vspace{.1in}
\noindent
\textbf{Degree-Corrected Mixed Membership.}  
The  papers which proposed the Degree-Corrected Mixed Membership model only presume that $Z_i$ is element-wise non-negative 
\citep{zhang2014detecting, jin2017sharp}.  As such, if the elements of $Z_i$ are sampled in a way which satisfy the identification assumptions for Varimax, then Theorem \ref{thm:main} shows that \texttt{vsp} can estimate this model.

\subsection{Latent Dirichlet Allocation}\label{sec:lda}
In the setting of text analysis and natural language processing,  let $A \in \N^{n \times d}$ be a document-term matrix on $n$ documents and $d$ unique words,
\begin{equation} \label{eq:docterm}
A_{ij} = \mbox{number of times that word $j$ appears in document $i$}.
\end{equation}
Latent Dirichlet Allocation (LDA) is a popular generative model for $A$ that is used for modeling the topics of documents \citep{blei}.

The LDA model has parameters $\xi>0$, $\alpha \in \R_+^k$, and $\beta \in \R_+^{d \times k}$ with $\1_d^T \beta = \1_k$.  The rows of $\beta$ index the unique words $1, \dots, d$. Because the elements of $\beta$ are positive and each column sums to one, each column makes a probability distribution on the unique words.
 LDA generates a single document $i = 1, \dots, n$ with the following steps, (1) choose $Z_i \sim Dirichlet(\alpha)$ to be the topic distribution for that document, (2) sample $N_i \sim Poisson(\xi)$ to be the number of words in the document, (3) for each of the words in the document $w = 1, \dots, N_i$, 
choose the topic for that word $z_w \sim Multinomial(Z_i) \in \{1, \dots, k\}$, and then sample the word $w$ as multinomial with probabilities specified by the $z_w$ column of $\beta$ (i.e., $w$ is the $j$th unique word with probability $\beta_{j,z_w}$).

\begin{lemma} \label{lemma:ldaFactor}
Under the LDA model, conditionally on the Dirichlet variables $Z_1, \dots, Z_n$, the document-term matrix $A$ has independent Poisson entries with 
\begin{equation}\label{eq:ldamodel}
\E(A|Z)   = \xi Z \beta^T,
\end{equation}
where $Z \in \R_+^{n \times k}$ has rows $Z_1, \dots, Z_n$.
\end{lemma} 

A short proof in Section \ref{sec:LDAlemmas} relies upon the Poisson-Multinomial relationship.  While Equation \eqref{eq:ldamodel} has the form of the semi-parametric factor model (e.g. set $B=I$ and $Y= \beta$), it does not satisfy the identification assumptions for Varimax because the elements in $Z_i$ sum to one and as such, they must be dependent. Moreover, this has the unnatural consequence of making $\E(A|Z)$ have rank $k-1$ or less.  However, the following modification makes $\E(A|Z)$ have rank $k$ and enables the application of Theorem \ref{thm:main}. 

In the original formulation of LDA, the number of words in document $i$ is $N_i \sim Poisson(\xi)$, for $\xi \in \R_+$. About this step, \cite{blei} says, ``more realistic document length distributions can be used as needed.'' 
If document lengths are more heterogenous than what is modeled by Poisson($\xi$), then a convenient way to increase the heterogeneity is to use Poisson overdispersion; first sampling $\xi_i$, then sampling $N_i \sim Poisson(\xi_i)$. 

\begin{quote}
\noindent \textbf{Natural modification to LDA:} Sample $N_i$, the number of words in  document $i$,  as overdispersed Poisson via (1)  $\xi_i \sim Gamma(\small{\sum}_i \alpha_i, s)$   for some scale parameter $s>0$ and (2) $N_i \sim Poisson(\xi_i)$. 
\end{quote}

This ``Gamma-Poisson mixture'' is a well studied model of Poisson overdispersion; under this model, $N_i$ has the negative binomial distribution.  Define $\Xi \in \R^{n \times n}$ as a diagonal matrix with $\Xi_{ii} = \xi_i$.  


\begin{lemma} \label{lemma:naturalLDA}
Under the LDA model with the natural modification to $N_i$, conditionally on $Z_1, \dots, Z_n$ and $\Xi$, the document-term matrix $A$ has independent Poisson entries satisfying 
\begin{equation}\label{eq:ldamodel}
\E(A|\Xi, Z) = (\Xi Z) \beta^T.
\end{equation}
Moreover, each element $(\Xi Z)_{ij}$ is independent Gamma$(\alpha_j, s)$ and this distribution is leptokurtic. 
Define $\Sigma$ as a diagonal matrix with $\Sigma_{jj} = \alpha_j s^2 $, the variance of Gamma$(\alpha_j,s)$.  Then, the factor matrix 
\begin{equation}\label{eq:zstar}
Z_* = (\Xi Z)\Sigma^{-1/2}
\end{equation}
satisfies the identification assumptions for Varimax.
\end{lemma}
See Section \ref{sec:LDAlemmas} for a short proof.  The next result shows that \texttt{vsp} applied to the column centered version of $A$ (i.e., $\breve A = A - \1_n (\1_n^T A/n)$) can estimate the LDA model with the natural modification.  
Similar to $\breve A$, define $\breve \A$ be the column centered version of $\A$.


%

\begin{corollary} \label{corollary:pois}
Let $A$ be generated from the natural modification to LDA given above with $k$ topics and let  ${\A} = \E(A|\Xi, Z)$. 
Define $Z_{*}$ as in Equation \eqref{eq:zstar}.
Let $\widehat{Z} $ be the output of \texttt{vsp} using $\breve A$ as input (and skipping step 1).  
In the asymptotic regime where  
$$ \Delta_n \succeq \log^{15.1} n, \quad \sigma_{\min}(\beta) \geq c_1, $$
for universal constant $c_1 \in (0,1)$, almost surely there exists $P_n \in \mathcal{P}(k)$ s.t. 
\begin{equation} \label{eq:docmember}
||\widehat{Z} - (Z_{*}-\E(Z_{*}))P_n||_{2\to \infty}  = O_{p}(\Delta_n^{-.24}\log^{2.75}n).
\end{equation}
Define the matrix $\Phi= \widehat{Z}^{T}\breve{A} \in \mathbb{R}^{k\times d}$ and estimate $\widehat{\beta} = (\Lambda_b^{-1} \Phi)^T\in \mathbb{R}^{d\times k}$, where $\Lambda_b$ is a diagonal matrix with $i$th diagonal element equals to $\ell_1$-norm of $i$th row of $\Phi$. Under this construction,
\begin{equation}
||\widehat{\beta}^T - P_n^T\beta^T||_{\infty} =O_p(\Delta_n^{-.24}\log^{3.75}n).\end{equation}
\end{corollary}

The elements of $Z_*$ are independent Gamma random variables that have been rescaled by the diagonal matrix $\Sigma^{-1/2}$ to ensure that they have unit variance. Corollary \ref{corollary:pois} shows that \texttt{vsp} using the column-centered matrix $\breve A$ estimates $Z_* - \E(Z_*)$; similar to the previous results, this $2 \rightarrow \infty$ convergence implies that each row of $\widehat Z$ converges to the corresponding row of $Z_* - \E(Z_*)$. Using $\widehat Z$, the corollary constructs $\hat \beta \in \R^{d \times k}$, a simple estimator for the probability distribution of   words within each of the $k$ topics.  Each of the $k$ estimated topic distributions converges in $\ell_1$ norm just a little slower than $\Delta_n^{-1/4}$.  A proof of Corollary \ref{corollary:pois} is given in Section \ref{sec:LDAproof}.

\section{Discussion}\label{sec:discussion}

%
%
%
%
%
%
%
%


PCA with Varimax is a vintage data analysis technique. Theorem \ref{thm:main} shows that it provides a unified spectral estimation strategy for a broad class of semi-parametric factor models.  The reason for this is that (1) the principal components have the same column space as the latent factors and (2) under the identification assumptions for Varimax, Varimax specifies a basis for that column space in which each basis vector corresponds to a latent factor; this is the intuition gained in Section \ref{sec:intuition} and Theorem \ref{thm:main}.  Leptokurtosis is a key identifiability condition in the identification assumptions for Varimax. This condition is satisfied if the factors are sparse.  Moreover, this condition can be examined in the data.  In fact, Section \ref{sec:diagnostics} reinterprets the diagnostics practices developed in \cite{thurstone1935vectors, thurstone1947} as examining that leptokurtic condition.  Taken together, the results in this paper show that the Vintage Factor Analysis know-how developed by Thurstone and Kaiser performs statistical inference. This know-how has survived for nearly a century, despite the conventional wisdom that the factor rotation cannot  perform statistical inference.

\begin{quote}
\textit{Things don't necessarily happen for a reason; but things survive for a reason.}\newline \vspace{.1in}\hspace{3.4in}Nassim Nicholas Taleb.  
\end{quote}

\vspace{.3in}
\textbf{Acknowledgements:} Thank you to De Huang for valuable discussions. Thank you to Joshua Cape for helpful comments on an early draft on this paper.   Thank you to Alex Hayes for help  creating an \texttt{R} package of the code.  Thank you to E Auden Krauska, Dan Bolt, Alex Hayes, Fan Chen, Stephen Stigler, Anru Zhang, Miaoyan Wang, and Sebastien Roch for helpful discussions during the course of this research.  This research is supported in part by NSF Grants DMS-1612456 and DMS-1916378 and ARO Grant W911NF-15-1-0423.

\bibliographystyle{unsrtnat}
\bibliography{refs.bib}

\begin{thebibliography}{57}
\providecommand{\natexlab}[1]{#1}
\providecommand{\url}[1]{\texttt{#1}}
\expandafter\ifx\csname urlstyle\endcsname\relax
  \providecommand{\doi}[1]{doi: #1}\else
  \providecommand{\doi}{doi: \begingroup \urlstyle{rm}\Url}\fi

\bibitem[Thurstone(1935)]{thurstone1935vectors}
Louis~Leon Thurstone.
\newblock \emph{The vectors of mind: Multiple-factor analysis for the isolation
  of primary traits.}
\newblock University of Chicago Press, 1935.

\bibitem[Kaiser(1958)]{kaiser}
Henry~F Kaiser.
\newblock The varimax criterion for analytic rotation in factor analysis.
\newblock \emph{Psychometrika}, 23\penalty0 (3):\penalty0 187--200, 1958.

\bibitem[Thurstone(1947)]{thurstone1947}
Louis~Leon Thurstone.
\newblock \emph{Multiple factor analysis.}
\newblock University of Chicago Press: Chicago, 1947.

\bibitem[Anderson and Rubin(1956)]{anderson1956statistical}
Theodore~W Anderson and Herman Rubin.
\newblock Statistical inference in factor analysis.
\newblock In \emph{Proceedings of the third Berkeley symposium on mathematical
  statistics and probability}, volume~5, pages 111--150, 1956.

\bibitem[Jolliffe(2002)]{jolliffe2002principal}
I.T. Jolliffe.
\newblock \emph{Principal Component Analysis}.
\newblock Springer Series in Statistics. Springer, 2002.
\newblock ISBN 9780387954424.

\bibitem[Shalizi(2009)]{cosma}
Cosma Shalizi.
\newblock Lecture notes on factor analysis, September 2009.
\newblock URL
  \url{http://www.stat.cmu.edu/~cshalizi/350/lectures/12/lecture-12.pdf}.

\bibitem[Ramsay and Silverman(2007)]{ramsay2007applied}
J~O Ramsay and B~W Silverman.
\newblock \emph{Applied functional data analysis: methods and case studies}.
\newblock Springer, 2007.

\bibitem[Johnson and Wichern(2007)]{johnson2007applied}
R.A. Johnson and D.W. Wichern.
\newblock \emph{Applied Multivariate Statistical Analysis}.
\newblock Pearson Education International. Pearson Prentice Hall, 2007.
\newblock ISBN 9780135143506.

\bibitem[Bartholomew et~al.(2011)Bartholomew, Knott, and
  Moustaki]{bartholomew2011latent}
D.J. Bartholomew, M.~Knott, and I.~Moustaki.
\newblock \emph{Latent Variable Models and Factor Analysis: A Unified
  Approach}.
\newblock Wiley Series in Probability and Statistics. Wiley, 2011.
\newblock ISBN 9780470971925.

\bibitem[Ramsay and Silverman(2005)]{fda}
J.~Ramsay and B.W. Silverman.
\newblock \emph{Functional Data Analysis}.
\newblock Springer Series in Statistics. Springer, 2005.
\newblock ISBN 9780387400808.

\bibitem[Maxwell(1860)]{maxwell}
James~Clerk Maxwell.
\newblock V. illustrations of the dynamical theory of gases. part i. on the
  motions and collisions of perfectly elastic spheres.
\newblock \emph{The London, Edinburgh, and Dublin Philosophical Magazine and
  Journal of Science}, 19\penalty0 (124):\penalty0 19--32, 1860.

\bibitem[Feller(1971)]{feller}
W.~Feller.
\newblock \emph{An Introduction to Probability Theory and its Applications,
  Volume 2}.
\newblock John Wiley and Sons, Inc., 1971.

\bibitem[Hyv{\"a}rinen et~al.(2004)Hyv{\"a}rinen, Karhunen, and Oja]{ica}
Aapo Hyv{\"a}rinen, Juha Karhunen, and Erkki Oja.
\newblock \emph{Independent component analysis}, volume~46.
\newblock John Wiley \& Sons, 2004.

\bibitem[Anandkumar et~al.(2014)Anandkumar, Ge, Hsu, Kakade, and
  Telgarsky]{anandkumar2014tensor}
Animashree Anandkumar, Rong Ge, Daniel Hsu, Sham~M Kakade, and Matus Telgarsky.
\newblock Tensor decompositions for learning latent variable models.
\newblock \emph{The Journal of Machine Learning Research}, 15\penalty0
  (1):\penalty0 2773--2832, 2014.

\bibitem[Dua and Graff(2017)]{nyt}
Dheeru Dua and Casey Graff.
\newblock {UCI} machine learning repository, 2017.
\newblock URL \url{http://archive.ics.uci.edu/ml}.

\bibitem[d'Aspremont et~al.(2005)d'Aspremont, Ghaoui, Jordan, and
  Lanckriet]{d2005direct}
Alexandre d'Aspremont, Laurent~E Ghaoui, Michael~I Jordan, and Gert~R
  Lanckriet.
\newblock A direct formulation for sparse pca using semidefinite programming.
\newblock In \emph{Advances in neural information processing systems}, pages
  41--48, 2005.

\bibitem[Vu and Lei(2013)]{vince}
Vincent~Q Vu and Jing Lei.
\newblock Minimax sparse principal subspace estimation in high dimensions.
\newblock \emph{The Annals of Statistics}, 41\penalty0 (6):\penalty0
  2905--2947, 2013.

\bibitem[Fiori and Zenga(2009)]{fiori2009karl}
Anna~M Fiori and Michele Zenga.
\newblock Karl pearson and the origin of kurtosis.
\newblock \emph{International Statistical Review}, 77\penalty0 (1):\penalty0
  40--50, 2009.

\bibitem[Erd{\H{o}}s et~al.(2013)Erd{\H{o}}s, Knowles, Yau, and
  Yin]{erdHos2013spectral}
L{\'a}szl{\'o} Erd{\H{o}}s, Antti Knowles, Horng-Tzer Yau, and Jun Yin.
\newblock Spectral statistics of erd{\H{o}}s--r{\'e}nyi graphs i: local
  semicircle law.
\newblock \emph{The Annals of Probability}, 41\penalty0 (3B):\penalty0
  2279--2375, 2013.

\bibitem[Cape et~al.(2019{\natexlab{a}})Cape, Tang, and Priebe]{cape2019signal}
Joshua Cape, Minh Tang, and Carey~E Priebe.
\newblock Signal-plus-noise matrix models: eigenvector deviations and
  fluctuations.
\newblock \emph{Biometrika}, 106\penalty0 (1):\penalty0 243--250,
  2019{\natexlab{a}}.

\bibitem[Mao et~al.(2018)Mao, Sarkar, and Chakrabarti]{mao2018overlapping}
Xueyu Mao, Purnamrita Sarkar, and Deepayan Chakrabarti.
\newblock Overlapping clustering models, and one (class) svm to bind them all.
\newblock \emph{arXiv preprint arXiv:1806.06945}, 2018.

\bibitem[Chaudhuri et~al.(2012)Chaudhuri, Chung, and
  Tsiatas]{chaudhuri2012spectral}
Kamalika Chaudhuri, Fan Chung, and Alexander Tsiatas.
\newblock Spectral clustering of graphs with general degrees in the extended
  planted partition model.
\newblock In \emph{Conference on Learning Theory}, pages 35--1, 2012.

\bibitem[Amini et~al.(2013)Amini, Chen, Bickel, Levina,
  et~al.]{amini2013pseudo}
Arash~A Amini, Aiyou Chen, Peter~J Bickel, Elizaveta Levina, et~al.
\newblock Pseudo-likelihood methods for community detection in large sparse
  networks.
\newblock \emph{The Annals of Statistics}, 41\penalty0 (4):\penalty0
  2097--2122, 2013.

\bibitem[Le et~al.(2017)Le, Levina, and Vershynin]{le2017concentration}
Can~M Le, Elizaveta Levina, and Roman Vershynin.
\newblock Concentration and regularization of random graphs.
\newblock \emph{Random Structures \& Algorithms}, 51\penalty0 (3):\penalty0
  538--561, 2017.

\bibitem[Zhang and Rohe(2018)]{zhang2018understanding}
Yilin Zhang and Karl Rohe.
\newblock Understanding regularized spectral clustering via graph conductance.
\newblock \emph{arXiv preprint arXiv:1806.01468}, 2018.

\bibitem[Rohe et~al.(2020)Rohe, Muzhe, and Hayes]{githubCode}
Karl Rohe, Zeng Muzhe, and Alex Hayes.
\newblock Vintage sparse pca for non-parametric factor analysis.
\newblock \url{https://github.com/karlrohe/vsp}, 2020.

\bibitem[Bates and Maechler(2017)]{Matrix}
Douglas Bates and Martin Maechler.
\newblock \emph{Matrix: Sparse and Dense Matrix Classes and Methods}, 2017.
\newblock URL \url{https://CRAN.R-project.org/package=Matrix}.
\newblock R package version 1.2-12.

\bibitem[Qiu et~al.(2016)Qiu, Mei, and authors of the ARPACK library. See file
  AUTHORS~for details.]{rARPACK}
Yixuan Qiu, Jiali Mei, and authors of the ARPACK library. See file AUTHORS~for
  details.
\newblock \emph{rARPACK: Solvers for Large Scale Eigenvalue and SVD Problems},
  2016.
\newblock URL \url{https://CRAN.R-project.org/package=rARPACK}.
\newblock R package version 0.11-0.

\bibitem[Mardia et~al.(1979)Mardia, Kent, and Bibby]{mardia1979multivariate}
K.V. Mardia, J.T. Kent, and J.M. Bibby.
\newblock \emph{Multivariate analysis}.
\newblock Probability and mathematical statistics. 10th printing in 1995.
  Academic Press, 1979.
\newblock ISBN 9780124712508.

\bibitem[Bickel and Chen(2009)]{bickel2009nonparametric}
Peter~J Bickel and Aiyou Chen.
\newblock A nonparametric view of network models and newman--girvan and other
  modularities.
\newblock \emph{Proceedings of the National Academy of Sciences}, 106\penalty0
  (50):\penalty0 21068--21073, 2009.

\bibitem[Comon(1994)]{comon1994independent}
Pierre Comon.
\newblock Independent component analysis, a new concept?
\newblock \emph{Signal processing}, 36\penalty0 (3):\penalty0 287--314, 1994.

\bibitem[Chen and Bickel(2005)]{chen2005consistent}
Aiyou Chen and Peter~J Bickel.
\newblock Consistent independent component analysis and prewhitening.
\newblock \emph{IEEE Transactions on Signal Processing}, 53\penalty0
  (10):\penalty0 3625--3632, 2005.

\bibitem[Chen and Bickel(2006)]{chen2006efficient}
Aiyou Chen and Peter~J Bickel.
\newblock Efficient independent component analysis.
\newblock \emph{The Annals of Statistics}, 34\penalty0 (6):\penalty0
  2825--2855, 2006.

\bibitem[Wei(2015)]{ieeeICA}
Tianwen Wei.
\newblock A convergence and asymptotic analysis of the generalized symmetric
  fastica algorithm.
\newblock \emph{IEEE transactions on signal processing}, 63\penalty0
  (24):\penalty0 6445--6458, 2015.

\bibitem[Miettinen et~al.(2015)Miettinen, Taskinen, Nordhausen, and
  Oja]{miettinen2015fourth}
Jari Miettinen, Sara Taskinen, Klaus Nordhausen, and Hannu Oja.
\newblock Fourth moments and independent component analysis.
\newblock \emph{Statistical science}, 30\penalty0 (3):\penalty0 372--390, 2015.

\bibitem[Samworth and Yuan(2012)]{samworth2012independent}
Richard~J Samworth and Ming Yuan.
\newblock Independent component analysis via nonparametric maximum likelihood
  estimation.
\newblock \emph{The Annals of Statistics}, 40\penalty0 (6):\penalty0
  2973--3002, 2012.

\bibitem[Kruskal(1977)]{kruskal1977three}
Joseph~B Kruskal.
\newblock Three-way arrays: rank and uniqueness of trilinear decompositions,
  with application to arithmetic complexity and statistics.
\newblock \emph{Linear algebra and its applications}, 18\penalty0 (2):\penalty0
  95--138, 1977.

\bibitem[Janzamin et~al.(2019)Janzamin, Ge, Kossaifi, Anandkumar,
  et~al.]{janzamin2019spectral}
Majid Janzamin, Rong Ge, Jean Kossaifi, Anima Anandkumar, et~al.
\newblock Spectral learning on matrices and tensors.
\newblock \emph{Foundations and Trends{\textregistered} in Machine Learning},
  12\penalty0 (5-6):\penalty0 393--536, 2019.

\bibitem[Holland et~al.(1983)Holland, Laskey, and Leinhardt]{holland}
Paul~W Holland, Kathryn~Blackmond Laskey, and Samuel Leinhardt.
\newblock Stochastic blockmodels: First steps.
\newblock \emph{Social networks}, 5\penalty0 (2):\penalty0 109--137, 1983.

\bibitem[Karrer and Newman(2011)]{karrer2011stochastic}
Brian Karrer and Mark~EJ Newman.
\newblock Stochastic blockmodels and community structure in networks.
\newblock \emph{Physical review E}, 83\penalty0 (1):\penalty0 016107, 2011.

\bibitem[Latouche et~al.(2011)Latouche, Birmel{\'e}, and
  Ambroise]{latouche2011overlapping}
Pierre Latouche, Etienne Birmel{\'e}, and Christophe Ambroise.
\newblock Overlapping stochastic block models with application to the french
  political blogosphere.
\newblock \emph{The Annals of Applied Statistics}, pages 309--336, 2011.

\bibitem[Airoldi et~al.(2008)Airoldi, Blei, Fienberg, and
  Xing]{airoldi2008mixed}
Edoardo~M Airoldi, David~M Blei, Stephen~E Fienberg, and Eric~P Xing.
\newblock Mixed membership stochastic blockmodels.
\newblock \emph{Journal of Machine Learning Research}, 9\penalty0
  (Sep):\penalty0 1981--2014, 2008.

\bibitem[Zhang et~al.(2014)Zhang, Levina, and Zhu]{zhang2014detecting}
Yuan Zhang, Elizaveta Levina, and Ji~Zhu.
\newblock Detecting overlapping communities in networks using spectral methods.
\newblock \emph{arXiv preprint arXiv:1412.3432}, 2014.

\bibitem[Jin and Ke(2017)]{jin2017sharp}
Jiashun Jin and Zheng~Tracy Ke.
\newblock A sharp lower bound for mixed-membership estimation.
\newblock \emph{arXiv preprint arXiv:1709.05603}, 2017.

\bibitem[Rohe et~al.(2016)Rohe, Qin, and Yu]{pnas}
Karl Rohe, Tai Qin, and Bin Yu.
\newblock Co-clustering directed graphs to discover asymmetries and directional
  communities.
\newblock \emph{Proceedings of the National Academy of Sciences}, 113\penalty0
  (45):\penalty0 12679--12684, 2016.

\bibitem[Blei et~al.(2003)Blei, Ng, and Jordan]{blei}
David~M Blei, Andrew~Y Ng, and Michael~I Jordan.
\newblock Latent dirichlet allocation.
\newblock \emph{Journal of machine Learning research}, 3\penalty0
  (Jan):\penalty0 993--1022, 2003.

\bibitem[Wang and Rohe(2016)]{song}
Song Wang and Karl Rohe.
\newblock Discussion of ?coauthorship and citation networks for statisticians?
\newblock \emph{The Annals of Applied Statistics}, 10\penalty0 (4):\penalty0
  1820--1826, 2016.

\bibitem[Cape et~al.(2019{\natexlab{b}})Cape, Tang, Priebe,
  et~al.]{cape2019two}
Joshua Cape, Minh Tang, Carey~E Priebe, et~al.
\newblock The two-to-infinity norm and singular subspace geometry with
  applications to high-dimensional statistics.
\newblock \emph{The Annals of Statistics}, 47\penalty0 (5):\penalty0
  2405--2439, 2019{\natexlab{b}}.

\bibitem[Tropp(2012)]{tropp2012user}
Joel~A Tropp.
\newblock User-friendly tail bounds for sums of random matrices.
\newblock \emph{Foundations of computational mathematics}, 12\penalty0
  (4):\penalty0 389--434, 2012.

\bibitem[Mao et~al.(2017)Mao, Sarkar, and Chakrabarti]{mao2017estimating}
Xueyu Mao, Purnamrita Sarkar, and Deepayan Chakrabarti.
\newblock Estimating mixed memberships with sharp eigenvector deviations.
\newblock \emph{arXiv preprint arXiv:1709.00407}, 2017.

\bibitem[Van~de Geer(2000)]{van2000applications}
Sara~A Van~de Geer.
\newblock \emph{Applications of empirical process theory}, volume~91.
\newblock Cambridge University Press Cambridge, 2000.

\bibitem[Pollard(1990)]{pollard1990empirical}
David Pollard.
\newblock Empirical processes: theory and applications.
\newblock In \emph{NSF-CBMS regional conference series in probability and
  statistics}, pages i--86. JSTOR, 1990.

\bibitem[Chu and Trendafilov(1998)]{chu1998orthomax}
Moody~T Chu and Nickolay~T Trendafilov.
\newblock Orthomax rotation problem. a differential equation approach.
\newblock \emph{Behaviormetrika}, 25\penalty0 (1):\penalty0 13--23, 1998.

\bibitem[Sherin(1966)]{sherin1966matrix}
Richard~J Sherin.
\newblock A matrix formulation of kaiser's varimax criterion.
\newblock \emph{Psychometrika}, 31\penalty0 (4):\penalty0 535--538, 1966.

\bibitem[Neudecker(1981)]{neudecker1981matrix}
H~Neudecker.
\newblock On the matrix formulation of kaiser's varimax criterion.
\newblock \emph{Psychometrika}, 46\penalty0 (3):\penalty0 343--345, 1981.

\bibitem[ten Berge(1984)]{ten1984joint}
Jos~MF ten Berge.
\newblock A joint treatment of varimax rotation and the problem of
  diagonalizing symmetric matrices simultaneously in the least-squares sense.
\newblock \emph{Psychometrika}, 49\penalty0 (3):\penalty0 347--358, 1984.

\bibitem[Riordan(1937)]{riordan1937moment}
John Riordan.
\newblock Moment recurrence relations for binomial, poisson and hypergeometric
  frequency distributions.
\newblock \emph{The Annals of Mathematical Statistics}, 8\penalty0
  (2):\penalty0 103--111, 1937.

\end{thebibliography}

\newpage
\appendix
\section{Scree plot for New York Times analysis}\label{sec:nytscree}

\begin{figure}[h] 
   \centering
   \includegraphics[width=3in]{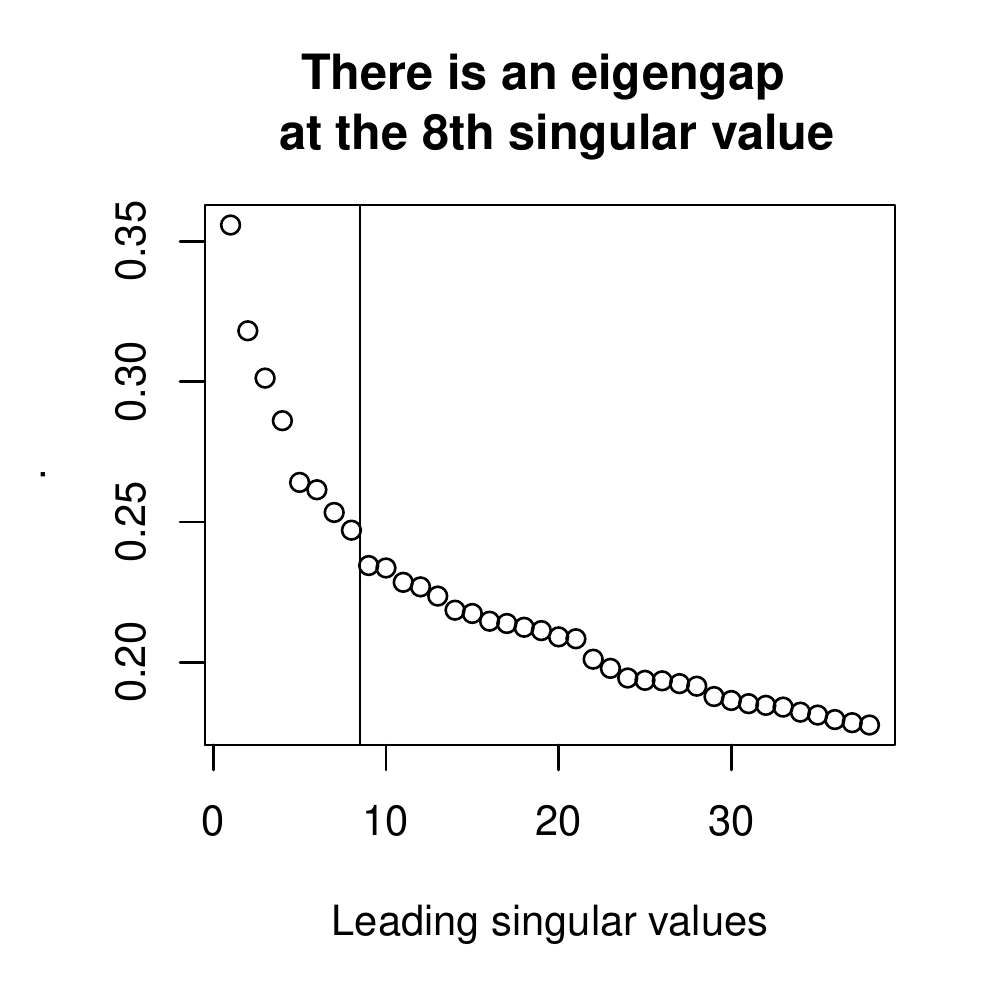} 
   \caption{This is a scree plot of the singular values $1, ..., 50$ of the matrix $L$ (defined in Remark \ref{remark:scalingstep}) after centering (i.e., step 1 of \texttt{vsp}).  A classical rule for selecting $k$ is to look for a ``gap'' in the scree plot.  Here, there is a gap after the 8th singular value that is illustrated with a vertical line.  This value is used for the illustrative analysis.  However, it should be noted that there is often more information contained in the principal components that correspond to singular values beyond the elbow \citep{song}. }
   \label{fig:nytscree}
\end{figure}

\section{Proofs for Proposition \ref{prop:centering} and Theorem \ref{thm:varimax1}} \label{app:popproofs}
The following is a proof of Proposition \ref{prop:centering}.  
\begin{proof}
With input $\A$, recall that the row, column, and grand means are
\[\mu_r = \A \1_d/d \in \R^n, \quad 
\mu_c = \1_n^T \A/n \in \R^d, \quad
\mu = \1_n^T \A \1_d /(nd) \in \R.\] 
The centered version of $\A$ is defined  to be $\widetilde \A = \A - \mu_r\1_d^T - \1_n \mu_c + \mu \1_n \1_d^T \in \R^{n \times d}$.  Define the column means of $Y$ and $Z$ as
\[\mu_Y = Y^T\1_d/d \quad\mbox{and} \quad  
\mu_Z = \1_n^TZ/n.\]
Note that 
\[\mu_r = ZB\mu_Y \quad 
\mu_c = \mu_ZBY^T \quad\mbox{and}\quad 
\mu = \mu_ZB\mu_Y
.\]
Also note that $\bar Y = \1_d \mu_Y^T$ and $\bar Z = \1_n\mu_Z$.  Putting the pieces together gives the result.
\begin{eqnarray*}
\widetilde \A &=& \A - \mu_r\1_d^T - \1_n \mu_c + \mu \1_n \1_d^T  \\
 &=& ZBY^T - ZB\mu_Y\1_d^T - \1_n \mu_ZBY +  \1_n \mu_ZB\mu_Y \1_d^T  \\
  &=& ZBY^T - ZB\bar Y^T - \bar Z BY^T +  \bar Z B\bar Y^T  \\
    &=& (Z - \bar Z) B(Y - \bar Y)^T
\end{eqnarray*}
\end{proof}

The following is a proof of Theorem \ref{thm:varimax1}.

\begin{proof} 
By the Assumption \ref{assumption:Varimax}, 
\[\E(Z^o_{i}) = 0, \E((Z^o_{i})^2) = 1, \E((Z^o_{i})^4) = \eta_i > 3, \forall i \in [k].\]
Thus, 
\[\E(v(R, Z^o \widetilde{R}^T)) = \sum_{j=1}^k  \E[Z^o \widetilde{R}^T R]_j^4.\]

To simplify notation, the proof reparameterizes the optimization parameter $R$ as follows.  For the rotation matrix $R$, define $O = \widetilde{R}^T R \in \oo(k)$.  We want to choose $O \in \oo(k)$ to optimize the quantities $\sum_j \E[Z^o O]_j^4 = \sum_j \E(Z^o O_{\cdot j})^4$, where $O_{\cdot j} \in \R^k$ is the $j$th column of $O$. Notice elements of $Z^o$ are independent and each has zero-mean. We have
\begin{eqnarray} \label{eq:reformulate1}
\sum_j \E(Z^o O_{\cdot j})^4 &=& \sum_{j=1}^k \left( \sum_{i=1}^k \E((Z^o_i)^4)O_{ij}^4 + 3\sum_{i\neq \ell}\E((Z^o_i)^2(z^o_\ell)^2)O_{ij}^2O_{\ell j}^2 \right) \nonumber \\
&=& \sum_{j=1}^k \left( \sum_{i=1}^k \eta_iO_{ij}^4 + 3\sum_{i\neq \ell}O_{ij}^2O_{\ell j}^2 \right). \label{eq:O2}
\end{eqnarray}

The above equation only depends on the squared elements of $O$.  Define $O^{(2)}  \in \R^{k \times k}$ such that 
$O^{(2)}_{ij} = O_{ij}^2$. Because $O \in \oo(k)$,  $O^{(2)}$ is a doubly stochastic matrix, where each element is non-negative and all row and column sums are equal to one.  
Define 
\begin{equation} \label{eq:reformulate2}
F_{\eta}(Q) =   \sum_{j=1}^k \left( \sum_{i=1}^k \eta_iQ_{ij}^2 + 3\sum_{i\neq \ell}Q_{ij}Q_{\ell j} \right)
\end{equation}
and define $\mathcal{S}(k)$ as the set of $k \times k$ doubly stochastic matrices. Note that
\[\sum_j \E(Z^o O_{\cdot j})^4 = F_{\mu}(O^{(2)}) \le \underset{Q\in \mathcal{S}(k)}{\max}  F_{\eta}(Q).\]
In this way, the Varimax problem relaxes from orthonormal matrices to doubly stochastic matrices.  
%


The rest of the proof will show that 
\begin{equation} \label{eq:doublyStochastic}
\underset{Q\in \mathcal{S}(k)}{\max}  F_{\eta}(Q) = \sum_{i=1}^k \eta_i.
\end{equation} 
Because $\sum_j \E(Z^o O_{\cdot j})^4$ evaluated with $O$ as the identity matrix, is equal to $\sum_{i=1}^k \eta_i$, it follows that $O=I$ or $R = \widetilde R$ obtains the maximum. Moreover, for $P \in \mathcal{P}(k)$, $O= P$ (i.e. $R = \widetilde R P$) obtains the maximum value.  It only remains to show Equation \eqref{eq:doublyStochastic}.

\begin{eqnarray*}
F_{\eta}(Q) &=& \sum_{j=1}^k \left( \sum_{i=1}^k \eta_iQ_{ij}^2 + 3\sum_{i\neq \ell}Q_{ij}Q_{\ell j} \right) \\
		  &=& \sum_{j=1}^k \left( \sum_{i=1}^k \eta_iQ_{ij}^2 + 3\left(\sum_{i=1}^kQ_{ij} \right)^2 - 3\sum_{i=1}^kQ_{ij}^2 \right) \\
		  &=& \sum_{j=1}^k \left( \sum_{i=1}^k \eta_iQ_{ij}^2 - 3\sum_{i=1}^kQ_{ij}^2 \right) + 3k \\
		  &=& \sum_{i=1}^k (\eta_i -3)\sum_{j=1}^k Q_{ij}^2 + 3k \\
		  &\leq& \sum_{i=1}^k (\eta_i -3)\sum_{j=1}^k Q_{ij} + 3k \\
		  &=&  \sum_{i=1}^k (\eta_i -3) + 3k \\
		  &=& \sum_{i=1}^k \eta_i.
\end{eqnarray*}
The inequality is because $Q_{ij} \in [0,1], \forall i,j$ (this is because $Q\in \mathcal{S}(k)$).

\vspace{.1in}

To see that the maximum of $\sum_j \E(Z^o O_{\cdot j})^4$ is \textit{only} attained by matrices in $\mathcal{P}(k)$, note that for any rotation matrix $O \not \in \mathcal{P}(k)$, then
\[\sum_j \E(Z^o O_{\cdot j})^4 = F_{\mu}(O^{(2)}) =
\sum_{i=1}^k (\eta_i -3)\sum_{j=1}^k O_{ij}^4 + 3k
<\sum_{i=1}^k (\eta_i -3)\sum_{j=1}^k O_{ij}^2 + 3k = \sum_{i=1}^k \eta_i,\]
where the inequality is now strict. 

%

\end{proof}

\vspace{.2in}

\subsection{A justification for the recentering step described in Remark \ref{remark:uncentering}} \label{app:recentering}
This section demonstrates how $\widehat \mu_Z = \sqrt{n}\widehat \mu_c \widehat V \widehat D^{-1} \Ruh$ can estimate $\mu_Z = \1^T Z/n$ under the Varimax assumptions on $Z$ by studying the population behavior of $\widehat \mu_Z$. Define the population version of the estimator as 
\[\mu_Z^* = \sqrt{n}\mu_c V D^{-1} \widetilde{R}_{U}, \]
where $\mu_c =  \1_n^T \A /n =  \1_n^T ZBY^T/n$, $V$ and $D$ are defined in Proposition \ref{prop:svd} with the SVD of $\widetilde \A$ as 
\[ D = \sqrt{nd} \widetilde D, \ \ 
V = d^{-1/2} \widetilde Y \widehat \Sigma_Y^{-1/2} \widetilde{R}_{V}^{T}, 
\] and $\Ru$ is the population Varimax rotation $\widetilde{R}_{U}$ (as justified by Theorem \ref{thm:varimax1}).  In the steps below, it is presumed that $Z$ satisfies the Varimax assumptions.  It is only presumed that $Y$ is full rank.  For simplicity, the $\approx$ correspond to approximating $\widehat \Sigma_Z$ as the identity matrix; under the Varimax assumptions, this is a reasonable approximation for large $n$. Recall that 
\[B\widehat \Sigma_Y^{1/2} \approx \widetilde{R}_{U}^T \widetilde D \widetilde{R}_{V}.\]
Thereby,
\begin{eqnarray*}
\mu_Z^* &=& \sqrt{n}\mu_c  V  D^{-1} \widetilde{R}_{U} \\
	      &=& \sqrt{n}\mu_Z (BY^T)  V D^{-1} \widetilde{R}_{U}  \\
	      &\approx& \sqrt{n}\mu_Z (\widetilde{R}_{U}^T \widetilde D \widetilde{R}_{V} \widehat \Sigma_Y^{-1/2} Y^T)  d^{-1/2} \widetilde Y \widehat \Sigma_Y^{-1/2} \widetilde{R}_{V}^T   D^{-1} \widetilde{R}_{U}.
\end{eqnarray*}
Then, $\widehat \Sigma_Y^{-1/2} Y^T \widetilde Y \widehat \Sigma_Y^{-1/2}$ is $d$ multiplied by the identity matrix. Substituting for $D$ and canceling out several terms yields the result,
\begin{eqnarray*}
\mu_Z^*  &\approx& (nd)^{1/2} \mu_Z \widetilde{R}_{U}^T \widetilde D \widetilde{R}_{V}   \widetilde{R}_{V}^T   D^{-1} \widetilde{R}_{U} \\
 &=& (nd)^{1/2} \mu_Z \widetilde{R}_{U}^T \widetilde D \widetilde{R}_{V}   \widetilde{R}_{V}^T   (nd)^{1/2} \widetilde{D}^{-1} \widetilde{R}_{U}  \\
  &=& \mu_Z \widetilde{R}_{U}^T \widetilde D \widetilde{R}_{V}   \widetilde{R}_{V}^T  \widetilde{D}^{-1} \widetilde{R}_{U} \\
&=& \mu_Z. 
\end{eqnarray*}

\noindent The rigorous proof will be shown later in Proposition \ref{prop:uncenter1} and Proposition \ref{prop:uncenter2}.

\subsection{Proofs for Lemmas \ref{lemma:ldaFactor} and\ref{lemma:naturalLDA}.} \label{sec:LDAlemmas}
The following is a proof of Lemma \ref{lemma:ldaFactor}.
\begin{proof}
For ease of notation, refer to the topic for word $w$ as $z$ (instead of $z_w$),
\[
\pr(w = j|Z_i) 
= \sum_{z=1}^k \pr(w = j|z, Z_i) \pr(z|Z_i) 
= \sum_{z=1}^k \beta_{j,z} Z_{i,z}
= \langle \beta_{j\cdot},Z_i\rangle.\]
So, step 3 in the LDA model is equivalent to choosing word $w$ to be word $j$ with probability $[\beta Z_i]_j$.  So, conditional on $N_i$ and $Z_i$, the $i$th row of $A$ is $Multinomial(N_i, \beta Z_i)$.
Then, unconditional on $N_i$, due to the Poisson-Multinomial relationship, each element in the $i$th row of $A$ is independent, with the distribution $A_{ij} \sim Poisson(\xi [\beta Z_i]_j)$.  So, $E(A|Z)  = \xi Z \beta^T$.
\end{proof}

The following is a proof of Lemma \ref{lemma:naturalLDA}.

\begin{proof}
There are three elements of Lemma \ref{lemma:naturalLDA}.  \textbf{Part 1:} 
conditionally on $Z_1, \dots, Z_n$ and $\Xi$, we need to show that the document-term matrix $A$ has independent Poisson entries satisfying 
\begin{equation}
\E(A|\Xi, Z) = (\Xi Z) \beta^T.
\end{equation}
The proof of this is equivalent to the proof of Lemma \ref{lemma:ldaFactor}.  

\textbf{Part 2:} The second part is that each element $(\Xi Z)_{ij}$ is independent Gamma$(\alpha_j, s)$.  To see this, let $X_i \in \R_+^k$ have independent Gamma elements, $X_{ij} \sim$ Gamma$(\alpha_j, s)$.  Define $\xi_i' = \sum_j X_{ij}$ and
\[Z_i' = \frac{X_i}{\xi_i'} .\]
It is well known that (1) $Z_i' \sim$ Dirichlet$(\alpha)$, (2) $\xi_i' \sim$ Gamma$(\sum_j \alpha_j, s)$, and (3) $\xi_i'$ is independent of $Z_i'$.  So, 
\[(\Xi Z)_{i} = \xi_i Z_i \stackrel{d}{=} \xi_i' Z_i' = X_i.\]

\textbf{Part 3:} We need to show that $\Xi Z \Sigma^{-1/2}$ satisfies the identification assumptions for Varimax. From part 2 above, each row contains $k$ independent random variables and each row is iid.  Then, each element of $\Xi Z$ is leptokurtic because the Gamma distribution is always leptokurtic.  Scaling by a constant $\Sigma^{-1/2}$ does not change this.  The fourth piece of the identification assumptions for Varimax is ensured by the scaling $\Sigma^{-1/2}$.


\end{proof}

\newpage

\section{Proofs of Main Theorems}

\subsection{Proof of Theorem \ref{thm:sparsity}} \label{appendix:sparsityproof}

\begin{proof}

Define the random variable $B \in \{0,1\}$ to be equal to 1 when $X \ne 0$ and equal to 0 when $X=0$.  For some $0<p<1/6$, $B\sim Bernouli(p)$.  Define random variable $S$ such that when $B=1$, $S=X$ and when $B=0$, then $S$ is equal in distribution to $X$ on the set $X\ne 0$.   So, \[X = SB.\] 

Under the conditions of the theorem and the construction above, $S$ has some arbitrary distribution with finite 4th moment and is also independent of $B$.  

\vspace{.1in}

Let $\mu_{i} = \mathbb{E}(S^{i})$. Then 
\begin{equation}
\theta := \mathbb{E}(X) = (1-p)0+p\mu_{1} = p\mu_{1},
\end{equation}
\begin{equation}
\mathbb{E}[(X-\theta)^{2}] = p\mu_{2} - p^{2}\mu_{1}^{2},
\end{equation}
\begin{equation}
\mathbb{E}[(X-\theta)^{4}] = p\mu_{4} - 4p^{2}\mu_{3}\mu_{1} + 6p^{3}\mu_{2}\mu_{1}^{2} - 3p^{4}\mu_{1}^{4}.
\end{equation}

So, in order to show that $\mathbb{E}[(X-\theta)^{4}] > 3\mathbb{E}[(X-\theta)^{2}]^{2}$, it is enough to show that
\begin{equation}\label{eq:sparse}
(\mu_{4} - 3p\mu_{2}^{2}) + 6p^{2}\mu_{1}^{2}(\mu_{2} - p\mu_{1}^{2}) > 4p\mu_{3}\mu_{1} - 6p^{2}\mu_{2}\mu_{1}^{2}. 
\end{equation}

Using  Lemma \ref{lemma:integral} with $g = S^{2}, h = 2pS$, $f$ being $S$'s pdf, we have 
\begin{equation} \label{eq:sparse2}
\mu_{4} - \mu_{2}^{2} + 4p^{2}\mu_{1}^{2}(\mu_{2} - \mu_{1}^{2}) \geq 4p\mu_{3}\mu_{1} - 4p\mu_{1}^{2}\mu_{2}.
\end{equation}

Subtract Equation \eqref{eq:sparse2} from Equation \eqref{eq:sparse} we only need to show
\begin{equation} \label{eq:sparse3}
(1-3p)\mu_{2}^{2} +p^{2}\mu_{1}^{2}(2\mu_{2}+4\mu_{1}^{2}-6p\mu_{1}^{2}) > (4p-6p^{2})\mu_{1}^{2}\mu_{2}.
\end{equation}

Notice $p<1/6$, thus $(6p-1)(p-1) > 0 \Rightarrow 1-3p > 4p - 6p^{2}$. Thus by Jensen's Inequality
$$ (1-3p)\mu_{2}^{2} \geq (4p-6p^{2})\mu_{2}^{2} \geq (4p-6p^{2})\mu_{1}^{2}\mu_{2}. $$

The first inequality is strict as long as $\mu_{2} > \mu_{1}^{2}$. Also with $p<1/6$ we have
$$ p^{2}\mu_{1}^{2}(2\mu_{2}+4\mu_{1}^{2}-6p\mu_{1}^{2}) \geq p^{2}\mu_{1}^{2}(2\mu_{2}+3\mu_{1}^{2}) \geq 0. $$

The second inequality is strict as long as $p>0, \mu_{1} \neq 0$. 

If $\mu_{2} = \mu_{1} = 0$ then $\pr(X=0) =1$, contradiction. Thus $X$ is leptokurtic. 

\begin{lemma} \label{lemma:integral}
Suppose $f$ is any distribution pdf. $g, h$ is any integrable functions. Then 
\begin{eqnarray*}
\int g^{2}fdx - (\int gfdx)^{2} + (\int h fdx)^{2}(\int h^{2}fdx - (\int h fdx)^{2}) \\
	\geq (\int ghfdx - \int gfdx\int hfdx)\int h fdx.
\end{eqnarray*}
\end{lemma}

Let $\widetilde{g} = g - \int gfdx$, $\widetilde{h} = h - \int hfdx$. Then

$$ \int g^{2}fdx - (\int gfdx)^{2} = \int \widetilde{g}^{2}fdx, $$
$$ \int h^{2}fdx - (\int h fdx)^{2} = \int \widetilde{h}^{2}fdx, $$
$$ \int ghfdx - \int gfdx\int hfdx = \int \widetilde{g}\widetilde{h}fdx. $$

By Cauchy-Schwartz Inequality,
\begin{eqnarray*}
&& \int g^{2}fdx - (\int gfdx)^{2} + (\int h fdx)^{2}(\int h^{2}fdx - (\int h fdx)^{2}) \\
&=& \int \widetilde{g}^{2}fdx + (\int h fdx)^{2} \int \widetilde{h}^{2}fdx \\
&\geq& |\int h fdx|\sqrt{\int \widetilde{g}^{2}fdx\int \widetilde{h}^{2}fdx} \\
&\geq& |\int h fdx|\int |\widetilde{g}\widetilde{h}|fdx \\
&\geq& |\int h fdx|(|\int ghfdx - \int gfdx\int hfdx|) \\
&\geq& (\int ghfdx - \int gfdx\int hfdx) \int hfdx.
\end{eqnarray*}

\end{proof}

\subsection{Leptokurtosis with soft sparsity} \label{app:sparsitylepto}
The random variable $X$ in Theorem \ref{thm:sparsity} satisfies a hard sparsity condition. Imagine $X$ as satisfying the conditions of Theorem \ref{thm:sparsity}.  The next proposition studies $X+W$, where $W$ is any independent random variable with a small variance.  So, if $W$ has a probability density, then $P(X+W = 0) \not > 0$, yet when $W$ has expectation zero, then $X+W$ is still close to zero with high probability. 
In this regime, the next proposition shows that if $X$ has a sufficiently large kurtosis, then $X+W$ is still leptokurtic, no matter the kurtosis of $W$.


\begin{prop} \label{prop:soft} 
Let $X$ and $W$ be any independent random variables with four finite moments.  
Let $\eta_{x,j} = \E(X - \E(X))^j$ and  $\eta_{w,j} = \E(W - \E(W))^j$. Let $\eta_{x,2}=1$.  For any $\epsilon>0$, if $\eta_{w,2}< \epsilon$,  and $\eta_{x,4} \ge 3(1+\epsilon)^2$,  then $X+W$ is leptokurtic.  
%
\end{prop}

Note that both $X$ and $W$ can be rescaled to satisfy the assumption  $\eta_{x,2}=1$.  In this way, it does not restrict the generality of the result. It only simplifies the notation.

\begin{proof}
 Note that $\eta_{x,1}= \eta_{w,1} =0$. Using that fact, 

\[\E\left(X+W - \E(X+W)\right)^2 = \eta_{x,2} +  \eta_{w,2} < 1 + \epsilon\]

and

\[\E\left(X+W - \E(X+W)\right)^4 = \eta_{x,4} + 6\eta_{x,2}\eta_{w,2} + \eta_{w,4} > 3(1+\epsilon)^2.\]

The result follows from the definition of leptokurtic.

\end{proof}

\subsection{Proofs for the main results, Theorem \ref{thm:main}} \label{appendix:main} 

\begin{proof}

We need six propositions listed below to prove Theorem \ref{thm:main}. Before the proof we clarify some notations. For a generic random matrix $X$, let $R_{X}$ be its sample Varimax rotation, i.e.
$$R_{X} \in \underset{R\in \oo(k)}{\arg \max} \ v(R, X),$$  
where $v(R,X)$ is defined in Equation \eqref{eq:Varimax}.
Then, let $R^{*}_{X}$ the population Varimax rotation, i.e. 
$$R^{*}_{X} \in \underset{R\in \oo(k)}{\arg \max} \ \mathcal{V}_{X}(R),$$ 
where the expectation in $\vv_{X}(R)= \E(v(R, X \widetilde R))$ is defined over the distribution of $X$ and the nuisance rotation $\widetilde R$ can be understood from the context. 
Define 
$$W = \underset{W_0\in \oo(k)}{\arg\min} \|\widehat{U}-UW_0\|_{2\to\infty}.$$ 
$\pp(k)$ is defined in Equation \eqref{eq:pp}. $P_n = P_n^{(1)}P_n^{(2)}P_n^{(3)}$ where $P_n^{(i)} \in \mathcal{P}(k), i=1,2,3$ are defined in Proposition \ref{prop:pop}, \ref{prop:converge}, \ref{prop:varifun} respectively. Let $\mu_Z = \1_n^TZ/n$. $J_n$ is $n$ by $n$ matrix with every entry equal to 1. $X^{\dagger}$ is the pseudo-inverse of $X$. Define $\xi = 1 + \epsilon$ for some small positive $\epsilon < 0.01$ for notation consistency with \cite{cape2019signal}. Recall $\Delta_n = n\rho_n, \bar{\Delta}_n = n\bar{\rho}_n$. 

Define
\begin{equation} \label{eq:gammadef}
\gamma_{ij}^{(n)} = \underset{s\geq 2}{\sup} \left( \frac{\mathbb{E} [(A_{ij} - \A_{ij})^s ]}{s!} \right)^{1/s} \quad \mbox{and} \quad \gamma^{(n)} = \underset{ij}{\sup} \ \gamma_{ij}^{(n)}.
\end{equation}
The $\gamma^{(n)}$ reveals the tail behaviors of sub-exponential random variables. It is useful in deriving matrix concentration results for sub-exponential random matrices later (Lemma \ref{lemma:Chung}). \\

See Sections \ref{proof:clt} through \ref{proof:uncenter2} for proofs to the following propositions \ref{prop:clt} through \ref{prop:uncenter2}. Several lemmas and technical details for these proofs are then delayed further into Sections \ref{sec:proof} and \ref{sec:FSOC}. \\

\begin{prop} \label{prop:clt}
Let $\widehat{\Sigma}_{Z} = \widetilde{Z}^{T}\widetilde{Z}/n$. Under the settings of Theorem \ref{thm:main},
\begin{equation}
\|U\widetilde{R}_{U} - U\widetilde{R}_{U}\widehat{\Sigma}_{Z}^{1/2}\|_{2\to \infty} = O_p(\frac{\log n}{n}).
\end{equation}
\end{prop}

\vspace{.2in}

\begin{prop} \label{prop:pop}
Under the settings of Theorem \ref{thm:main}, there exists $P_n^{(1)} \in \mathcal{P}(k)$ s.t.
\begin{equation}
\|\widehat{U}R_{UW}^{*} - U\widetilde{R}_{U}P_{n}^{(1)}\|_{2\to\infty} =  O_{p}\left(  (n\rho_{n})^{-1/2}n^{-1/2}\log^{\frac{5}{2}} n \right). 
\end{equation}
\end{prop}

\vspace{.2in}

\begin{prop} \label{prop:converge}
Under the settings of Theorem \ref{thm:main}, there exists $P_n^{(2)} \in \mathcal{P}(k)$ such that for any $\delta > 0$,
\begin{equation}
\|\widehat{U}R_{UW} - \widehat{U}R^{*}_{UW}P_{n}^{(2)}\|_{2\to \infty} = O_{p}(n^{\delta/2-3/4}\log n).
\end{equation}
\end{prop}

\vspace{.2in}

\begin{prop} \label{prop:varifun}
Under the settings of Theorem \ref{thm:main}, there exists $P_n^{(3)} \in \mathcal{P}(k)$ s.t.
\begin{equation}
\|\widehat{U}R_{\widehat{U}} - \widehat{U}R_{UW}P_n^{(3)}\|_{2\to \infty} =  O_{p}\left(  (n\rho_{n})^{-1/4}n^{-1/2}\log^{\frac{11}{4}} n \right).
\end{equation}
\end{prop}

\vspace{.2in}

\begin{prop} \label{prop:uncenter1}
Define $P_n = P_n^{(1)}P_n^{(2)}P_n^{(3)}$ with $P_n^{(1)}, P_n^{(2)}, P_n^{(3)}$ defined in Proposition \ref{prop:pop}, \ref{prop:converge}, \ref{prop:varifun} respectively.  Under the settings of Theorem \ref{thm:main}, for any $\delta > 0$,
\begin{equation}
\|J_{n}(A\widehat{V}\widehat{D}^{-1}R_{\widehat{U}} - \A VD^{-1}\widetilde{R}_{U}P_n)\|_{2\to\infty} = O_p\left(n^{\delta / 2 + 1/4} +  (n\rho_{n})^{-1/4}n^{1/2}\log^{\frac{7}{4}} n \right).
\end{equation}
\end{prop}

\vspace{.2in}

\begin{prop} \label{prop:uncenter2}
Under the settings of Theorem \ref{thm:main},
\begin{equation}
\|J_{n}(\sqrt{n}\A VD^{-1}\widetilde{R}_{U} - Z)\|_{2\to\infty} = O_p(\sqrt{n}\log n).
\end{equation}
\end{prop}

\vspace{.2in}
\newpage

We are going to show the bound for $\|\sqrt{n}\widehat{U}R_{\widehat{U}} -  \widetilde{Z}P_n\|_{2\to\infty}$ by splitting it into four parts using triangle inequalities. Proposition \ref{prop:clt}, \ref{prop:pop}, \ref{prop:converge}, \ref{prop:varifun} give the bound for each split component. Similarly we show the bound for $\|\1_n\widehat{\mu}_{Z} - \1_n \mu_{Z}^{T}P_n\|_{2\to\infty}$ by decomposing it into two parts and use Proposition \ref{prop:uncenter1}, \ref{prop:uncenter2} to give bounds. The proofs of these propositions are shown after the proof of Theorem \ref{thm:main}. 
The propositions that justify the equalities below are numbered on the left side of the equalities. 

\begin{eqnarray} \label{eq:main1}
&& \|\sqrt{n}\widehat{U}R_{\widehat{U}} -  \widetilde{Z}P_n\|_{2\to\infty}  \nonumber \\
(\mbox{Proposition }\ref{prop:svd}) &=& \|\sqrt{n}\widehat{U}R_{\widehat{U}} - \sqrt{n}U\widetilde{R}_{U}\widehat{\Sigma}_{Z}^{1/2}P_n\|_{2\to\infty} \nonumber \\
		      &=& \|\sqrt{n}\widehat{U}R_{\widehat{U}} - \sqrt{n}\widehat{U}R_{UW}P_{n}^{(3)} + \sqrt{n}\widehat{U}R_{UW}P_{n}^{(3)} - \sqrt{n}\widehat{U}R_{UW}^{*}P_{n}^{(2)}P_{n}^{(3)} \nonumber \\
		      && + \sqrt{n}\widehat{U}R_{UW}^{*}P_{n}^{(2)}P_{n}^{(3)} - \sqrt{n}U\widetilde{R}_{U}P_n + \sqrt{n}U\widetilde{R}_{U}P_n - \sqrt{n}U\widetilde{R}_{U}\widehat{\Sigma}_{Z}^{1/2}P_n\|_{2\to\infty}\nonumber  \\
		      &\leq& \|\sqrt{n}\widehat{U}R_{\widehat{U}} - \sqrt{n}\widehat{U}R_{UW}P_{n}^{(3)}\|_{2\to\infty} + \|\sqrt{n}\widehat{U}R_{UW}P_{n}^{(3)} - \sqrt{n}\widehat{U}R_{UW}^{*}P_{n}^{(2)}P_{n}^{(3)}\|_{2\to\infty}\nonumber  \\
		      && + \|\sqrt{n}\widehat{U}R_{UW}^{*}P_{n}^{(2)}P_{n}^{(3)} - \sqrt{n}U\widetilde{R}_{U}P_n\|_{2\to\infty} + \|\sqrt{n}U\widetilde{R}_{U}P_n - \sqrt{n}U\widetilde{R}_{U}\widehat{\Sigma}_{Z}^{1/2}P_n\|_{2\to\infty} \nonumber \\		      
(\mbox{Proposition }\ref{prop:clt})	&=& \|\sqrt{n}\widehat{U}R_{\widehat{U}} - \sqrt{n}\widehat{U}R_{UW}P_{n}^{(3)}\|_{2\to\infty} + \|\sqrt{n}\widehat{U}R_{UW}P_{n}^{(3)} - \sqrt{n}\widehat{U}R_{UW}^{*}P_{n}^{(2)}P_{n}^{(3)}\|_{2\to\infty} \nonumber \\
		      && + \|\sqrt{n}\widehat{U}R_{UW}^{*}P_{n}^{(2)}P_{n}^{(3)} - \sqrt{n}U\widetilde{R}_{U}P_n\|_{2\to\infty} +  O_{p}(\frac{\log n}{\sqrt{n}})\nonumber  \\
(\mbox{Proposition }\ref{prop:pop})	&=& \|\sqrt{n}\widehat{U}R_{\widehat{U}} - \sqrt{n}\widehat{U}R_{UW}P_{n}^{(3)}\|_{2\to\infty} + \|\sqrt{n}\widehat{U}R_{UW}P_{n}^{(3)} - \sqrt{n}\widehat{U}R_{UW}^{*}P_{n}^{(2)}P_{n}^{(3)}\|_{2\to\infty}\nonumber  \\
		      && +  O_{p}\left(  (n\rho_{n})^{-1/2}\log^{\frac{5}{2}} n \right) +  O_{p}(\frac{\log n}{\sqrt{n}})\nonumber \\	
(\mbox{Proposition }\ref{prop:converge})	&=& \|\sqrt{n}\widehat{U}R_{\widehat{U}} - \sqrt{n}\widehat{U}R_{UW}P_{n}^{(1)}\|_{2\to\infty} + O_{p}(n^{\delta/2-1/4}\log n) \nonumber \\
		      && +  O_{p}\left( (n\rho_{n})^{-1/2}\log^{\frac{5}{2}} n \right) +  O_{p}(\frac{\log n}{\sqrt{n}})\nonumber \\	
(\mbox{Proposition }\ref{prop:varifun})	&=& O_{p}\left(  (n\rho_{n})^{-1/4}\log^{\frac{11}{4}} n \right) + O_{p}(n^{\delta/2-1/4}\log n)  +  O_{p}\left( (n\rho_{n})^{-1/2}\log^{\frac{5}{2}} n \right) +  O_{p}(\frac{\log n}{\sqrt{n}})\nonumber \\
		      &=& O_{p}\left(  (n\rho_{n})^{-1/4}\log^{\frac{11}{4}} n \right) + O_{p}(n^{\delta/2-1/4}\log n)  \nonumber \\
		      &=& O_p\left(  \Delta_n^{-1/4+\delta/2}\log^{\frac{11}{4}} n \right).
\end{eqnarray}

For the recentering part, by Proposition \ref{prop:uncenter1}, \ref{prop:uncenter2}, 
\begin{eqnarray} \label{eq:main2}
&& \|\1_n\widehat{\mu}_{Z} - \1_n \mu_{Z}P_n\|_{2\to\infty} =  \frac{1}{n}\|J_n^{T}(\sqrt{n}A\widehat{V}\widehat{D}^{-1}R_{\widehat{U}} - ZP_n)\|_{2\to\infty} \nonumber \\
&\leq& \frac{1}{n}\|J_n(\sqrt{n}A\widehat{V}\widehat{D}^{-1}R_{\widehat{U}} - \sqrt{n}\A VD^{-1}\widetilde{R}_{U}P_n)\|_{2\to\infty} + \frac{1}{n}\|J_n(\sqrt{n}\A VD^{-1}\widetilde{R}_{U}P_n - ZP_n)\|_{2\to\infty} \nonumber \\
&=& O_p(n^{\delta / 2 - 1/4}  + (n\rho_{n})^{-1/4}\log^{\frac{7}{4}} n + \frac{\log n}{\sqrt{n}}) \nonumber \\
&=& O_p( \Delta_n^{-1/4+\delta/2}\log^{\frac{7}{4}} n ).
\end{eqnarray}

Take $\delta = 0.2$. Equation \eqref{eq:main1}, \eqref{eq:main2} and triangle inequality accomplish the proof.
\end{proof}

\vspace{.1in}

Before the proofs for the six propositions, two useful lemmas are given. Lemma \ref{lemma:maxZ} gives bound for the maximum absolute value of $Z$'s elements. Lemma \ref{lemma:twonorm} borrows matrix $2\to\infty$ norm's property from \cite{cape2019two}.

\begin{lemma} \label{lemma:maxZ}
\[ \underset{i,j}{\max}|Z_{ij}| = O_{p}(\log n),\quad \underset{i,j}{\max}|\widetilde{Z}_{ij}| = O_{p}(\log n).\]
\[ \underset{i,j}{\max}|Y_{ij}| = O_{p}(\log d),\quad \underset{i,j}{\max}|\widetilde{Y}_{ij}| = O_{p}(\log d).\]
\end{lemma}
\begin{proof} Assumption \ref{assumption:ztail} indicates $Z$'s columns are sub-exponential variables. Thus there exists $C_0, \lambda_j > 0, j\in [k]$'s s.t.
\begin{equation} \label{eq:lambda}
\pr(|Z_{ij} - \mathbb{E}Z_{ij}| > t) \leq C_{0} \exp(-\lambda_{j} t) \leq C_{0} \exp(-\lambda t),
\end{equation}
with $\lambda = \underset{j\in [k]}{\min}\lambda_{j}$. Then 
\begin{eqnarray*}
\pr(\underset{i,j}{\max}|Z_{ij} - \mathbb{E}Z_{ij}| > t) &\leq& \sum_{i,j}\pr(|Z_{ij} - \mathbb{E}Z_{ij}| > t) \\
						   &\leq& \sum_{i,j}C_{0} \exp(-\lambda t) \\
						   &\leq& knC_{0}  \exp(-\lambda t).
\end{eqnarray*}
$\Rightarrow$
$$ \underset{i,j}{\max}|Z_{ij}| = O_{p}(\log n),  \quad \underset{i,j}{\max}|\widetilde{Z}_{ij}| = O_{p}(\log n). $$
Similar conclusion also applies to $Y$.
\end{proof}

With Lemma \ref{lemma:maxZ}, it could be trivially inferred that
\begin{equation} \label{eq:rhos}
\bar\rho_n = O(\rho_n\log^2 n).
\end{equation}

\begin{lemma} \label{lemma:twonorm}
Suppose $X_1 \in \R^{n_1 \times n_2}, X_2 \in \R^{n_2 \times n_3}$ are real matrices. Then  
\begin{equation}
\|X_1 X_2\|_{2\to \infty} \leq \|X_1\|_{2\to \infty} \|X_2\|.
\end{equation}
\end{lemma}
This is a direct conclusion from Proposition 6.5 in \cite{cape2019two}.

\vspace{.2in}

\subsubsection{Proof of Proposition \ref{prop:clt}} \label{proof:clt}
\begin{proof}
The $(i,j)$ entry of $\widehat{\Sigma}_{Z} \in \R^{k \times k}$ is
\begin{displaymath}
\widehat{\Sigma}_{Z}[i, j] = \left \{ \begin{array} {ll}
\frac{1}{n}\sum_{q=1}^{n}(Z_{qi}-\widehat{\mu}_{Z}[i])^{2} & \textrm{if $i = j$,}\\
\frac{1}{n}\sum_{q=1}^{n}(Z_{qi}-\widehat{\mu}_{Z}[i])(Z_{qj}-\widehat{\mu}_{Z}[j]) & \textrm{if $i\neq j$.} \\
\end{array} \right.
\end{displaymath}
By LLN,  $\|\widehat{\Sigma}_{Z}-I\|_{\max} = O_{p}(\frac{k^{2}}{\sqrt{n}})= O_{p}(\frac{1}{\sqrt{n}})$,  thus $\|\widehat{\Sigma}_{Z}-I\| \leq \sqrt{k^{2}}\|\widehat{\Sigma}_{Z}-I\|_{\max} = O_{p}(\frac{1}{\sqrt{n}})$. Suppose eigendecomposition of $\widehat{\Sigma}_{Z}$ is $\widehat{\Sigma}_{Z} = \Psi\Lambda_{Z} \Psi^{T}$. Then 
$$\|\widehat{\Sigma}_{Z} - I\| =\|\Lambda_{Z} - I\| = O_{p}(\frac{1}{\sqrt{n}}) \Rightarrow \|\widehat{\Sigma}_{Z}^{1/2} - I\| =  \|\Lambda_{Z}^{1/2} - I\|  = O_{p}(\frac{1}{\sqrt{n}}).$$
Also $\|\widehat{\Sigma}_{Z}-I\| = O_{p}(\frac{1}{\sqrt{n}})$ implies $\|\widehat{\Sigma}_{Z}^{-1/2}\| = O_{p}(1).$
By Proposition \ref{prop:svd} and Lemma \ref{lemma:twonorm}, \ref{lemma:maxZ},
\begin{equation} \label{eq:maxU}
\|U\|_{2\to\infty} = \frac{1}{\sqrt{n}}\|\widetilde{Z}\widehat{\Sigma}^{-1/2}_{Z}\|_{2\to \infty} \leq \frac{1}{\sqrt{n}}\|\widetilde{Z}\|_{2\to\infty}\|\widehat{\Sigma}^{-1/2}_{Z}\| = O_p(\frac{\log n}{\sqrt{n}}).
\end{equation}
Putting the above pieces together provides a bound on the quantity of interests.
\begin{eqnarray*}
\|U\widetilde{R}_{U} - U\widetilde{R}_{U}\widehat{\Sigma}_{Z}^{1/2}\|_{2\to \infty} 
&\leq& \|U\widetilde{R}_{U}\|_{2\to \infty}\|I-\widehat{\Sigma}_{Z}^{1/2}\| \\
	&=& \|U\|_{2\to \infty}\|I-\widehat{\Sigma}_{Z}^{1/2}\| \\
	&=& O_{p}(\frac{\log n}{n}).
\end{eqnarray*}
\end{proof}

\subsubsection{Proof of Proposition \ref{prop:pop}}

We give the statement of Lemma \ref{lemma:Chung}, \ref{lemma:Cape} below and use them to prove proposition \ref{prop:pop}. The proof of these two lemmas will be shown in Section \ref{sec:proof}.

\vspace{.1in}
\begin{lemma} \label{lemma:Chung}	
Define the symmetrized adjacent matrix as
$\widetilde{A}_{sym} = \begin{pmatrix}
0 & \widetilde A \\
\widetilde{A}^{T} & 0
\end{pmatrix}$ and its population version as 
$\widetilde{\A}_{sym} = \begin{pmatrix}
0 & \widetilde\A \mbox{ }\\
\widetilde{\A}^{T} & 0
\end{pmatrix}$.
Under the settings in Theorem \ref{thm:main},
\begin{equation}
 \|A_{sym} - \A_{sym}\| = O_p((n\rho_n\log^3 n)^{\frac{1}{2}}). 
\end{equation}
\end{lemma}

\vspace{.1in}

\begin{lemma} \label{lemma:Cape}
Presume the conditions in Theorem \ref{thm:main}. There exists $W \in \oo(k)$, such that
\begin{equation} \label{EigenPerturb}
\|\widehat{U} - UW\|_{2\to \infty} =  O_{p}\left(  (n\rho_{n})^{-1/2}n^{-1/2}\log^{\frac{5}{2}} n \right). 
\end{equation}
\end{lemma}

\vspace{.1in}

Lemma \ref{lemma:Cape} gives a row-wise bound for the eigenvectors' fluctuations.  This lemma follows from Theorem 1 in \cite{cape2019signal}, which requires several conditions. Lemma \ref{lemma:Chung} is used for one of the conditions. The other conditions are either already satisfied by the  assumptions of Theorem \ref{thm:main} or checked inside the proof of Lemma \ref{lemma:Cape}.

\begin{proof} 

 Notice the fact that $2\to \infty$ norm is invariant to rotations. 
 From Theorem \ref{thm:varimax1} there exist $P_n^{(1)} \in \mathcal{P}(k)$ s.t. $R_{UW}^{*} = W^{T}\widetilde{R}_{U}P_n^{(1)}$. Therefore
\begin{eqnarray*}
\|\widehat{U}R_{UW}^{*} - U\widetilde{R}_{U}P_{n}^{(1)}\|_{2\to\infty} &=& \|\widehat{U}W^{T}\widetilde{R}_{U}P_{n}^{(1)} - U\widetilde{R}_{U}P_{n}^{(1)}\|_{2\to\infty} \\
			&=& \|\widehat{U}W^{T} - U\|_{2\to\infty} \\
			&=& \|\widehat{U} - UW\|_{2\to\infty} \\
(\mbox{Lemma } \ref{lemma:Cape}) &=&  O_{p}\left(  (n\rho_{n})^{-1/2}n^{-1/2}\log^{\frac{5}{2}} n \right).
\end{eqnarray*}
\end{proof}

\subsubsection{Proof of Proposition \ref{prop:converge}}
The proof of Proposition \ref{prop:converge} uses the following lemma to bound the distance between sample and population Varimax solutions (modulo permutation and sign flip).

\begin{lemma}\label{lemma:SOC} Recall that $R_{\widetilde{Z}} \in \underset{R_0 \in \mathcal{O}(k)}{\arg \max}$ $v(R_0, \widetilde{Z})$. There exists $P_{n}^{(2)} \in \mathcal{P}(k)$ s.t. for $\forall \delta > 0$
$$\|R_{\widetilde{Z}}-P_{n}^{(2)}\|_{2\to\infty}  = O_{p}(n^{\delta/2-1/4}). $$ 
\end{lemma}

The proof of Lemma \ref{lemma:SOC} is in Section \ref{sec:proof}. 

\begin{proof} 

With some previous lemmas,
\begin{eqnarray*}
&& \|\widehat{U}R_{UW} - \widehat{U}R^{*}_{UW}P_{n}^{(2)}\|_{2\to \infty} \\
&=& \|\widehat{U}(R_{UW}-R_{UW}^{*}P_n^{(2)})\|_{2\to\infty} \\
(\mbox{Lemma } \ref{lemma:twonorm})&\leq&\|\widehat{U}\|_{2\to\infty}\|R_{UW}-R_{UW}^{*}P_n^{(2)}\| \\
(\mbox{Lemma }\ref{lemma:SOC})&=&O_{p}(n^{\delta/2-1/4}\|\widehat{U}\|_{2\to\infty}) \\
	&\leq& O_{p}(n^{\delta/2-1/4}\|\widehat{U} - UW\|_{2\to\infty} + n^{\delta/2-1/4}\|UW\|_{2\to\infty}) \\
(\mbox{Lemma } \ref{lemma:Cape})&=&  O_{p}\left(  (n\rho_{n})^{-1/2}n^{-3/4}\log^{\frac{5}{2}} n \right) + O_{p}(n^{\delta/2-1/4}\|UW\|_{2\to\infty}) \\
  	&\leq&  O_{p}\left(  (n\rho_{n})^{-1/2}n^{-3/4}\log^{\frac{5}{2}} n \right) + O_{p}(n^{\delta/2-1/4}\|U\|_{2\to\infty}) \\
(\mbox{Equation } (\ref{eq:maxU}))&\leq&  O_{p}\left(  (n\rho_{n})^{-1/2}n^{-3/4}\log^{\frac{5}{2}} n \right)+O_{p}(n^{\delta/2-3/4}\log n) \\
(n\rho_n \succeq \log^{2\xi}n)&=& O_{p}(n^{\delta/2-3/4}\log n). 
\end{eqnarray*}

\end{proof}

\subsubsection{Proof of Proposition \ref{prop:varifun}}

This proposition shows that $R_{\widehat{U}}$ converges to $R_{UW}$.  The proof of Proposition \ref{prop:varifun} is contained in Section \ref{sec:proof}.  
This proof uses the fact that the Varimax objective function is smooth and each row of $\widehat U$ converges to the corresponding row of $UW$ (i.e. $ \|\widehat{U} - UW\|_{2\to\infty} \rightarrow 0$). This implies that the Varimax solution computed with   $\widehat U$ (i.e. $R_{\widehat{U}}$) converges to the Varimax solution computed with $UW$ (i.e. $R_{UW}$).


%


\subsubsection{Proof of Proposition \ref{prop:uncenter1}}

\begin{proof}


%
%
\begin{eqnarray} \label{eq:4ineqs}
 && \|J_{n}(A\widehat{V}\widehat{D}^{-1}R_{\widehat{U}} - \A VD^{-1}\widetilde{R}_{U}P_n)\|_{2\to\infty}  \nonumber \\
(\mbox{Lemma } \ref{lemma:twonorm})&\leq& \sqrt{n}\|A\widehat{V}\widehat{D}^{-1}R_{\widehat{U}} - \A VD^{-1}\widetilde{R}_{U}P_n\| \nonumber \\
(WR^{*}_{UW} = \widetilde{R}_{U}P_n^{(1)})&=& \sqrt{n}\|A\widehat{V}\widehat{D}^{-1}R_{\widehat{U}} - AVD^{-1}W R_{\widehat{U}} + AVD^{-1}W R_{\widehat{U}} - AVD^{-1}WR_{UW}P_n^{(3)} \nonumber \\
&& \quad + AVD^{-1}WR_{UW}P_n^{(3)} - AVD^{-1}WR_{UW}^{*}P_n^{(2)}P_n^{(3)}  \nonumber \\
&& \quad + AVD^{-1}WR_{UW}^{*}P_n^{(2)}P_n^{(3)} - \A VD^{-1}WR_{UW}^{*}P_n^{(2)}P_n^{(3)}\| \nonumber \\
&\leq& \sqrt{n} ( \|A\widehat{V}\widehat{D}^{-1}R_{\widehat{U}} - AVD^{-1}W R_{\widehat{U}}\|  + \|AVD^{-1}W R_{\widehat{U}} - AVD^{-1}WR_{UW}P_n^{(3)}\| \nonumber  \\
&& \quad + \|AVD^{-1}WR_{UW} - AVD^{-1}WR_{UW}^{*}P_n^{(2)}\|  \nonumber  \\
&& \quad + \|AVD^{-1}WR_{UW}^{*} - \A VD^{-1}WR_{UW}^{*}\| ).
\end{eqnarray}

\vspace{.1in}
The fact that $WR^{*}_{UW} = \widetilde{R}_{U}^{T}P_n^{(1)}$ is a direct result of Theorem \ref{thm:varimax1}. The remaining part of the proof wants to show the bounds for each term of RHS of Equation \eqref{eq:4ineqs}.

\vspace{.2in}

First term of Equation \eqref{eq:4ineqs} is $\|A\widehat{V}\widehat{D}^{-1}R_{\widehat{U}} - AVD^{-1}R_{\widehat{U}}\|$. By Lemma \ref{lemma:Cape}  
$$ \|\widehat{U} - UW\|_{2\to\infty} = O_{p}\left(  (n\rho_{n})^{-1/2}n^{-1/2}\log^{\frac{5}{2}} n \right),  $$
and by the same virtue (notice $Y$ also satisfies Assumption 2, it could be shown by transposing the adjacency matrix) there exists $W_2 \in \oo(k)$ s.t.
$$ \|\widehat{V} - VW_2\|_{2\to\infty} = O_{p}\left(  (n\rho_{n})^{-1/2}n^{-1/2}\log^{\frac{5}{2}} n \right). $$

By assumptions and Lemma \ref{lemma:Chung},
$$\|D^{-1}\| = O_p((n\rho_n)^{-1}), \|A-\A\| = O_p((n\rho_n \log^3 n)^{\frac{1}{2}}). $$ 

Notice that $\|X\| \leq \sqrt{m_1}\|X\|_{2\to\infty}$ for $\forall X \in \R^{m_1 \times m_2}$ and $\|V\| = 1$. Therefore

\begin{eqnarray} 
&& \|A\widehat{V}\widehat{D}^{-1}R_{\widehat{U}} - AVD^{-1}W R_{\widehat{U}}\|  \nonumber \\
			&=& \|A\widehat{V}\widehat{D}^{-1} - AVD^{-1}W\|  \nonumber \\
			&\leq& \|A\| \|\widehat{V}\widehat{D}^{-1} - VD^{-1}W\|  \nonumber \\
			&=& \|A-\A+\A\| \|\widehat{V}\widehat{D}^{-1} - VW_2\widehat{D}^{-1} + VW_2\widehat{D}^{-1} - VD^{-1}W\| \nonumber \\
			&\leq& (\|A-\A\|+\|\A\|) (\|\widehat{V}\widehat{D}^{-1} - VW_2\widehat{D}^{-1}\| + \|VW_2\widehat{D}^{-1} - VD^{-1}W\|) \nonumber \\
			&\leq& (\|A-\A\|+\|\A\|) (\|\widehat{V} - VW_2\| \|\widehat{D}^{-1}\| + \|V\| \|W_2\widehat{D}^{-1} - D^{-1}W\|) \nonumber \\
			&=& O_p(\|\A\| \times \|D^{-1}\|) \times  [O_{p}\left(  (n\rho_{n})^{-1/2}n^{-1/2}\log^{\frac{5}{2}} n \right) + O_{p}\left(  (n\rho_{n})^{-1/2}\log^{\frac{5}{2}} n \right) ] \nonumber \\
			&=& O_{p}\left( (n\rho_{n})^{-1/2}\log^{\frac{5}{2}} n \right).  \label{eq:aa}
\end{eqnarray}

The third equation employs the bound of $\|W_2\widehat{D}^{-1} - D^{-1}W\|$ from the following deduction:
\begin{eqnarray*}
&& \|W_2\widehat{D}^{-1} - D^{-1}W\| \\
&=& \|\widehat{D}^{-1} - W_{2}^{T}D^{-1}W\| \\
&=& \|\widehat{V}\widehat{D}^{-1}\widehat{U}^{T} - \widehat{V}W_{2}^{T}D^{-1}W\widehat{U}^{T}\| \\
&=& \|\widehat{V}\widehat{D}^{-1}\widehat{U}^{T} - VD^{-1}U^{T} + (V-\widehat{V}W_{2}^{T})D^{-1}WU + \widehat{V}W_{2}^{T}D^{-1}(U^{T} - W\widehat{U}^{T})\| \\
&\leq& \|\widehat{V}\widehat{D}^{-1}\widehat{U}^{T} - VD^{-1}U^{T}\| + \|(V-\widehat{V}W_{2}^{T})D^{-1}WU\| + \|\widehat{V}W_{2}^{T}D^{-1}(U^{T} - W\widehat{U}^{T})\| \\
&\leq& \|A^{\dagger} - \A^{\dagger}\| + \sqrt{d}\|V-\widehat{V}W_{2}^{T}\|_{2\to\infty}\|D^{-1}\| + \sqrt{n}\|D^{-1}\| \|U^{T} - W\widehat{U}^{T}\|_{2\to\infty} \\
&\leq& \|A^{\dagger}\|\|A-\A\|\|\A^{\dagger}\| + \sqrt{d}\|V-\widehat{V}W_{2}^{T}\|_{2\to\infty}\|D^{-1}\| + \sqrt{n}\|D^{-1}\| \|U^{T} - W\widehat{U}^{T}\|_{2\to\infty} \\
&=& O_p(\|D^{-1}\|) \times \left( O_p\left((n\rho_n)^{-1/2}\log^{\frac{3}{2}}n\right) +  O_{p}\left( (n\rho_{n})^{-1/2}\log^{\frac{5}{2}} n \right) \right) \nonumber \\
&=& O_p(\|D^{-1}\|) \times O_{p}\left(  (n\rho_{n})^{-1/2}\log^{\frac{5}{2}} n \right).
\end{eqnarray*}

\vspace{.1in}

Second term of Equation \eqref{eq:4ineqs} is $\|AVD^{-1}R_{\widehat{U}} - AVD^{-1}WR_{UW}P_n^{(3)}\|$. According to Equation \eqref{eq:Pn3} (in the proof of Proposition \ref{prop:varifun}) there exists a $P_n^{(3)} \in \mathcal{P}(k)$ s.t.
$$ \|R_{\widehat{U}} - R_{UW}P_{n}^{(3)}\|_{2\to\infty} = O_{p}\left( (n\rho_{n})^{-1/4}\log^{\frac{7}{4}} n \right). $$

Therefore,
\begin{eqnarray} \label{eq:bb}
&& \|AVD^{-1}WR_{\widehat{U}} - AVD^{-1}WR_{UW}P_n^{(3)}\|  \nonumber \\
&\leq& \|A\|\|V\|\|D^{-1}W\|\sqrt{k}\|R_{\widehat{U}} - R_{UW}P_n^{(3)}\|_{2\to\infty} \nonumber \\
			&=& O_p(1) \times \|R_{\widehat{U}} - R_{UW}P_n^{(3)}\|_{2\to\infty} \nonumber \\
			&=& O_{p}\left(  (n\rho_{n})^{-1/4}\log^{\frac{7}{4}} n \right).
\end{eqnarray}

\vspace{.1in}

Third term of Equation \eqref{eq:4ineqs} is $\|AVD^{-1}WR_{UW} - AVD^{-1}WR_{UW}^{*}P_n^{(2)}\|$. Recall Proposition \ref{prop:converge}, Theorem \ref{thm:varimax1}, there is $P_n^{(2)} \in \mathcal{P}(k)$, s.t. for any $\delta>0$,
$$ \|R_{UW} - R_{UW}^{*}P_n^{(2)}\|_{2\to\infty} = O_{p}(n^{\delta / 2 - 1/4}). $$

Therefore,
\begin{eqnarray} \label{eq:cc}
\|AVD^{-1}WR_{UW} - AVD^{-1}WR^{*}_{UW}P_n^{(2)}\| &\leq& \|A\|\|V\|\|D^{-1}W\|\sqrt{k}\|R_{UW} - R_{UW}^{*}P_n^{(2)}\|_{2\to\infty} \nonumber \\
			&=& O_p(1) \times \|R_{UW} - R_{UW}^{*}P_n^{(2)}\|_{2\to\infty} \nonumber \\
			&=& O_{p}(n^{\delta / 2 - 1/4}).
\end{eqnarray}

\vspace{.1in}

Fourth term of Equation \eqref{eq:4ineqs} is $\|AVD^{-1}WR_{UW}^{*} - \A VD^{-1}WR_{UW}^{*}P_n^{(1)}\|$. Reusing Lemma \ref{lemma:Chung}, we have
\begin{eqnarray} \label{eq:dd}
\|AVD^{-1}WR_{UW}^{*} - \A VD^{-1}WR_{UW}^{*}\| &=& \|AVD^{-1} - \A VD^{-1}\| \nonumber \\
									&\leq& \|A - \A\| \|V\| \| D^{-1}\| \nonumber  \\
									&=& O_p((n\rho_n)^{-\frac{1}{2}} \log^{\frac{3}{2}} n).
\end{eqnarray}

\vspace{.2in}

Plugging \eqref{eq:aa}, \eqref{eq:bb}, \eqref{eq:cc}, \eqref{eq:dd} into \eqref{eq:4ineqs} arrives our combined bound: 

\begin{eqnarray}
 && \|J_{n}(A\widehat{V}\widehat{D}^{-1}R_{\widehat{U}} - \A VD^{-1}\widetilde{R}_{U}P_n)\|_{2\to\infty}  \nonumber \\
&=& O_p( \sqrt{n} \times ( (n\rho_{n})^{-1/2}\log^{\frac{5}{2}} n +  (n\rho_{n})^{-1/4}\log^{\frac{7}{4}} n  +  n^{\delta / 2 - 1/4} +  (n\rho_n)^{-1/2} \log^{\frac{3}{2}} n )) \nonumber  \\
&=& O_p\left(n^{\delta / 2 + 1/4} +  (n\rho_{n})^{-1/4}n^{1/2}\log^{\frac{7}{4}} n \right).
\end{eqnarray}


\vspace{.1in}

\end{proof}

\subsubsection{Proof of Proposition \ref{prop:uncenter2}} \label{proof:uncenter2}

\begin{proof}
With Lemma \ref{lemma:twonorm} and Proposition \ref{prop:svd}, 
\begin{eqnarray*}
&& \|J_{n}(\sqrt{n}\A VD^{-1}\widetilde{R}_{U} - Z)\|_{2\to\infty} \\
&\leq& \sqrt{n}\|\sqrt{n}\A VD^{-1}\widetilde{R}_{U} - Z\| \\
&=& \sqrt{n}\|\sqrt{n}ZBY^{T} VD^{-1}\widetilde{R}_{U} - Z\| \\
&=& \sqrt{n}\|ZB(Y^{T}\widetilde{Y}/d)\widetilde{\Sigma}_{Y}^{-1/2} \widetilde{R}_{V}^T\widetilde{D}^{-1}\widetilde{R}_{U} - Z\| \\
&=& \sqrt{n}\|ZB(Y^{T}\widetilde{Y}/d)\widetilde{\Sigma}_{Y}^{-1} B^{-1} \widetilde{\Sigma}_{Z}^{-1/2}  - Z\| \\
&\leq& \sqrt{n}\|Z\|(\|B(Y^{T}\widetilde{Y}/d)\widetilde{\Sigma}_{Y}^{-1} B^{-1} \widetilde{\Sigma}_{Z}^{-1/2} - I \|) \\
&\leq& \sqrt{n}\|Z\| ( \|B(Y^{T}\widetilde{Y}/d)\widetilde{\Sigma}_{Y}^{-1} B^{-1} \widetilde{\Sigma}_{Z}^{-1/2} - B(Y^{T}\widetilde{Y}/d)\widetilde{\Sigma}_{Y}^{-1} B^{-1}\| + \| B(Y^{T}\widetilde{Y}/d)\widetilde{\Sigma}_{Y}^{-1} B^{-1} - I \| ) \\
&\leq& \sqrt{n}\|Z\| ( \|B\|\|B^{-1}\|\|Y^{T}\widetilde{Y}/d\| \|\widetilde{\Sigma}_{Y}^{-1}\| \|\widetilde{\Sigma}_{Z}^{-1/2}-I \| + \|B\| \|B^{-1}\| \|(Y^{T}\widetilde{Y}/d)\widetilde{\Sigma}_{Y}^{-1} - I \| ). \\
\end{eqnarray*}

By Lemma \ref{lemma:maxZ}, $\|Z\| \leq \sqrt{nk}\max |Z_{ij}| = O_p(\sqrt{n}\log n)$. Conditions in main theorem statement imply $\|B\| \|B^{-1}\| = O_p(1)$. Using LLN results (similar to proofs in Proposition \ref{prop:clt}) there are $\|\widetilde{\Sigma}_{Y}^{-1}\| = O_p(1), \|\widetilde{\Sigma}_{Z}^{-1/2}-I \| = O_p(1/\sqrt{n})$. \\

Notice that the $(i, j)$ entry of $\bar{Y}^{T}\widetilde{Y}/d$ is $\frac{1}{d}\widehat{\mu}_{Y}[i]\sum_{q=1}^{n}(Y_{qj}-\widehat{\mu}_{Y}[j]) = 0$. By LLN
$$ \|Y^{T}\widetilde{Y}/d \| \leq  \|\widetilde{Y}^{T}\widetilde{Y}/d \| +  \|\bar{Y}^{T}\widetilde{Y}/d \| = O_p(1) $$
and 
$$ (Y^{T}\widetilde{Y}/d)\widetilde{\Sigma}_{Y}^{-1} - I = (\bar{Y}^{T}\widetilde{Y}/d)\widetilde{\Sigma}_{Y}^{-1},$$
is the zero matrix. Summarize these results and simplify the bounds give the desired conclusion,
$$ \|J_{n}(\sqrt{n}\A VD^{-1}\widetilde{R}_{U} - Z)\|_{2\to\infty} = O_p(\sqrt{n}\log n). $$

\end{proof}

\section{Technical Proofs} \label{sec:proof}


\subsubsection*{Proof of Lemma \ref{lemma:Chung}}
This part of proof needs a matrix concentration bound for sub-exponential random variables. Here we cite an existing result shown below.

\begin{lemma}[\cite{tropp2012user}] \label{lemma:Chung2}
Let $X_{1}, X_{2}, ..., X_{n}$ be independent random $N \times N$ self-adjoint matrices. Assume that $\mathbb{E}(X_{i}) = 0$ for all $i$, and $\mathbb{E}(X_{i}^{p})\preceq \frac{p!}{2}R^{p-2}A_{i}^{2}$ for $p\geq 2$. Compute the variance parameter 
$$ \sigma^{2} := \|\sum_{k}A_{k}^{2}\|.$$
Then for any $t > 0$, 
\begin{equation}
\pr(\|\sum_{i=1}^n X_{i}\| \geq t) \leq n \times \exp(-\frac{t^{2}}{2\sigma^{2}+2Rt}).
\end{equation}
\end{lemma}

Now we make use of Lemma \ref{lemma:Chung2} to prove Lemma \ref{lemma:Chung}.

\begin{proof} Let $E^{i,j}$ be the $(n+d)\times(n+d)$ matrix with 1 in the $(i,j)$ and $(j,i)$ entries and 0 elsewhere. $\gamma_{ij}, \gamma$ are defined in Equation \eqref{eq:gammadef} (for simplicity we ignore the $(n)$-superscripts). To utilize Lemma \ref{lemma:Chung2}, we express $\widetilde{A}_{sym} - \widetilde{\A}_{sym}$ as the sum of matrices,

$$ Y_{i,n+j} = (A_{ij}-\A_{ij})E^{i,n+j}, i = 1,...,n, j = 1,...,d.$$

\noindent Noice that 
$$ \|\widetilde{A}_{sym} - \widetilde{\A}_{sym}\| = \|\sum_{i = 1}^{n}\sum_{j=1}^{d}Y_{i,n+j}\|, $$

\noindent and $\mathbb{E}(Y_{i,n+j}) = 0$. Moreover,
$$ (E^{i,n+j})^{p} = E^{i,i} + E^{n+j,n+j}, p = 2,4,..., $$
$$ (E^{i,n+j})^{p} = E^{i,n+j}, p = 3,5,7,..., $$

\noindent and $\mathbb{E}[(A_{ij}-\A_{ij})^{p}] \leq \gamma_{ij}^{p}p! \leq \gamma^{p}p!$, for $\forall i,j, p \geq 2$. These relations indicate
\begin{equation}
\mathbb{E}(Y_{i,n+j}^{p}) \preceq \frac{p!}{2} \cdot \gamma_{ij}^{p-2} \cdot (\frac{\gamma_{ij}^{2}}{2}(E^{i,i} + E^{n+j,n+j})) \preceq \frac{p!}{2} \cdot \gamma^{p-2} \cdot (\frac{\gamma^{2}}{2}(E^{i,i} + E^{n+j,n+j})), \forall p\geq 2.
\end{equation}

\noindent We can treat $A_i$'s in Lemma \ref{lemma:Chung2} as $\frac{\gamma^{2}}{2}(E^{i,i} + E^{n+j,n+j})$ in our scenario. By Assumption \ref{assumption:atail},
\begin{eqnarray*}
\sigma^{2} &=& \frac{\gamma^2}{2}\|\sum_{i = 1}^{n}\sum_{j=1}^{d}(E^{i,i} + E^{n+j,n+j})\| \\
	         &\leq& \frac{\bar\rho_n}{4}\|\sum_{i = 1}^{n}\sum_{j=1}^{d}(E^{i,i} + E^{n+j,n+j})\| \\
	         &=& \frac{\bar\rho_n}{4}\| \sum_{i=1}^{n}[\sum_{j=1}^{d}E^{i,i}] + \sum_{j=1}^{d}[\sum_{i=1}^{n}E^{n+j,n+j}]\| \\
	         &\leq& \frac{\bar\rho_n}{4}(\| \sum_{i=1}^{n}[\sum_{j=1}^{d}E^{i,i}]\| + \|\sum_{j=1}^{d}[\sum_{i=1}^{n}E^{n+j,n+j}]\|) \\
	         &=& \frac{(n+d)\bar\rho_n}{4}. \\
\end{eqnarray*}


\noindent Therefore, by Lemma \ref{lemma:Chung2}, the bound for $\|\widetilde{A}_{sym} - \widetilde{\A}_{sym}\|$ is obtained.
$$ \pr(\|\widetilde{A}_{sym} - \widetilde{\A}_{sym}\| \geq t) \leq (n+d)\exp(-\frac{t^{2}}{\frac{(n+d)\bar\rho_n}{4} +2\gamma t}). $$

\noindent With Assumption 3 and Equation \eqref{eq:rhos} this also implies,
$$ \|\widetilde{A}_{sym} - \widetilde{\A}_{sym}\| = O((n\rho_n\log^3 n)^{\frac{1}{2}}). $$ 

\end{proof}

\noindent Before we show the proof of Lemma \ref{lemma:Cape}, we illustrate the following lemma that shows important property of matrix with special structure and could be utilized to convert bounds of eigenvectors' perturbation of symmeticed matrices to original adjacency matrices'.
\begin{lemma} \label{lemma:block}
Suppose $M = \begin{pmatrix} M_1 & M_2 \\ M_2 & M_1 \end{pmatrix}$ is a blockwise symmetric matrix ($M_1, M_2 \in \R^{k \times k}$). Let $M = U_M D_M V_{M}^{T}$ be $M$'s singular vector decomposition. Then $N = U_M V_{M}^{T}$ has the same blockwise symmetric structure: $N = \begin{pmatrix} N_1 & N_2 \\ N_2 & N_1 \end{pmatrix}$.
\end{lemma}

\begin{proof}
Let $M_1 + M_2 = S_1\Sigma_1 T_{1}^{T}, M_1 - M_2 = S_2\Sigma_2 T_{2}^{T}$ be the singular decompositions of them. Then
$$ M_1 = \frac{1}{2}(S_1\Sigma_1 T_{1}^{T} + S_2\Sigma_2 T_{2}^{T}), M_2 = \frac{1}{2}(S_1\Sigma_1 T_{1}^{T} - S_2\Sigma_2 T_{2}^{T}). $$

\noindent Plug in the equations, 
\begin{eqnarray*}
M &=& \begin{pmatrix} M_1 & M_2 \\ M_2 & M_1 \end{pmatrix} \\
    &=& \frac{1}{2}\begin{pmatrix} S_1\Sigma_1 T_{1}^{T} + S_2\Sigma_2 T_{2}^{T} & S_1\Sigma_1 T_{1}^{T} - S_2\Sigma_2 T_{2}^{T} \\ S_1\Sigma_1 T_{1}^{T} - S_2\Sigma_2 T_{2}^{T} & S_1\Sigma_1 T_{1}^{T} + S_2\Sigma_2 T_{2}^{T} \end{pmatrix} \\
    &=& \begin{pmatrix} \frac{\sqrt{2}}{2}S_1 & \frac{\sqrt{2}}{2}S_2 \\ \frac{\sqrt{2}}{2}S_1 & -\frac{\sqrt{2}}{2}S_2 \end{pmatrix} \begin{pmatrix} \Sigma_1 &  \\  & \Sigma_2 \end{pmatrix} \begin{pmatrix} \frac{\sqrt{2}}{2}T_1 & \frac{\sqrt{2}}{2}T_2 \\ \frac{\sqrt{2}}{2}T_1 & -\frac{\sqrt{2}}{2}T_2 \end{pmatrix}^{T}. 
\end{eqnarray*}

\noindent Then $U_M = \frac{\sqrt{2}}{2}\begin{pmatrix} S_1 & S_2 \\ S_1 & -S_2 \end{pmatrix}, V_M = \frac{\sqrt{2}}{2}\begin{pmatrix} T_1 & T_2 \\ T_1 & -T_2 \end{pmatrix}, D_M = \begin{pmatrix} \Sigma_1 &  \\  & \Sigma_2 \end{pmatrix}$. We prove the result by the following observation.
$$ U_MV_{M}^{T} = \frac{1}{2} \begin{pmatrix}  S_1T_{1}^{T} + S_2T_{2}^{T} & S_1T_{1}^{T} - S_2T_{2}^{T} \\ S_1T_{1}^{T} - S_2T_{2}^{T} & S_1T_{1}^{T} + S_2T_{2}^{T} \end{pmatrix}. $$
\end{proof}

\noindent Lemma \ref{lemma:Chung} obtains the bound of spectral norm of $\widetilde{A}_{sym} - \widetilde{\A}_{sym}$. Next lemma makes use of this result to show the bound for distance between eigen-spaces of $\widetilde{A}_{sym}$ and $\widetilde{\A}_{sym}$. Some theoretical results from \cite{cape2019signal} is borrowed to show row-wise bounds.

\subsubsection*{Proof of Lemma \ref{lemma:Cape}}
\begin{proof}
Let $\mu_{*}=\mu_{r}\1^{T}_{d} + \1_{n}\mu_{c}^{T} - \mu\1_{n}\1_{d}^{T}, \widehat{\mu}_{*}=\widehat{\mu}_{r}\1^{T}_{d} + \1_{n}\widehat{\mu}_{c}^{T} - \widehat{\mu}\1_{n}\1_{d}^{T}, $
we symmetrize centered adjacent matrix as before:
$\widetilde{A}_{sym} = \begin{pmatrix}
0 & \widetilde{A} \\
\widetilde{A}^{T} & 0
\end{pmatrix}$,
$\widetilde{\A}_{sym} = \begin{pmatrix}
0 & \widetilde{\A} \mbox{ } \\
\widetilde{\A}^{T} & 0
\end{pmatrix}$.
\begin{eqnarray*}
\widetilde{A}_{sym} - \widetilde{\A}_{sym} &=& \begin{pmatrix} 0 & \widetilde{A} \\ \widetilde{A}^{T} & 0 \end{pmatrix} - \begin{pmatrix} 0 & \widetilde{\A}\mbox{ } \\ \widetilde{\A}^{T} & 0 \end{pmatrix} \\
			  &=& \begin{pmatrix} 0 & A-\widehat{\mu}_{*} \\ (A-\widehat{\mu}_{*})^{T} & 0 \end{pmatrix} - \begin{pmatrix} 0 & \A-\mu_{*} \\ (\A-\mu_{*})^{T} & 0 \end{pmatrix} \\
			  &=& \begin{pmatrix} 0 & A-\widehat{\mu}_{*} \\ (A-\widehat{\mu}_{*})^{T} & 0 \end{pmatrix} - \begin{pmatrix} 0 & \A-\widehat{\mu}_{*} \\ (\A-\widehat{\mu}_{*})^{T} & 0 \end{pmatrix} + \\
&& \quad \quad \quad \quad\quad \quad \quad \quad\quad \quad \begin{pmatrix} 0 & \A-\widehat{\mu}_{*} \\ (\A-\widehat{\mu}_{*})^{T} & 0 \end{pmatrix} -  \begin{pmatrix} 0 & \A-\mu_{*} \\ (\A-\mu_{*})^{T} & 0 \end{pmatrix} \\
		          &:=& A_{1} - A_{2} + A_{2} - A_{3}.
\end{eqnarray*}

From Proposition \ref{prop:svd} we know $\A-\mu_{*}$ is of rank $k$. Suppose the eigen-decomposition of $A_{3}$ is $U_{3}D_{3}U_{3}^{T}$. Then $U_{3} \in \R^{(n+d)\times 2k}$. $D_{3}$ is diagonal matrix with $2k$ non-zero elements. Let $U_{i}\in \R^{(n+d)\times 2k }$ be $A_{i}$'s eigenvectors corresponding to $A_{i}$'s $k$ largest and $k$ smallest eigenvalues and $D_{i}$ being diagonal matrix contains these eigenvalues, $i = 1,2$.  

\vspace{.3in}
Define
$$ W_{(3\to1)} := \underset{W_{0}\in \oo(2k)}{\arg \min}\|U_{1} - U_{3}W_{0}\|_{2\to\infty},  $$
$$ W_{(2\to1)} := \underset{W_{0}\in \oo(2k)}{\arg \min}\|U_{1} - U_{2}W_{0}\|_{2\to\infty},  $$
$$ W_{(3\to2)} := \underset{W_{0}\in \oo(2k)}{\arg \min}\|U_{2} - U_{3}W_{0}\|_{2\to\infty},  $$
then 
\begin{eqnarray}  \label{eq:UW}
\|U_{1} - U_{3}W_{(3\to1)}\|_{2\to\infty} &\leq& \|U_{1} - U_{3}W_{(3\to2)}W_{(2\to1)}\|_{2\to\infty} \nonumber \\
				  &=& \|U_{1} - U_{2}W_{(2\to1)} +U_{2}W_{(2\to1)} - U_{3}W_{(3\to2)}W_{(2\to1)}\|_{2\to\infty} \nonumber \\
				  &\leq& \|U_{1} - U_{2}W_{(2\to1)}\|_{2\to\infty} +\|U_{2}W_{(2\to1)} - U_{3}W_{(3\to2)}W_{(2\to1)}\|_{2\to\infty} \nonumber \\
				  &=& \|U_{1} - U_{2}W_{(2\to1)}\|_{2\to\infty} +\|U_{2} - U_{3}W_{(3\to2)}\|_{2\to\infty}.
\end{eqnarray}

In the following steps, we deal with the two terms in (\ref{eq:UW}) separately.

We bound $\|U_{1} - U_{2}W_{(2\to1)}\|_{2\to\infty}$ by using Theorem 1 in \cite{cape2019signal}. There are four conditions of Cape's Theorem that should be evaluated.  The first two conditions are followed trivially from the statements of Theorem \ref{thm:main}. Let $\widetilde \sigma_{\max}$ and $\widetilde \sigma_{\min}$ represent the largest and $k$th largest singular values of  $\widetilde \A$. Then with Equation \eqref{eq:svdA}, it is obvious to show
\begin{equation} \label{eq:tildesigma}
 \widetilde \sigma_{\min} \geq C_1 \Delta_n = C_1 n \rho_n, \quad \widetilde \sigma_{\max}/ \widetilde \sigma_{\min} \leq C_2,
\end{equation} 
for some positive constants $C_1, C_2$.


\vspace{.1in}
The following part checks the fourth condition stated below. After confirming the Cape's fourth condition, the proof will address the third condition. 

\vspace{.1in}

\textbf{Cape's 4th Condition:} Write $E_{1} = A_{1} - A_{2}$. There exist constants $C_{E_{1}}, v>0, \nu>0, \xi>1$, such that for all integers $1\leq s\leq s(n):= \lceil \log n/\log (n\bar\rho_n)\rceil$, for each fixed standard basis vector $e_{i}$ and any fixed unit vector $u$, with probability at least $1 - \exp(-\nu \log^\xi n)$ (provided $n \geq n_0(C_E,\nu,\xi))$),
\begin{equation}
|\langle E_{1}^{s}u, e_{i}\rangle| \leq C_{E_{1}}^{s}(n\bar\rho_{n})^{s/2}(\log n)^{s\xi}\|u\|_{\infty}.
\end{equation}

\vspace{.1in}

Using an argument in Lemma 7.10 of \cite{erdHos2013spectral}, \cite{mao2017estimating} shows that the following Upper Bound Condition is sufficient for Cape's 4th Condition. That appears as Lemma 5.5 of \cite{mao2017estimating}. The key idea to show Cape's 4th assumption is applying the inequality of Upper Bound Condition to upper bound the number of non-zero tems in the summation via a multigraph construction for paths counting. The difference is that our main theorem allows sub-exponential random variables, which is possibly unbounded. Thus, the only thing that needs to be checked is the following upper bound condition, with the order of magnitude $m$ ranges from 2 to the number of vertices in constructed multigraph.

\vspace{.1in}
\textbf{Upper Bound Condition} Let $H = \frac{A - \A}{\sqrt{n\bar\rho_{n}}}$ and $H_{ij}$ represents its element on $(i,j)$th entry. Then there exists a positive constant $C_{E_1}$, such that eventually in $n$,
\begin{equation} \label{eq:ubound}
\mathbb{E}(|H_{ij}|^{m}) \leq \frac{C_{E_1}}{n}, \quad \forall \mbox{ } 2 \leq m \leq \log^\xi n.
\end{equation}

For any positive even number $m \leq \log^\xi n$. If $\bar \rho_n \geq (m-1)!\bar\rho^{m/2}$, by Assumption \ref{assumption:atail}, 
$$ \mathbb{E}(|H_{ij}|^{m}) = \mathbb{E}(H_{ij}^{m})  \leq  \frac{\bar\rho_n}{(n\bar\rho_n)^{m/2}} , $$
therefore, choosing $C_{E_1}=1$, Equation \eqref{eq:ubound} is implied by $n\bar \rho_n \geq 1$, which is true under the main theorem's assumptions. \\

If $\bar \rho_n \leq (m-1)!(\bar\rho_n)^{m/2}$, by Assumption \ref{assumption:atail},
$$ \mathbb{E}(|H_{ij}|^{m}) = \mathbb{E}(H_{ij}^{m}) \leq \frac{(m-1)! \bar\rho_n^{m/2}}{(n\bar\rho_n)^{m/2}} = \frac{(m-1)!}{n^{m/2}}. $$

The target boils down to $(m-1)! \leq C_{E_1} n^{m/2 - 1}$. Notice that for all positive integers $N$, there is $N! \leq N^{N+1/2}\exp\{-N+1\}$. Then,
\begin{eqnarray}
(m-1)! \leq n^{m/2 - 1} &\Leftarrow& (m-1)^{m-1/2}\exp\{-m+2\} \leq  n^{m/2-1} C_{E_1}  \nonumber \\
				 &\Leftrightarrow& (m-\frac{1}{2})\log (m-1) +2 - m \leq (\frac{m}{2} - 1) \log n + \log (C_{E_1}) \nonumber \\
\mbox{(since $\log n^{\xi} \geq m$)}&\Leftarrow& (m-\frac{1}{2})\log (m-1) +2 - m \leq (\frac{m}{2} - 1)m^{1/\xi} + \log (C_{E_1}).  \label{eq:upperbound} 
\end{eqnarray}

\vspace{.1in}

Since $1/\xi > 0$, there exists a positive integer $M$, such that $(\frac{m}{2} - 1)m^{1/\xi} > (m-\frac{1}{2})\log (m-1) +2 - m$ for all integer $m > M$. Choose $C_{E_1}$ such that $\log (C_{E_1}) > (m-\frac{1}{2})\log (m-1) +2 - m$ for all integer $2 \leq m \leq M$. Then Equation \eqref{eq:upperbound} is proved. \\


Hence the upper bound condition holds for even $m$. For odd number $m \geq 3$, by Cauchy-Schwartz inequality,
$$ (\mathbb{E}(|H_{ij}|^{m}))^2 \leq \mathbb{E}(|H_{ij}|^{m-1}) \mathbb{E}(|H_{ij}|^{m+1}) \leq \frac{1}{n^2}. $$

Thus the upper bound condition holds for all integer $m\geq 2$ and therefore Cape's fourth condition is valid. 

We claim that Cape's third condition could be relaxed from $\|E\| = O_p((n\rho_n)^{\frac{1}{2}})$ to $\|E\| = O_p((n\rho_n\log^3 n)^{\frac{1}{2}})$ as inferred from Lemma \ref{lemma:Chung}, with only slight modifications in Cape's converge rate result. In the proof of Theorem 1 of \cite{cape2019signal}, the bound of LHS of (5) comes from three quantities: $\|E\widehat{U}\widehat{\Lambda}^{-1}\|_{2\to\infty}, \|R^{(1)}\|_{2\to\infty}, \|R^{(2)}_{W}\|_{2\to\infty}$ (these three terms' notations are from \cite{cape2019signal}). Our relaxation of $\|E\|$ adds an extra $\log^{\frac{3}{2}} n$ term to $\|R^{(1)}\|_{2\to\infty}$'s bound, while the third term remains the same because of Cape's 4th assumption (we have already checked). Therefore by Theorem 1 of \cite{cape2019signal}, we arrives the following conclusion. 
\begin{eqnarray} \label{eq:U1U2W2}
\|U_{1} - U_{2}W_{(2\to1)}\|_{2\to\infty} &=& O_{p}\left( ( (n\bar \rho_{n})^{-1/2}\log^{\xi} n + (n\rho_{n})^{-1/2}\log^{\frac{3}{2}} n) \times \|U_2\|_{2\to\infty} \right) \nonumber \\
&=& O_{p}\left( (n\rho_{n})^{-1/2}\log^{\frac{3}{2}} n \times \|U_2\|_{2\to\infty} \right).  \label{eq:U1U2W2}
\end{eqnarray}

\vspace{.1in}
To bound $\|U_{2} - U_{3}W_{(3\to2)}\|_{2\to\infty}$, we employ Theorem 4.2 in \cite{cape2019two}. 

\begin{theorem}\label{lemma:TwoInfinity} [Theorem 4.2 in \cite{cape2019two}]
Suppose the diagonal elements of $D_{3}$ are sorted in descending order. If $|D_{3}[k]| > 4\|E_{2}\|_{\infty}$, where $E_{2} = A_{2} - A_{3}$, $D_{3}[j]$ is the $j$-th diagonal element of $D_{3}$. Then there exists $W_{3} \in \oo(k)$ such that
\begin{equation}
\|U_{2} - U_{3}W_{(3\to2)}\|_{\max} \leq 14(\frac{\|E_{2}\|_{\infty}}{|D_{3}[k]|})\|U_{3}\|_{2\to \infty}.
\end{equation}
\end{theorem}


Before applying Theorem \ref{lemma:TwoInfinity}, we should check its assumptions. Reuse the notation for the singular values of $\widetilde \A$ as in Equation \eqref{eq:tildesigma}. Notice $\widetilde{\sigma}_{\min} \succeq c_1n\rho_n$ and $\mu_{*} = \mu_{r}\1_{d}^{T} + \1_{n}\mu_{c}^{T} - \mu\1_{n}\1_{d}^{T}$. Denotes $\mu_{r} = \A\1_{d}/d := T/d$. Similarly define $\widehat{T} := A\1_{d}$. By Hoeffding's Concentration Inequality, 
$$ \pr(|\widehat{T}_{i} - T_{i}| \geq a) \leq 2\exp\{ - \frac{a^{2}}{2\sum_{j=1}^{d}var(A_{ij})}\} \leq 2\exp\{ - \frac{a^{2}}{d\bar\rho_n}\}. $$

Therefore, $\widehat{T}_{i} - T_{i} = O_{p}((n\bar\rho_n)^{\frac{1}{2}})$. The same bound applies to $\|\widehat{\mu}_{c} - \mu_{c}\|, \|\widehat{\mu} - \mu\|$. These imply that any entry of $E_{2}$ could be bounded by $O_{p}(\sqrt{n\bar\rho_n}/n)$. Thus $\|E_{2}\|_{\infty} = O_{p}((n\rho_n\log^2 n)^{\frac{1}{2}})$. Then with large enough $n$, there must be $|D_{3}[k]| > 4\|E_{2}\|_{\infty}$.  

\vspace{.1in}
With Theorem \ref{lemma:TwoInfinity},
\begin{equation} \label{eq:twoinf}
\|U_{2} - U_{3}W_{(3\to2)}\|_{2\to\infty} \leq \sqrt{k}\|U_{2} - U_{3}W_{(3\to2)}\|_{\max}  = O_{p}((n\rho_n)^{-\frac{1}{2}}\log n \|U_{3}\|_{2\to\infty}).
\end{equation}

Combining \eqref{eq:UW}, \eqref{eq:U1U2W2}, \eqref{eq:twoinf} gives
\begin{eqnarray} \label{eq:eigenbound}
&& \|U_{1} - U_{3}W_{(3\to1)}\|_{2\to\infty} \nonumber \\
 &\leq& \|U_{1} - U_{2}W_{(2\to1)}\|_{2\to\infty} +\|U_{2} - U_{3}W_{(3\to2)}\|_{2\to\infty} \nonumber \\
&=& O_{p}\left( (n\rho_{n})^{-1/2}\log^{\frac{3}{2}} n \times \|U_2\|_{2\to\infty} + (n\rho_n)^{-\frac{1}{2}}\log n \|U_{3}\|_{2\to\infty}\right).  \nonumber \\
&\preceq& O_{p}\left( (n\rho_{n})^{-1/2}\log^{\frac{3}{2}} n \times (\|U_2 - U_3W_{(3\to 2)}\|_{2\to\infty} + \|U_3\|_{2\to\infty}) + (n\rho_n)^{-\frac{1}{2}}\log n \|U_{3}\|_{2\to\infty}\right)  \nonumber \\
&=& O_{p}\left( (n\rho_{n})^{-1/2}\log^{\frac{3}{2}} n \times \|U_3\|_{2\to\infty} \right).
\end{eqnarray}

The next step is to convert the symmetrized adjacency matrices' eigenvectors' perturbation bound to that of original adjacency matrices. Suppose the $k$-rank singular value decompositions (only retains the top $k$ singular values and corresponding eigenvectors) of $A-\widehat{\mu}_{*}, \A-\widehat{\mu}_{*}$ and $\A-\mu_{*}$ are $F_{1}\Lambda_{1}L_{1}^{T}, F_{2}\Lambda_{2}L_{2}^{T}$, $F_{3}\Lambda_{3}L_{3}^{T}$ repectively. Then, for $i=1,2,3$ there must be 
$$ U_{i} = \frac{1}{\sqrt{2}}\begin{pmatrix}  F_{i} & -F_{i}\\ L_{i} & L_{i} \end{pmatrix}.$$  

Suppose
$$ W_{(3\to1)} = \begin{pmatrix}  W^{11} & W^{12} \\ W^{21} & W^{22} \end{pmatrix}$$
with each sub-matrix having $k \times k$ dimension. 

\vspace{.1in}

Arguments in \cite{cape2019signal} (proof of Theorem 1) implies that if we have singular value decomposition $U_3^TU_1 = U_oD_oV_o^T$, then $W_{(3\to1)} = U_oV_o^T$. It is trivial to see that $U_3^TU_1$ has special block-wise structure
$$ U_{3}^{T}U_1 = \frac{1}{2}\begin{pmatrix} F_{3}^{T}F_1 + L_{3}^{T}L_1 & -F_{3}^{T}F_1 + L_{3}^{T}L_1 \\ -F_{3}^{T}F_1 + L_{3}^{T}L_1 &F_{3}^{T}F_1 + L_{3}^{T}L_1\end{pmatrix}. $$ 

Thus Lemma \ref{lemma:block} indicates 
\begin{equation} \label{eq:W1122} 
W^{11} = W^{22}, \quad W^{12} = W^{21}. 
\end{equation}

Notice $W_{(3\to1)}$ is orthogonal matrix, therefore 
\begin{eqnarray*}
W_{(3\to1)}W_{(3\to1)}^T=I &\Rightarrow& \begin{pmatrix}  W^{11} & W^{12} \\ W^{21} & W^{22} \end{pmatrix} \begin{pmatrix}  W^{11} & W^{12} \\ W^{21} & W^{22} \end{pmatrix}^T = I \\
			&\Rightarrow& \begin{pmatrix}  W^{11} & W^{12} \\ W^{12} & W^{11} \end{pmatrix} \begin{pmatrix}  W^{11} & W^{12} \\ W^{12} & W^{11} \end{pmatrix}^T = I \\
			&\Rightarrow& W^{11}W^{11T} + W^{12}W^{12T} = I, \quad W^{11}W^{12T} + W^{12}W^{11T} = 0, \\
(\mbox{Equation  } \eqref{eq:W1122})	&\Rightarrow& (W^{21}-W^{22})(W^{21}-W^{22})^T = I, \quad (W^{11}-W^{12})(W^{11}-W^{12})^T = I.
\end{eqnarray*}
These indicate $W^{21} - W^{22}$ and $W^{11} - W^{12}$ are both orthogonal matrices. 

\vspace{.1in}

Notice $ \|U\|_{2\to\infty} = O_p(\log n / \sqrt{n})$ from Equation \eqref{eq:maxU}, similarly there is same upper bound for $\|V\|_{2\to\infty}$. Therefore
%
%

%
\begin{eqnarray*}
\underset{W\in \oo(k)}{\inf}\|\widehat{U}-UW\|_{2\to\infty}  &=&  \underset{R\in \oo(k)}{\inf}\|F_{1}-F_{3}R\|_{2\to\infty} \\
&\leq& \max\{\|F_{1}-F_{3}(W^{11}-W^{21})\|_{2\to\infty}, \|F_{1}-F_{3}(W^{12}-W^{22})\|_{2\to\infty}\}  \\
&\leq& \|\begin{pmatrix} F_{1} &, -F_{1} \end{pmatrix} -  \begin{pmatrix} F_{3} &, -F_{3} \end{pmatrix}W_{(3\to1)}\|_{2\to\infty} \\
&\leq&  \|U_{1} - U_{3}W_{(3\to1)}\|_{2\to\infty}  \\
&=& O_{p}\left(  (n\rho_{n})^{-1/2}\log^{\frac{3}{2}} n \times \|U_3\|_{2\to\infty} \right) \\
&=& O_{p}\left(  (n\rho_{n})^{-1/2}\log^{\frac{3}{2}} n \times (\|U\|_{2\to\infty} + \|V\|_{2\to\infty}) \right) \\
&=& O_{p}\left(  (n\rho_{n})^{-1/2}n^{-1/2}\log^{\frac{5}{2}} n \right).
\end{eqnarray*}

\end{proof}

\subsubsection*{Proof for requisite results of Lemma \ref{lemma:SOC}}

The proof of Lemma \ref{lemma:SOC} is decomposed into  Lemmas \ref{lemma:singleTruncate},  \ref{lemma:seriesTruncate},  \ref{lemma:seriesTruncate2},  \ref{lemma:maxi}, and  \ref{lemma:uniform}. All of which are stated below. 

Lemmas \ref{lemma:singleTruncate}, \ref{lemma:seriesTruncate}, \ref{lemma:seriesTruncate2} bound the  tail behavior of the sample Varimax objective function. Lemmas \ref{lemma:maxi} and \ref{lemma:uniform}  bound the difference between the sample and population versions of the Varimax objective function, uniformly over the space of orthogonal matrices. Lemma \ref{lemma:SOC} puts these pieces together with the first and second order conditions for Varimax described in Section \ref{sec:FSOC} to show that the optimum of the sample Varimax objective function must be close to the optimum of the population Varimax function (modulo permutation and sign-flip).


\vspace{.2in}
The next few lines show the existence of moment generating function (MGF) of linear inner-product of sub-exponential random vectors. Recall $Z_{i} \in \R^{k}$ contains sub-exponential random variables. Following the notations in the proof of Lemma \ref{lemma:maxZ} (Equation \eqref{eq:lambda}) the tail property of $Z$ could be shown as 
\begin{equation} \label{eq:mgf}
\pr(Z_{ij} - \mathbb{E}Z_{ij} > t) \leq C_{0} \exp(-\lambda t),
\end{equation}
then for $\forall r \in \R^{k}$, by independency 
$$\mathbb{E}\exp(t\langle \widetilde{Z}_{i}, r\rangle) = \mathbb{E}\exp(t\sum_{j=1}^{k} \widetilde{Z}_{ij}r_{j}) = \Pi_{j=1}^{k} \mathbb{E} \exp(t \widetilde{Z}_{ij}r_{j}), $$
$\Rightarrow$ $\langle \widetilde{Z}_{i}, r\rangle$ also has MGF.  

\vspace{.2in}

\begin{lemma}\label{lemma:singleTruncate}
$\widetilde{Z}_{i} \in \R^{k}$ is $i$-th row of $\widetilde{Z}$. $r \in \R^{k}$ is arbitrary. $\lambda$ is defined in Equation \eqref{eq:mgf}. Let 
$$J_{i} = \langle \widetilde{Z}_{i}, r \rangle^{4} - \mathbb{E}[\langle \widetilde{Z}_{i}, r \rangle^{4}], $$
then for $\forall t > 0$, there exists a positive constant $C_1$ s.t.
\begin{equation}
\pr(J_{i} > t) \leq C_{1} \exp(-\lambda t^{\frac{1}{4}}).
\end{equation}
\end{lemma}
\begin{proof} Write $X_{i} = \langle \widetilde{Z}_{i}, r \rangle$. By Markov Inequality,
\begin{eqnarray*}
\pr(J_{i} > t) &=& \pr(X_{i}^{4} - \mathbb{E}(X_{i}^{4}) > t) \\
			&=& \pr(X_{i}^{4} > t + \mathbb{E}(X_{i}^{4})) \\
			&=& \pr(X_{i} - \mathbb{E}X_{i} > (t + \mathbb{E}(X_{i}^{4}))^{\frac{1}{4}} - \mathbb{E}X_{i}) \\
			&\leq& C_{0}' \exp(-\lambda((t+\mathbb{E}[X_{i}^{4}])^{\frac{1}{4}} - \mathbb{E}X_{i})) \\
			&\leq& C_{0}' \exp(\lambda\mathbb{E}X_{i})\exp(-\lambda t^{\frac{1}{4}}) \\
			&:=& C_{1} \exp(-\lambda t^{\frac{1}{4}}).
\end{eqnarray*}
\end{proof}

The following lemma makes use of Lemma \ref{lemma:singleTruncate} and gives bound to the sum of sequence $|\sum_{i=1}^{n}J_{i}|$.

 \begin{lemma}\label{lemma:seriesTruncate}
 With previous definitions, then for any $\delta > 0$, 
 $$ |\sum_{i=1}^{n}J_{i}| = O_{p}(n^{\frac{1}{2}+\delta}). $$
 \end{lemma}

\begin{proof}
For sequence $\alpha_{n} \uparrow \infty$. Define $\widetilde{J}_{i} = J_{i}\1(J_{i} \leq \alpha_{n})$. Then $\widetilde{J}_{i}$'s are independent bounded random variables. Let $\mathfrak{A} = \{\bigcap_{i=1}^{n}\{J_{i} = \widetilde{J}_{i}\}\}$ and $\mathfrak{B} = \{|\sum_{i=1}^{n}J_{i}| > t\}$. Then
\begin{eqnarray}
\pr(\mathfrak{B}) &=& \pr(\mathfrak{B} \cap \mathfrak{A}) + \pr(\mathfrak{B} \cap \mathfrak{A}^{c})\nonumber \\
			 &\leq& \pr(\{|\sum_{i=1}^{n}\widetilde{J}_{i}| > t\}) + \pr(\mathfrak{A}^{c}) \label{eq:PAB}. 
\end{eqnarray}

Notice $\widetilde{J}_{i} \in [-c, \alpha_{n}]$ with $c = -\mathbb{E}(X_{1}^{4})$ is a bounded random variable. Therefore it is sub-gaussian with domain interval length $\sigma \leq \alpha_{n} + c \leq 2 \alpha_{n}$ for large enough $n$. By Hoeffding Concentration Inequality,
\begin{equation} \label{eq:truncateA}
\pr(\{|\sum_{i=1}^{n}\widetilde{J}_{i}| > t\}) \leq 2\exp(-\frac{2 t^{2}}{\sum_{i}\sigma^{2}}) \leq 2\exp(-\frac{t^{2}}{2n\alpha_{n}^{2}}).
\end{equation}

By Lemma \ref{lemma:singleTruncate} there is,
\begin{eqnarray} 
\pr(\mathfrak{A}^{c}) &=& \pr(\{\cap_{i}\{J_{i}\leq \alpha_{n}\}\}^{c}) \nonumber \\
				&=& \pr(\cup_{i}\{J_{i}>\alpha_{n}\}) \nonumber\\
				&\leq& n\pr(J_{i} > \alpha_{n}) \nonumber\\
				&\leq& nC_{2}\exp(-\lambda \alpha_{n}^{\frac{1}{4}}). \label{eq:truncateB}
\end{eqnarray}
 Plugging (\ref{eq:truncateA})(\ref{eq:truncateB}) in (\ref{eq:PAB}) and choosing $\varepsilon > \delta, t = n^{1/2+\varepsilon}, \alpha_{n} = n^{\delta}$ gives
 \begin{equation}
 |\sum_{i=1}^{n}J_{i}| = O_{p}(n^{\frac{1}{2}+\delta}).
 \end{equation}
 
 \end{proof}
 
 \vspace{.2in}
 
 \noindent Similar conclusion applies to second moment terms.
 \begin{lemma} \label{lemma:seriesTruncate2}
With same notations as Lemma \ref{lemma:singleTruncate}. Define $Y_{i} = \langle \widetilde{Z}_{i}, r \rangle^{2} - \mathbb{E}[\langle \widetilde{Z}_{i}, r \rangle^{2}]$, then
\begin{equation}
 |\sum_{i=1}^{n}Y_{i}| = O_{p}(n^{\frac{1}{2}+\delta}).
\end{equation}
 \end{lemma}

\begin{proof} This part employs the same strategy as the proof of Lemma \ref{lemma:seriesTruncate}. The only difference is the bound of (\ref{eq:truncateB}). But the dominating bound (\ref{eq:truncateA}) is the same. Thus we obtain the similar bound for $ |\sum_{i=1}^{n}Y_{i}| $. 
\end{proof}

\begin{lemma} \label{lemma:maxi}
Suppose $r_{1}, r_{2}, ... r_{n_{0}} \in \R^{k}$. Denote
$$\mathbb{J}_{\ell} = |\sum_{i=1}^{n} \langle \widetilde{Z}_{i}, r_{\ell} \rangle^{4} - n\mathbb{E}[\langle \widetilde{Z}_{i}, r_{\ell} \rangle^{4}]|. $$
Assume $n_{0} = an^{b}$ with some positive constant $a,b$. Then for any $\delta > 0$, 
\begin{equation}
\max_{\ell}\mathbb{J}_{\ell} =  O_{p}(n^{\frac{1}{2}+\delta}).
\end{equation}
Similarly if we define
$$ \mathbb{Y}_{\ell} = \langle \widetilde{Z}_{i}, r_{\ell} \rangle^{2} - \mathbb{E}[\langle \widetilde{Z}_{i}, r_{\ell} \rangle^{2}], $$
we have
\begin{equation}
\max_{\ell}\mathbb{Y}_{\ell} =  O_{p}(n^{\frac{1}{2}+\delta}).
\end{equation}
\end{lemma}

\begin{proof}
Take $\varepsilon > \delta, t = n^{1/2+\varepsilon}, \alpha_{n} = n^{\delta}$. Applying the same strategy in Lemma \ref{lemma:seriesTruncate} (Equation \eqref{eq:truncateA}, \eqref{eq:truncateB}) and basic probability rules,
\begin{eqnarray*}
\pr(\max_{\ell}\mathbb{J}_{\ell} > t) &\leq& 2n_{0}\exp(-\frac{t^{2}}{2n\alpha_{n}^{2}}) + n_{0}nC_{2}\exp(-\lambda \alpha_{n}^{\frac{1}{4}}) \\
						     &=& 2n_{0}\exp(-\frac{n^{1+2\varepsilon}}{2n^{1+2\delta}}) + n_{0}nC_{2}\exp(-\lambda n^{\delta/4})
\end{eqnarray*}
Since
$$ \log(2n_{0}) - \frac{1}{2}n^{2\varepsilon-2\delta} = \log2a + b\log n - n^{2\varepsilon-2\delta} \to -\infty, $$
$$ \log(C_{2}n_{0}n) - \lambda n^{\delta/4} = \log C_{2} \log a + (b+1)\log n - \lambda n^{\delta/4} \to -\infty, $$
they could be reduced to $\max_{\ell}\mathbb{J}_{\ell} =  O_{p}(n^{\frac{1}{2}+\delta})$. Similar approach gives $\max_{\ell}\mathbb{Y}_{\ell} =  O_{p}(n^{\frac{1}{2}+\delta})$. 
\end{proof}

\vspace{.2in}

\begin{lemma} \label{lemma:uniform}
Recall notations in equation \eqref{eq:Varimax}, Section 2.1 and Section \ref{def:model}. For readability we slightly abuse the notations and write $\widehat{V}(R) = v(R, \widetilde{Z})$. And
 \[ V(R) := \vv_{\widetilde U}(R) = \sum_{j=1}^k Var([\widetilde{Z}_i \widetilde U R]_j).\]
Then for $\forall \delta > 0$, there is a uniform bound between these two quantities,
\begin{equation}
\underset{ O \in \oo(k)}{\sup}|\widehat{V}(O) - V(O)| = O_{p}(n^{\delta - 1/2}).
\end{equation}
\end{lemma}

\begin{proof} This part of the proof adapts the covering balls strategy to give this uniform bound. Let $\R = \{R_{1}, R_{2}, ..., R_{N} \}$ be a $\varepsilon$-cover for orthogonal matrices. This means for $\forall O \in \mathcal{O}(k)$, there exists $\ell$ s.t. $d(O, R_{\ell}) < \varepsilon$ where $d(X,Y)$ is the $\sin\Theta$ distance. Let $D(\varepsilon, \mathcal{O}(k), d)$ be the $\varepsilon$-packing number. Using the notes in \cite{van2000applications} and Lemma 4.1 in \cite{pollard1990empirical} we have,
\begin{equation}
N \leq D(\varepsilon, \mathcal{O}(k), d) \leq D(\varepsilon, \mathcal{O}(k), d_{F}) \leq (\frac{6}{\varepsilon})^{k^{2}} := N_{0}.
\end{equation}

$d_{F}$ is Frobenius norm distance. The second inequality is true because for any $O_{1}, O_{2} \in \mathcal{O}(k)$,
$$ \|\sin (O_{1}, O_{2})\|^{2}_{F} \leq \underset{Q \in \mathcal{O}(k)}{\inf}\|O_{1}-O_{2}Q\|_{F}^{2} \leq \|O_{1}-O_{2}\|_{F}^{2} $$


Let $\varepsilon = 1/n $ then $N \leq N_{0} = (6n)^{k^{2}}$. For $\forall O \in \mathcal{O}(k)$, choose $R_{\ell}\in \R$ such that $d(O, R_{\ell}) < \varepsilon$. Then by Triangle Inequality,
\begin{equation} \label{eq:triangles}
|\widehat{V}(O) - V(O)| \leq |\widehat{V}(O) - \widehat{V}(R_{\ell})| + |\widehat{V}(R_{\ell}) - V(R_{\ell})| + |V(O) - V(R_{\ell})|.
\end{equation}

Lemma \ref{lemma:seriesTruncate}, \ref{lemma:maxi} indicates
\begin{equation}
|\widehat{V}(R_{\ell}) - V(R_{\ell})| \leq \underset{j}{\max} |\widehat{V}(R_{j}) - V(R_{j})| = O_{p}(n^{\delta - \frac{1}{2}}).
\end{equation}

Notice $\widehat{V}(R) = \sum_{j=1}^k\left(\frac{1}{n}\sum_{i=1}^n [\widetilde{Z}R]_{ij}^4 - \left(\frac{1}{n} \sum_{i=1}^n [\widetilde{Z}R]_{ij}^2\right)^2\right)$ and $\underset{ij}{\max}|\widetilde{Z}_{ij}| = O(\log n)$ (Lemma \ref{lemma:maxZ}). Also for $\forall O, R \in \mathcal{O}(k)$ such that $d(O, R) < \varepsilon$, there is fact that
$$d(O, R) \geq \frac{1}{\sqrt{2} }\|O-R\|_{F}, \forall O, R \in \oo(k), $$

then
\begin{eqnarray} 
|\sum_{ij}([\widetilde{Z}O]_{ij}^{4} - [\widetilde{Z}R]_{ij}^{4})| &=& | \sum_{ij}([\widetilde{Z}O]_{ij}^{2} + [\widetilde{Z}R]_{ij}^{2})([\widetilde{Z}O]_{ij}+[\widetilde{Z}R]_{ij})([\widetilde{Z}O]_{ij}-[\widetilde{Z}R]_{ij})| \nonumber \\
							&\leq& \sum_{ij}|([\widetilde{Z}O]_{ij}^{2} + [\widetilde{Z}R]_{ij}^{2})([\widetilde{Z}O]_{ij}+[\widetilde{Z}R]_{ij})([ZO]_{ij}-[\widetilde{Z}R]_{ij})|  \nonumber\\
							&\leq& O_p(\log^2 n)\sum_{ij}|([\widetilde{Z}O]_{ij}+[\widetilde{Z}R]_{ij})([\widetilde{Z}O]_{ij}-[\widetilde{Z}R]_{ij})|  \nonumber\\
		 					&\leq& O_{p}(\log^{3}n)\sum_{ij}|[\widetilde{Z}O]_{ij}-[\widetilde{Z}R]_{ij}| \nonumber\\
							&\leq& O_{p}(n\log^{3}n)\|\widetilde{Z}(O-R)\|_{1\to\infty} \nonumber\\
							&\leq& O_{p}(n\log^{3}n)\|\widetilde{Z}\|_{\max}(\sum_{ij}|O_{ij}-R_{ij}|) \nonumber\\
							&=& O_{p}(n\log^{4}n)(\sum_{ij}|O_{ij}-R_{ij}|)   \nonumber\\
							&\leq& O_{p}(n\log^{4}n)\times k\sqrt{\sum_{ij}|O_{ij}-R_{ij}|^{2}} \nonumber\\
							&=& O_{p}(n\log^{4}n)\|O-R\|_{F} \nonumber\\
							&\leq& O_{p}(n\log^{4}n) \times \sqrt{2}d(O,R) \nonumber\\
							&=& O_{p}(\varepsilon n\log^{4}n) \nonumber\\
							&=& O_{p}(\log^{4}n),			 \label{eq:A7one}
\end{eqnarray}

\newpage

and

\begin{eqnarray}
&&   |\sum_{j}[(\sum_{i}[\widetilde{Z}O]_{ij}^{2})^{2} - (\sum_{i}[\widetilde{Z}R]^{2}_{ij})^{2}]| \nonumber\\
&=& |\sum_{j}[(\sum_{i}([\widetilde{Z}O]_{ij}^{2}+[\widetilde{Z}R]_{ij}^{2}))(\sum_{i}([\widetilde{Z}O]_{ij}^{2}-[\widetilde{Z}R]_{ij}^{2}))]| \nonumber\\
&\leq& \sum_{j}|(\sum_{i}([\widetilde{Z}O]_{ij}^{2}+[\widetilde{Z}R]_{ij}^{2}))(\sum_{i}([\widetilde{Z}O]_{ij}^{2}-[\widetilde{Z}R]_{ij}^{2}))| \nonumber\\
&\leq& \sum_{j}|(2n\|\widetilde{Z}\|_{\max}^{2})(\sum_{i}([\widetilde{Z}O]_{ij}^{2}-[\widetilde{Z}R]_{ij}^{2}))| \nonumber\\
&\leq& O_{p}(n\log^{2}n)\sum_{j}|\sum_{i}([\widetilde{Z}O]_{ij}+[\widetilde{Z}R]_{ij})([\widetilde{Z}O]_{ij}-[\widetilde{Z}R]_{ij})| \nonumber\\
&\leq& O_{p}(n\log^{3}n)\sum_{j}|\sum_i [\widetilde{Z}O]_{ij}-[\widetilde{Z}R]_{ij}| \nonumber\\
&\leq& O_{p}(n^{2}\log^{3}n)\|\widetilde{Z}(O-R)\|_{1\to\infty} \nonumber\\
&\leq& O_{p}(n^{2}\log^{3}n)\|\widetilde{Z}\|_{\max}(\sum_{ij}|O_{ij}-R_{ij}|)\nonumber \\
&\leq& O_{p}(n^{2}\log^{4}n)\times k\|O-R\|_{F} \nonumber\\
&\leq& O_{p}(n^{2}\log^{4}n) \times \sqrt{2}\times d(O,R) \nonumber\\
&=& O_{p}(n^{2}\varepsilon\log^{4}n)  \nonumber\\
&=& O_{p}(n\log^{4}n).    \label{eq:A7two}
\end{eqnarray}

With Equation \eqref{eq:A7one}, \eqref{eq:A7two},
\begin{eqnarray*}
&& |\widehat{V}(O) - \widehat{V}(R_{\ell})| \\
&\leq& \frac{1}{n}|\sum_{ij}([\widetilde{Z}O]_{ij}^{4} - [\widetilde{Z}R_{\ell}]_{ij}^{4}|) + \frac{1}{n^{2}}|\sum_{j}[(\sum_{i}[\widetilde{Z}O]_{ij}^{2})^{2} - (\sum_{i}[\widetilde{Z}R_{\ell}]^{2}_{ij})^{2}]| \\
&=& O_{p}(\frac{\log^{4}n}{n}).
\end{eqnarray*}

\newpage
Assume $\mathbb{E}(\widetilde{Z}_{1\ell}^{j}) = \mu_{j}^{(\ell)}$ refers to $j$-th moment of $Z$'s $\ell$-th column, similar to the proof of Theorem \ref{thm:varimax1} the population Varimax function could be expressed as, 
$$V(Q) = \sum_{j}(\mathbb{E}([\widetilde{Z}Q]_{j}^{4}) - \mathbb{E}([\widetilde{Z}Q]_{j}^{2})^{2}) = \sum_{i=1}^{k}(\mu_{4}^{(i)} - 3)\|Q_{i\cdot}\|_{4}^{4}+3k.$$ 
Write $\xi_{i} = \mu_{4}^{(i)} - 3\mu_{2}^{(i)2} = \mu_{4}^{(i)} - 3$ (is positive by leptokurtic assumption) and $\xi_{0} = \underset{i}{\max}$ $\xi_{i}$ (is a finite positive constant). For $\forall O, R \in \mathcal{O}(k)$ such that $d(O, R) < \varepsilon$, there is
\begin{eqnarray*}
|V(O) - V(R)| &=& |\sum_{i=1}^{k}\xi_{i}(\|O_{i\cdot}\|^{4}_{4} - \|R_{i\cdot}\|_{4}^{4})| \\
		    &\leq& \sum_{i=1}^{k}|\xi_{i}(\|O_{i\cdot}\|^{4}_{4} - \|R_{i\cdot}\|_{4}^{4})| \\
		    &\leq& \xi_{0}\sum_{i=1}^{k}|\|O_{i\cdot}\|^{4}_{4} - \|R_{i\cdot}\|_{4}^{4}| \\
		    &=& \xi_{0}\sum_{i=1}^{k}|\sum_{j=1}^{k}(O_{ij}^{2}+R_{ij}^{2})(O_{ij}+R_{ij})(O_{ij}-R_{ij})| \\
		    &\leq& \xi_{0}\sum_{i=1}^{k}\sum_{j=1}^{k}|(O_{ij}^{2}+R_{ij}^{2})(O_{ij}+R_{ij})(O_{ij}-R_{ij})| \\
		    &\leq& 4\xi_{0}\sum_{i=1}^{k}\sum_{j=1}^{k}|O_{ij}-R_{ij}| \\
		    &\leq& 4\xi_{0}k\|O-R\|_{F} \\
		    &=& O_{p}(\frac{1}{n}). 
\end{eqnarray*}

Summing up three bounds of Equation \eqref{eq:triangles} obtains a uniform bound of $|\widehat{V}(O) - V(O)|$: 
\begin{equation}
\underset{O\in\oo(k)}{\sup}|\widehat{V}(O) - V(O)| = O_{p}(n^{\delta - \frac{1}{2}}).
\end{equation}

\end{proof}

\newpage

\subsubsection*{Proof of Lemma \ref{lemma:SOC}}
This part of proof needs first/second order condition of population varimax function here. These two conditions are stated as Corollary \ref{corollary:foc}, \ref{corollary:soc} below. For now the notations of sample \& population varimax functions follow the proofs of Lemma \ref{lemma:uniform}.

\begin{corollary}[FOC] \label{corollary:foc}
If identity matrix $I$ is a stationary point of $\vv_{I}(R_0)$ then $Z$ satisfies the following condition,
\begin{equation}\label{eq:foc}
\mathbb{E}Z_{1i}^{2}\mathbb{E}(Z_{1i}Z_{1j}) - \mathbb{E}(Z_{1i}^{3}Z_{1j}) = \mathbb{E}Z_{1j}^{2}\mathbb{E}(Z_{1i}Z_{1j}) - \mathbb{E}(Z_{1j}^{3}Z_{1i}), \forall i\neq j.
\end{equation}
\end{corollary}

\begin{corollary}[SOC] \label{corollary:soc}
Notate $O = \widetilde{Z}_{1}^{T}\widetilde{Z}_1$, if $I$ is a local maxima of the population Varimax $\vv_{I}(R_0)$, then the following condition is true,
\begin{equation} \label{eq:soc}
3\mathbb{E}[tr(diag(OK)^{2})] \leq \mathbb{E}\langle OdiagO,KK^{T}\rangle,
\end{equation}
for any skew-symmetric matrix K.
\end{corollary}

The proof of Corollary \ref{corollary:foc} is trivial. Corollary \ref{corollary:soc} is a direct result of Theorem \ref{thm:soc} in Section \ref{sec:FSOC}. Now we could proceed to the proof of Lemma \ref{lemma:SOC}.

\begin{proof}  By Proposition \ref{prop:converge}, $R_{\widetilde{Z}}$ is converging to elements of $\pp(k)$, WLOG we may assume $P_{n}^{(2)} = I$ and $R_{\widetilde{Z}} \to I$ (i.e. let $P_n^{(2)}=\widetilde{R}_{U}^{T}$), since elements of $\pp(k)$ are isolated to each other ($\forall P_1\neq P_2 \in \mathcal{P}(k)$, $\|P_1 - P_2\| \geq 2/\sqrt{k}$). We may constrain our analysis on a fixed neighborhood of $I$, $B(I, \delta_c)$ s.t. $\{I\} = B(I, \delta_c) \cap \mathcal{P}(k)$. Now we want to show,
$$\|R_{\widetilde{Z}}-I\|_{2\to\infty}  = O_{p}(n^{\delta/2-1/4}). $$ 

By Lie algebra theory there is a $k\times k$ skew-symmetric matrix $K$ s.t. $R_{\widetilde{Z}} = \exp(K)$. Define 
$$\gamma(t) = \exp(tK).$$ 

Then $\gamma(0) = I$, $\gamma(1) = R_{\widetilde{Z}}$. We want to evaluate population varimax function's first order and second order condition at $I$ (global optimal solution). To achieve that we should show: $R \to I \Rightarrow K \to 0$. This could be proved by using matrix logarithm algebra.
$$ \|K\| = \|\log R_{\widetilde{Z}} - \log I \| \leq \sum_{i=1}^{\infty}\frac{(-1)^{i+1}}{i}\|R_{\widetilde{Z}}-I\|^{i} \to 0. $$

Differential calculations indicates:

\begin{equation}
\frac{d}{dt}V(\gamma(t))|_{t=0} = \nabla V(\gamma(t))^{T}\frac{d\gamma}{dt}|_{t=0} = \langle \nabla V(I), K \rangle,
\end{equation}

\begin{eqnarray*}
\frac{d^{2}}{dt^{2}}V(\gamma(t))|_{t=0} &=& \langle \nabla^{2} V(\frac{d\gamma(t)}{dt})|_{t=0}, K\gamma(t)\rangle + \langle \nabla V, K^{2} \gamma(0) \rangle \\
							   &=& \langle \nabla^{2} V\cdot K, K\rangle + \langle \nabla V, K^{2} \rangle.
\end{eqnarray*}

Equation (2) of \cite{chu1998orthomax} is a reformulated version of Varimax function. Taking expectation of it (Lemma \ref{lemma:regularization} allows exchanging differential and expectation), and notate $E = \widetilde{Z}^{T}\widetilde{Z} - n\widetilde{Z}_{1}^{T}\widetilde{Z}_{1}$. The varimax function could be rewritten as: 
\begin{equation}
V(Q) := \mathbb{E}[trace(diag(Q^{T}EQ)^{2})] \Rightarrow \nabla V(Q) = 4\mathbb{E}[EQdiag(Q^{T}EQ], Q \in \oo(k).
\end{equation}

Let $O = \widetilde{Z}_{1}^{T}\widetilde{Z}_1$, then
\begin{equation} \label{eq:FI}
\nabla V(I) = 4\mathbb{E}[(I-O)diag(I-O)] = 4\mathbb{E}(Odiag(O)) - 4I.
\end{equation}

Thus by Corollary \ref{corollary:foc}, $\nabla V(I)$ is symmetric $\Rightarrow$ $\langle \nabla V(I), K\rangle = 0$.

By Frechet derivatives, for any $H\in\R^{k\times k}$ and $t > 0$,
\begin{eqnarray*}
\nabla V(Q + tH) - \nabla V(Q) &=& 4\mathbb{E}[E(Q+tH)diag((Q+tH)^{T}E(Q+tH)] - 4\mathbb{E}[EQdiag(Q^{T}EQ)] \\
					      &=& 4t\mathbb{E}[EHdiag(Q^{T}EQ) + 2EQdiag(H^{T}EQ)] + O(t^{2}),
\end{eqnarray*}

choose $H=K, Q=I$, which means the derivative of $\nabla V(Q)$ evaluated at $I$ in the direction of $K$ is
\begin{equation}
\nabla^{2}V(I)\cdot K = -4K + 4\mathbb{E}[OKdiag(O)] + 8\mathbb{E}[Odiag(OK)].
\end{equation}

Theorem \ref{thm:varimax1} indicates identity matrix $I$ is one of the global maximas of population Varimax function. Applying Second Order Condition result (Corollary \ref{corollary:soc}), Lemma \ref{lemma:tr} and reusing $\langle \nabla V(I), K\rangle = 0$ obtains
\begin{eqnarray*}
\langle \nabla^{2}V(I)\cdot K, K \rangle &=& -4\langle K, K\rangle + 4\mathbb{E}\langle OKdiag(O), K \rangle + 8\mathbb{E}\langle Odiag(OK), K\rangle \\
							   &=& -4\|K\|_{F}^{2} + 12\mathbb{E}[trace(diag(OK)^{2})] \\
							   &\leq& -4\|K\|_{F}^{2} + 4\mathbb{E}\langle Odiag(O), KK^{T} \rangle,
\end{eqnarray*}
and
\begin{equation}
\langle \nabla V(I), K^{2} \rangle = -4\mathbb{E}\langle Odiag(O), KK^{T}\rangle + 4\|K\|_{F}^{2}.
\end{equation}

Let $K^{u} = K / \|K\|$, then
\begin{eqnarray*}
\frac{\partial^{2} V}{\partial t^{2}}|_{t=0} &=& \langle \nabla^{2}V(I)\cdot K, K\rangle + \langle\nabla V(I), K^{2}\rangle \\
							    &=& -4\|K\|_{F}^{2} + 12\mathbb{E}[trace(diag(OK)^{2})] -4\mathbb{E}\langle Odiag(O), KK^{T}\rangle + 4\|K\|_{F}^{2} \\
							    &=& 4\times(3\mathbb{E}[trace(diag(OK)^{2}) - \mathbb{E}\langle Odiag(O), KK^{T}\rangle]) \\
							    &=& 4\times(3\mathbb{E}[trace(diag(OK^{u})^{2}) - \mathbb{E}\langle Odiag(O), K^{u}K^{uT}\rangle])\|K\|^{2} \\
							    &\leq& -C_{s}\|K\|^{2}.
\end{eqnarray*}

\noindent Here 
\begin{equation} \label{eq:minusCs}
-C_{s} = \underset{\|K^{u}\|=1,K^{u} \in \mathcal{O}(k)^{\bot}}{\max} 4\times(3\mathbb{E}[trace(diag(OK^{u})^{2}) - \mathbb{E}\langle Odiag(O), K^{u}K^{uT}\rangle]), 
\end{equation}
is a negative constant (thus $C_s$ is a positive constant) since RHS of Equation \eqref{eq:minusCs} is upper-bounded and the set of skewed symmetric matrix with unit Frobenius norm is a bounded, closed and compact space. With derived results and Taylor expansion,
\begin{eqnarray*}
V(R_{\widetilde{Z}}) &=& V(I) + \langle \nabla V(I), K \rangle + \langle \nabla^{2}V(I)\cdot K, K\rangle + \langle\nabla V(I), K^{2}\rangle + o(\|K\|^{2}) \\
	&=& V(I) + \langle \nabla^{2}V(I)\cdot K, K\rangle + \langle\nabla V(I), K^{2}\rangle + o(\|K\|^{2}) \\
	&\leq& V(I) - C_{s}\|K\|^{2} + o(\|K\|^{2}).
\end{eqnarray*}

By Lemma \ref{lemma:uniform} there exists $\varepsilon_0 = O_{p}(n^{\delta-1/2})$ s.t.

$$ |\widehat{V}(R_{\widetilde{Z}}) - V(R_{\widetilde{Z}})| < \varepsilon_0, \quad  |\widehat{V}(I) - V(I)| < \varepsilon_0,  $$
$\Rightarrow$
$$ \widehat{V}(R_{\widetilde{Z}}) - \varepsilon_0  < V(R_{\widetilde{Z}}) < V(I)  < \widehat{V}(I) + \varepsilon_0,$$
$\Rightarrow$
$$V(I) - V(R_{\widetilde{Z}}) < 2\varepsilon_0. $$

These implies
\begin{equation}
\|K\|^{2} < \frac{2\varepsilon_0 + o(\|K\|^{2})}{C_{s}}.
\end{equation}

Therefore $\|K\| = O_{p}(n^{\delta/2-1/4})$. By matrix exponential algebra, 
\begin{equation}
\|R_{\widetilde{Z}}-I\|_{2\to\infty} \leq \|R_{\widetilde{Z}}-I\| = \|\sum_{i=1}^{\infty}\frac{K^{i}}{i!}\| \leq \sum_{i=1}^{\infty}\frac{\|K\|^{i}}{i!} = O_{p}(n^{\delta/2-1/4}). 
\end{equation}

\end{proof}

\subsubsection*{Detailed Proof of Proposition \ref{prop:varifun}}

\begin{proof}  
 
Write
$$ V_{1}(O) = \sum_{\ell=1}^k\left(\frac{1}{n}\sum_{i=1}^n [\sqrt{n}\widehat{U}O]_{i\ell}^4 - \left(\frac{1}{n} \sum_{i=1}^n [\sqrt{n}\widehat{U}O]_{i\ell}^2\right)^2\right), $$
$$ V_{2}(O) = \sum_{\ell=1}^k\left(\frac{1}{n}\sum_{i=1}^n [\sqrt{n}UWO]_{i\ell}^4 - \left(\frac{1}{n} \sum_{i=1}^n [\sqrt{n}UWO]_{i\ell}^2\right)^2\right). $$

\vspace{.2in}
To be specific, $V_1$ is the sample version of Varimax function with perturbed eigenvectors as input. $V_2$ is sample version of Varimax function with true eigenvectors rotated with $W$ (specified in Lemma \ref{lemma:Cape}). The proof of Proposition \ref{prop:varifun} could be described as two parts. First part shows the uniform upper bound for difference between $V_1, V_2$ (Equation \eqref{eq:V1V2}). Similar to the proof of Lemma \ref{lemma:SOC}, the second part explores the first and second order condition of Equation \eqref{eq:gamma2} to obtain the bound for the difference between solutions of $V_1$ and $V_2$ (modulo permutation and sign-flip).

\vspace{.1in}

Mathematically speaking, the first part (uniform upper bound for difference between $V_1, V_2$) is equivalent to
\begin{equation} \label{eq:V1V2}
\underset{O \in \mathcal{O}(k)}{\sup} |V_{1}(O) - V_{2}(O)| \leq O_{p}\left( (n\rho_{n})^{-1/2}\log^{\frac{7}{2}} n \right).
\end{equation}

 \vspace{.1in}

In the proof of Proposition \ref{prop:varifun} let $X_{i}$ be the $i$th row of $\sqrt{n}\widehat{U}$, and $i$th row of $\sqrt{n}UW$ be $X_{i} + \epsilon_{i}$. From Lemma \ref{lemma:Cape}, for any unit length vector $r \in \R^{k}$, we have 
$$ |(\langle X_{i},r\rangle - \langle X_{i}+\epsilon_{i},r\rangle) | \leq \|\epsilon_i\|\|r\| \leq \sqrt{n}\|\widehat{U} - UW\|_{2\to\infty} = O_{p}(  (n\rho_{n})^{-1/2} \log^{\frac{5}{2}} n ). $$

Therefore,
\begin{eqnarray*}
&& |\sum_{i=1}^{n}(\langle X_{i},r\rangle^{4} - \langle X_{i}+\epsilon_{i},r\rangle^{4})|   \\
&\leq& \sum_{i=1}^{n}| (\langle X_{i},r\rangle^{2} + \langle X_{i}+\epsilon_{i},r\rangle^{2})(\langle X_{i},r\rangle + \langle X_{i}+\epsilon_{i},r\rangle)(\langle X_{i},r\rangle - \langle X_{i}+\epsilon_{i},r\rangle) | \\
&\leq& \sum_{i=1}^{n} (\langle X_{i},r\rangle^{2} + \langle X_{i}+\epsilon_{i},r\rangle^{2})(\|X_{i}\|_{2}+\|X_{i}+\epsilon_{i}\|_{2})|(\langle X_{i},r\rangle - \langle X_{i}+\epsilon_{i},r\rangle) | \\
&\leq& \sum_{i=1}^{n} (\langle X_{i},r\rangle^{2} + \langle X_{i}+\epsilon_{i},r\rangle^{2})O_{p}(\log n)|(\langle X_{i},r\rangle - \langle X_{i}+\epsilon_{i},r\rangle) | \\
&\leq& \sum_{i=1}^{n} (\langle X_{i},r\rangle^{2} + \langle X_{i}+\epsilon_{i},r\rangle^{2})O_{p}(\log n) \sqrt{n} \|\widehat{U} - UW\|_{2\to\infty}  \\
&\leq&  \sum_{i=1}^{n} (\langle X_{i},r\rangle^{2} + \langle X_{i}+\epsilon_{i},r\rangle^{2})O_{p}\left(  (n\rho_{n})^{-1/2} \log^{\frac{7}{2}} n \right).
\end{eqnarray*}


 \vspace{.1in}

Notice that columns of $\widehat{U}$ and $UW$ have unit length and $R$ is an orthogonal matrix. Thus the columns of $\widehat{U}R$ and $UWR$ are all of unit length. Therefore


$$ \left(\frac{1}{n} \sum_{i=1}^n [\sqrt{n}\widehat{U}R]_{i\ell}^2\right)^2  = \left(\sum_{i=1}^n [\widehat{U}R]_{i\ell}^2\right)^2 = 1^2 = \left(\sum_{i=1}^n [UWR]_{i\ell}^2\right)^2 = \left(\frac{1}{n} \sum_{i=1}^n [\sqrt{n}UWR]_{i\ell}^2\right)^2. $$

\newpage
Let $O_{\ell}$ be the $\ell$th column of $O$. Then for any $O \in \mathcal{O}(k)$,

\begin{eqnarray*}
&& |V_{1}(O) - V_{2}(O)|  \\
&\leq& \sum_{\ell=1}^k|\left(\frac{1}{n}\sum_{i=1}^n [\sqrt{n}\widehat{U}O]_{i\ell}^4 - \left(\frac{1}{n} \sum_{i=1}^n [\sqrt{n}\widehat{U}O]_{i\ell}^2\right)^2\right) \\
&& \quad \quad - \left(\frac{1}{n} \sum_{i=1}^n [\sqrt{n}UWO]_{i\ell}^4 - \left(\frac{1}{n} \sum_{i=1}^n [\sqrt{n}UWO]_{i\ell}^2\right)^2\right)| \\
&\leq& \frac{1}{n}\sum_{\ell=1}^{k}|\sum_{i=1}^n ([\sqrt{n}\widehat{U}O]_{i\ell}^4 - [\sqrt{n}UWO]_{i\ell}^4)| \\
&\leq& \frac{1}{n}\sum_{\ell=1}^{k}|\sum_{i=1}^{n}(\langle X_{i},O_{\ell}\rangle^{4} - \langle X_{i}+\epsilon_{i},O_{\ell}\rangle^{4})| \\
&\leq& \frac{1}{n}\sum_{\ell=1}^{k} \sum_{i=1}^{n} (\langle X_{i},O_{\ell}\rangle^{2} + \langle X_{i}+\epsilon_{i},O_{\ell}\rangle^{2})O_{p}\left(  (n\rho_{n})^{-1/2} \log^{\frac{7}{2}} n \right) \\
&=& \sum_{\ell=1}^{k}O_{p}\left(  (n\rho_{n})^{-1/2} \log^{\frac{7}{2}} n \right) \\
&=& O_{p}\left(  (n\rho_{n})^{-1/2} \log^{\frac{7}{2}} n \right).
\end{eqnarray*}


Since the orthogonal matrix $O$ here is arbitrary, therefore the Equation \eqref{eq:V1V2} is proved.


 \vspace{.2in}
For the next step, we want to show the upper bound of $2\to\infty$ norm distance between $R_{\widehat{U}}$ and $R_{UW}P_{n}^{(3)}$ ($P_{n}^{(3)} \in \mathcal{P}(k)$ is defined in Proposition \ref{prop:varifun}),

\begin{equation} \label{eq:Pn3}
\|R_{\widehat{U}} - R_{UW}P_{n}^{(3)}\|_{2\to\infty} = O_{p}\left(  (n\rho_{n})^{-1/4} \log^{\frac{7}{4}} n \right).
\end{equation}

\newpage

For simplicity notate $R_{1} = R_{\widehat{U}}$, $R_{2} = R_{UW}$. There are $k\times k$ skew-symmetric matrices $K_{1}$, $K_{2}$ s.t. $R_{1} = \exp(K_{1})$, $R_{2} = \exp(K_{2})$. Define 
\begin{equation} \label{eq:gamma2}
\gamma_{2}(t) = \exp((1-t)K_{2} + tK_{1}),
\end{equation}
then $\gamma_{2}(0) = R_{2}$, $\gamma_{2}(1) = R_{1}$. Again, as in the proof of Lemma \ref{lemma:SOC} we assume 
$$ I = \underset{P_0 \in \mathcal{P}(k)}{\arg \min}\|R_1 - R_2 P_0\|_{2\to \infty}, $$


and we constrain our analysis on a neighborhood of $R_{2}$: $B(R_2, \delta_{p}):=\{P\in \oo(k)| \|P-R_2\| < \delta_p \}$, such that 
$$ B(R_2, \delta_{p}) \cap \{R_2P_0| P_0 \in \mathcal{P}(k) \} = \{R_2\}. $$

This indicates for any $R \in B(R_2, \delta_p)$ there is $V_2(R) \leq V_2(R_2)$.

Before Taylor expansion analysis, we should check that $\|R_1 - R_2\| \overset{p}{\to}0$ is true. After that we should show $\|K_1 - K_2\| \overset{p}{\to}0$ is also true. By definition,
\begin{equation}
V_1(R_1) \geq V_1(R_2) - o_{p}(1), |V_1(R_2) - V_2(R_2)| \overset{p}{\to}0 \Rightarrow V_1(R_1) \geq V_2(R_2) - o_p(1). 
\end{equation}

Then, 
\begin{eqnarray}
V_2(R_2) - V_2(V_1) &\leq& V_1(R_1) - V_2(R_1) + o_p(1) \\
				 &\leq& \underset{R\in \mathcal{O}(k)}{\sup}|V_1(R)-V_2(R)| + o_p(1) \overset{p}{\to} 0.
\end{eqnarray}

By conditions, for any $\epsilon_0 > 0, \eta_0 > 0$ such that $V_2(R) < V_2(R_2) - \eta_0$ for every $R\in\oo(k)$ with $\|R - R_2\| \geq \epsilon_0$. Thus the event $\{\|R_2 - R_1\|\}$ is contained in the event $\{V_2(R_1) < V_2(R_2) - \eta_0\}$. The probability of the latter event goes to 0. Therefore $\|R_1 - R_2\| \overset{p}{\to}0$. By Lemma \ref{lemma:SOC}, with high probability, $R_2$ and $R_1R_{2}^{T}$ are both converging to $I$ as $n$ grows. Variant of Baker-Cambell-Hausdorff formula gives
\begin{eqnarray*}
\|K_1 - K_2\| &=& \| \log R_1R_{2}^{T}R_2 - \log R_2\| \\
			  &=& \| \log R_1R_{2}^{T} + \frac{1}{2}[\log R_1R_{2}^{T}, \log R_2] + \cdots \| \\
			  &\leq& \| \log (I + R_1(R_{2}^{T}-R_{1}^{T})) \|_{F} + o_{p}(\|\log R_2\|) \\
			  &\overset{p}{\to}& 0.
\end{eqnarray*}

With differential calculation results in \cite{chu1998orthomax},

\begin{equation}
\frac{d}{dt}V_{2}(\gamma_{2}(t))|_{t=0} = \nabla V_{2}(\gamma_{2}(t))^{T}\frac{d\gamma_{2}}{dt}|_{t=0} = \langle \nabla V_{2}(R_{2}), R_{2}(K_{1} - K_{2}) \rangle,
\end{equation}

\begin{eqnarray*}
\frac{d^{2}}{dt^{2}}V_{2}(\gamma_{2}(t))|_{t=0} &=& \langle \nabla^{2} V_{2}(\frac{d\gamma_{2}(t)}{dt})|_{t=0}, \gamma_{2}(0)(K_{1} - K_{2})\rangle + \langle \nabla V, \gamma_{2}(0)(K_{1} - K_{2})^{2}  \rangle \\
							   &=& \langle \nabla^{2} V_{2}\cdot R_{2}(K_{1} - K_{2}), R_{2}(K_{1} - K_{2})\rangle + \langle \nabla V_{2}, R_{2}(K_{1} - K_{2})^{2} \rangle. \\
\end{eqnarray*}

By Equation (7),(8) of \cite{chu1998orthomax},
\begin{equation}
V_{2}(Q) = n^{-3}trace\left[\sum_{i=1}^{n}diag(Q^{T}E_{i}Q)^{2}\right],
\end{equation}
where $E_{i} = (UW)^{T}UW - n(X_{i}+\epsilon_{i})(X_{i}+\epsilon_{i})^{T}$. And

\begin{equation}
\nabla V_{2}(Q) = 4n^{-3}\left[\sum_{i=1}^{n}E_{i}Qdiag(Q^{T}E_{i}Q)\right].
\end{equation}

Theorem 3.1 in \cite{chu1998orthomax} implies 
$$R_{2}^{T}\nabla V_{2}(R_{2}) = \sum_{i=1}^{n}R_{2}^{T}E_{i}R_{2}diag(R_{2}^{T}E_{2}Q_{2})$$ 
is symmetric, thus
\begin{eqnarray}
\langle \nabla V_{2}(R_{2}), (K_{1} - K_{2})R_{2} \rangle &=& trace[\nabla V_{2}(K_{1}-K_{2})^{T}R_{2}^{T}] \nonumber \\
			&=& trace[R_{2}^{T}\nabla V_{2}(K_{1}-K_{2})^{T}] \nonumber \\
			&=& \langle R_{2}^{T}\nabla V_{2}, K_{1}-K_{2} \rangle \nonumber \\
			&=& 0.   \label{eq:sampletaylor1}
\end{eqnarray}

The last equality is because $K$ is skew-symmetric and $R_{2}^{T}\nabla V_{2}(R_{2})$ is symmetric.

\vspace{.1in}
By Frechet derivatives, for any $H\in \R^{k\times k}$ and $t > 0$, we have
\begin{eqnarray*}
&& \nabla V_{2}(Q + tH) - \nabla V_{2}(Q) \\
&=& 4n^{-3}\left[\sum_{i=1}^{n}E_{i}(Q+tH)diag((Q+tH)^{T}E_{i}(Q+tH))\right] - 4n^{-3}\left[\sum_{i=1}^{n}E_{i}Qdiag(Q^{T}E_{i}Q)\right] \\
&=& 4n^{-3}\left[\sum_{i=1}^{n}(E_{i}Hdiag(Q^{T}E_{i}Q+E_{i}Q(diag(H^{T}E_{i}Q)+diag(Q^{T}E_{i}H))\right] + O(t^{2}). \\
\end{eqnarray*}

Choosing $H=R_{2}(K_{1}-K_{2}), Q=R_{2}$ gives the derivative of $\nabla V_{2}(Q)$ evaluated at $R_{2}$ in the direction of $R_{2}(K_{1}-K_{2})$
\begin{eqnarray*}
\nabla^{2}V_{2}(R_{2})\cdot R_{2}(K_{1}-K_{2}) &=& 4n^{-3}\sum_{i=1}^{n}[E_{i}R_{2}(K_{1}-K_{2})diag(R_{2}^{T}E_{i}R_{2}) \\
&+& E_{i}R_{2}diag((K_{1}-K_{2})^{T}R_{2}^{T}E_{i}R_{2}) \\
&+& E_{i}R_{2}diag(R_{2}^{T}E_{i}R_{2}(K_{1}-K_{2}))].
\end{eqnarray*}

Applying Corollary 3.3 of \cite{chu1998orthomax}, 
\begin{eqnarray}
&& \langle \nabla^{2}V_{2}(R_{2})\cdot R_{2}(K_{1}-K_{2}), R_{2}(K_{1}-K_{2}) \rangle + \nabla V, R_{2}(K_{1}-K_{2})^{2} \rangle \nonumber \\
&=& 4n^{-3}\sum_{i=1}^{n}[\langle E_{i}R_{2}(K_{1}-K_{2})diag(R_{2}^{T}E_{i}R_{2}),R_{2}(K_{1}-K_{2})\rangle\nonumber  \\
&+& \langle E_{i}R_{2}diag((K_{1}-K_{2})^{T}R_{2}^{T}E_{i}R_{2}),R_{2}(K_{1}-K_{2})\rangle \nonumber \\
&+& \langle E_{i}R_{2}diag(R_{2}^{T}E_{i}R_{2}(K_{1}-K_{2})), R_{2}(K_{1}-K_{2}) \rangle \nonumber \\
&+& \langle E_{i}R_{2}diag(R_{2}^{T}E_{i}R_{2}), R_{2}(K_{1}-K_{2})^{2} \rangle]\nonumber  \\
&=& 4n^{-3}\sum_{i=1}^{n}[\langle R_{2}^{T}E_{i}R_{2}(K_{1}-K_{2})diag(R_{2}^{T}E_{i}R_{2}),(K_{1}-K_{2})\rangle \nonumber \\
&+& 2\langle R_{2}^{T}E_{i}R_{2}diag(R_{2}^{T}E_{i}R_{2}(K_{1}-K_{2})), K_{1}-K_{2} \rangle \nonumber \\
&+& \langle R_{2}^{T}E_{i}R_{2}diag(R_{2}^{T}E_{i}R_{2}), (K_{1}-K_{2})^{2} \rangle] \leq 0. \label{eq:sampletaylor2}
\end{eqnarray}

Let $K^{u}_{o} =  (K_{1}-K_{2}) / \|K_{1}-K_{2}\|_{F}$, then
\begin{eqnarray}
&& \langle \nabla^{2}V_{2}(R_{2})\cdot R_{2}(K_{1}-K_{2}), R_{2}(K_{1}-K_{2}) \rangle + \nabla V, R_{2}(K_{1}-K_{2})^{2} \rangle \nonumber  \\
&=& 4n^{-3}\|K_{1}-K_{2}\|_{F}\sum_{i=1}^{n}[\langle R_{2}^{T}E_{i}R_{2}K^{u}_{o}diag(R_{2}^{T}E_{i}R_{2}),K^{u}_{o}\rangle \nonumber \\
&+& 2\langle R_{2}^{T}E_{i}R_{2}diag(R_{2}^{T}E_{i}R_{2}K^{u}_{o}), K^{u}_{o} \rangle \nonumber \\
&+& \langle R_{2}^{T}E_{i}R_{2}diag(R_{2}^{T}E_{i}R_{2}), (K^{u}_{o})^{2} \rangle] \nonumber \\
&\leq& -C_{ss}\|K_{1}-K_{2}\|_{F}.  \label{eq:sampletaylor3}
\end{eqnarray}

Here 
\begin{eqnarray*}
&& -C_{ss} = \underset{\|K\|_{F} =1,K \in \mathcal{O}(k)^{\bot}}{\max} 4n^{-3}\sum_{i=1}^{n}[\langle R_{2}^{T}E_{i}R_{2}Kdiag(R_{2}^{T}E_{i}R_{2}),K\rangle \\
&&\quad \quad \quad \quad \quad + 2\langle R_{2}^{T}E_{i}R_{2}diag(R_{2}^{T}E_{i}R_{2}K), K \rangle + \langle R_{2}^{T}E_{i}R_{2}diag(R_{2}^{T}E_{i}R_{2}), K^{2} \rangle]
\end{eqnarray*}
is a negative constant. With Taylor expansion and Equation \eqref{eq:sampletaylor1}, \eqref{eq:sampletaylor2}, \eqref{eq:sampletaylor3}, there is

\begin{eqnarray}  \label{eq:K1K2}
V_{2}(R_{1}) &=& V_{2}(R_{2}) + \langle \nabla V_{2}(R_{2}), (K_{1} - K_{2})R_{2} \rangle + \langle \nabla^{2}V_{2}(R_{2})\cdot R_{2}(K_{1}-K_{2}), R_{2}(K_{1}-K_{2}) \rangle +\nonumber \\ && \langle \nabla V, R_{2}(K_{1}-K_{2})^{2} \rangle + o(\|K_{1}-K_{2}\|^{2}_{F}) \nonumber \\
	&=& V_{2}(R_{2}) + \langle \nabla^{2}V_{2}(R_{2})\cdot R_{2}(K_{1}-K_{2}), R_{2}(K_{1}-K_{2}) \rangle + \nonumber \\ 
	&& \quad \quad \quad \nabla V, R_{2}(K_{1}-K_{2})^{2} \rangle + o(\|K_{1}-K_{2}\|^{2}_{F}) \nonumber \\
	&\leq& V_{2}(R_{2}) - C_{ss}\|K_{1}-K_{2}\|^{2}_{F} + o(\|K_{1}-K_{2}\|^{2}_{F}). 
\end{eqnarray}

With Equation \eqref{eq:V1V2}, there exists $\varepsilon_1 = O_{p}(  (n\rho_{n})^{-1/2} \log^{\frac{7}{2}} n )$, s.t.

$$ |V_{2}(R_{1}) - V_{1}(R_{1})| < \varepsilon_1, \quad   |V_{2}(R_{2}) - V_{1}(R_{2})| < \varepsilon_1,  $$
$\Rightarrow$
$$ V_{1}(R_{1}) - \varepsilon_1 < V_{2}(R_{1}) < V_{2}(R_{2}) < V_{1}(R_{2}) + \varepsilon_1, $$
$\Rightarrow$
$$V_{2}(R_{2}) - V_{2}(R_{1}) < 2\varepsilon_1. $$

Then from Equation \eqref{eq:K1K2} there is
\begin{equation}
\|K_{1}-K_{2}\|^{2}_{F} < \frac{2\varepsilon_1 + o(\|K_{1}-K_{2}\|^{2}_{F})}{C_{ss}}.
\end{equation}

Thus $\|K_{1}-K_{2}\|_{F} = O_{p}\left(  (n\rho_{n})^{-1/4} \log^{\frac{7}{4}} n \right)$. By Lie Product Formula, for any $k\times k$ matrices $S_{1}, S_{2}$ the exponential of their sum could be expressed as
$$ \exp(S_{1}+S_{2}) = \underset{m\to \infty}{\lim}(\exp(\frac{S_{1}}{m})\exp(\frac{S_{2}}{m}))^{m}. $$

\vspace{.1in}
Thus 
\begin{eqnarray*}
&& R_{1} - R_{2}  \\
&=& \exp(K_{2} + K_{1} - K_{2}) - \exp(K_{2}) \\
&=& \underset{m\to \infty}{\lim} \{ [\exp(\frac{K_{2}}{m})\exp(\frac{K_{1}-K_{2}}{m})]^{m} - (\exp(\frac{K_{2}}{m}))^{m} \} \\
&=& \underset{m\to \infty}{\lim} [ \exp(\frac{K_{2}}{m})\exp(\frac{K_{1}-K_{2}}{m})-\exp(\frac{K_{2}}{m}) ] \\
&& \quad \quad \times  \quad [ \sum_{i=1}^{m} ( \exp(\frac{K_{2}}{m})\exp(\frac{K_{1}-K_{2}}{m}))^{i}(\exp(\frac{K_{2}}{m}))^{m-i} ].  \\
\end{eqnarray*}

\newpage
Since $K_{1}, K_{2}$ are skew-symmetric matrices, we have $\|\exp(\frac{K_{1} - K_{2}}{m})\| = \|\exp(\frac{K_{2}}{m})\| = 1$, and

\begin{eqnarray} 
&& \|R_{1}-R_{2}\|_{2\to\infty} \nonumber \\
&\leq& \|R_{1}-R_{2}\| \nonumber \\
&=& \|\exp(K_{2} + K_{1} - K_{2}) - \exp(K_{2})\| \nonumber \\
&=&  \underset{m\to \infty}{\lim} \|[ \exp(\frac{K_{2}}{m})\exp(\frac{K_{1}-K_{2}}{m})-\exp(\frac{K_{2}}{m}) ] [ \sum_{i=1}^{m} ( \exp(\frac{K_{2}}{m})\exp(\frac{K_{1}-K_{2}}{m}))^{i}(\exp(\frac{K_{2}}{m}))^{m-i} ]\| \nonumber \\
&\leq& \underset{m\to \infty}{\lim} \|\exp(\frac{K_{2}}{m})\|\cdot\|	\exp(\frac{K_{1}-K_{2}}{m}) - I\| (\sum_{i=1}^{m} \| \exp(\frac{K_{2}}{m}) \|^{m-i} \cdot \|\exp(\frac{K_{1}-K_{2}}{m}) \|^{i}) \nonumber \\
&=& 	\underset{m\to \infty}{\lim} \|\exp(\frac{K_{1}-K_{2}}{m}) - I\|(m-1) \nonumber \\
&=&  \underset{m\to \infty}{\lim}	\|\sum_{i=1}^{\infty}(\frac{K_{1}-K_{2}}{m})^{i}\|(m-1) \nonumber \\
&\leq&  \underset{m\to \infty}{\lim} \sum_{i=1}^{\infty}\|\frac{K_{1}-K_{2}}{m}\|^{i}(m-1) \nonumber \\
&=& O_{p}(\|K_{1}-K_{2}\|) \nonumber \\
&\leq& O_{p}(\|K_{1}-K_{2}\|_{F}) \nonumber \\
&=& O_{p}\left(  (n\rho_{n})^{-1/4} \log^{\frac{7}{4}} n \right).     \label{eq:R1R2}
\end{eqnarray}

\vspace{.2in}
Therefore,
\begin{eqnarray*}
&&\|\sqrt{n}\widehat{U}R_{\widehat{U}} - \sqrt{n}\widehat{U}R_{UW}P_n^{(3)}\|_{2\to \infty}  \\
(\mbox{Lemma }\ref{lemma:twonorm})&\leq& \sqrt{n}\|\widehat{U}\|_{2\to\infty}\|R_{\widehat{U}} - R_{UW}P_n^{(3)}\| \\
			   &\leq& \sqrt{n}(\|\widehat{U} - UW\|_{2\to\infty} + \|UW\|_{2\to\infty})\|R_{\widehat{U}} - R_{UW}P_n^{(3)}\| \\
(\mbox{Equation }\eqref{eq:maxU})	&=& O_p(\log n)\times \|R_{\widehat{U}} - R_{UW}P_n^{(3)}\| \\
(\mbox{Equation }\eqref{eq:R1R2})  &=& O_{p}\left(  (n\rho_{n})^{-1/4} \log^{\frac{11}{4}} n \right).
\end{eqnarray*}

\end{proof}

\newpage

\section{First and Second Order Condition for Population Varimax}\label{sec:FSOC}

This section exploits the first and second order condition of the population Varimax function based on the similar results of the sample Varimax function (\cite{sherin1966matrix}, \cite{neudecker1981matrix}, \cite{chu1998orthomax}). This section is self-contained and only reuses the notations of \eqref{eq:Varimax} and the definition of Assumption \ref{assumption:Varimax}. We redefine some notations here.

\vspace{.1in}
\begin{assumption} \label{assumption:fsoc}
$U = ZR_U^{T}$ with $R_U \in \mathcal{O}(k)$, $Z \in \R^{n\times k}$ with $Z$ satisfying Assumption \ref{assumption:Varimax}. Let $z_0$ represents first row of $Z$. $O = z_0 z_0^T$. $z_i$ is the $i$th element of $z_0$ with $\mathbb{E}(z_i)=0, \forall i\in[k]$. Population Varimax function is $\vv(R) = \mathbb{E}(v(R, U))$.  
\end{assumption}

\vspace{.1in}

\noindent Optimization conditions for population Varimax function borrows conclusions from \cite{chu1998orthomax}. The math requires switching the order of expectation and differential operations. Lemma \ref{lemma:regularization} shows that this is valid for Varimax function. The proof of the lemmas and Theorems in current section are all contained in Section \ref{sec:Eproof}. 

\vspace{.1in}
\begin{lemma} \label{lemma:regularization} Under Assumption \ref{assumption:fsoc}, the expectation operator and differential operator of Varimax function are exchangeable,
$$  \frac{\partial \mathbb{E}v(R, U)}{\partial R} = \mathbb{E}\frac{\partial v(R, U)}{\partial R}.$$
\end{lemma}


\vspace{.1in}
\subsection{First Order Condition (FOC)}
The First Order Condition for sample varimax function is (\cite{sherin1966matrix}, \cite{neudecker1981matrix}): 
\begin{equation} \label{eq:samplefoc}
(U^{T}URD - U^{T}H)R^{T} = R(U^{T}URD - U^{T}H)^{T}
\end{equation}
Where $R$ is a stationary point of $v(R, U)$. $D$ is a diagonal matrix with $j$-th element equals to $\frac{1}{n}\sum_{i=1}^{n}(UR)_{ij}^{2}$. $H$ is $n \times k$ matrix with $(H)_{ij} = (UR)_{ij}^{3}$.

\begin{theorem}[FOC] \label{thm:foc}
Under Assumption \ref{assumption:fsoc}, if $R\in \oo(k)$ is a stationary point of $\vv(R)$, then
\begin{equation}\label{equation:foc}
\mathbb{E}z_{i}^{2}\mathbb{E}(z_{i}z_{j}) - \mathbb{E}(z_{i}^{3}z_{j}) = \mathbb{E}z_{j}^{2}\mathbb{E}(z_{i}z_{j}) - \mathbb{E}(z_{j}^{3}z_{i}), \forall i\neq j.
\end{equation}
\end{theorem}

With Assumption \ref{assumption:fsoc}, Theorem \ref{thm:foc} is a trivial result. The FOC only tells about local stationary points. To consider curvature information and ensure the stationary point is local maxima, we also need to figure out the Second Order Condition.

\subsection{Second Order Condition (SOC)} 

\cite{chu1998orthomax} shows SOC result for sample Varimax function. The current subsection is deriving counterpart results of the population Varimax function. To describe the SOC on sample data, \cite{chu1998orthomax} reformulate the Varimax criterion and express the problem in a simultaneously diagonalizing symmetric matrices form (\cite{ten1984joint}). The detailed SOC statement for sample Varimax is shown below.

\vspace{.15in}

Write $E_{i} = U^{T}U - nu_{i}u_{i}^{T}$ with $u_{i}^{T}$ being $U$'s $i$-th row. The sufficient (necessary) SOC of $v(R,U)$ is:
\begin{equation} \label{equation:soc}
\begin{aligned}
\sum^{n}_{i=1}(\langle U^{T}E_{i}Rdiag(R^{T}E_{i}R),K^{2}\rangle+\langle R^{T}E_{i}RKdiag(R^{T}E_{i}R),K\rangle \\
+2\langle R^{T}E_{i}Rdiag(R^{T}E_{i}RK),K\rangle) <(\leq) 0,
\end{aligned}
\end{equation}

for any non-zero skew-symmetric matrix $K$. Since Varimax condition gives us a special covariance structure of $z_0$ (e.g. $Cov(z_0) = I$), we could derive SOC for population Varimax function from \eqref{equation:soc}.

\begin{theorem}[SOC] \label{thm:soc}
Under Assumption \ref{assumption:fsoc},
a sufficient (necessary) condition for $R_U$ to be one of the maximas of the population Varimax is
\begin{equation}
3\mathbb{E}[tr(diag(OK)^{2})] <(\leq) \mathbb{E}\langle OdiagO,KK^{T}\rangle.
\end{equation}
\end{theorem}

%
%

\vspace {.2in}
\subsection{Proofs in Section \ref{sec:FSOC}} \label{sec:Eproof}

\subsubsection{Proof of Lemma \ref{lemma:regularization}}

\begin{proof}  The main idea of the proof is applying Dominant Converge Theorem (DCT). For simplicity, write $\mathbb{E}[U_{ij}^{q}] = \mu_{q}^{(j)}$ as $q$-th moment of $U$'s $j$-th column and $G_{i} = U^{T}U - nu_{i}u_{i}^{T}$, $i \in [n]$, with $u_{i}^{T}$ being the $i$-th row of $U$. By (8) of \cite{chu1998orthomax},
\begin{equation} \label{eq:differential}
 \frac{\partial v(R, U)}{\partial R} = \frac{4}{n^{3}}\sum_{i=1}^{n}G_{i}Rdiag(R^{T}G_{i}R).
\end{equation}

\noindent The goal is to bound the spectral norm of RHS of Equation \eqref{eq:differential}. Notice for $\forall i \in [n]$,
\begin{eqnarray} \label{eq:upperboundreg}
\|G_{i}Rdiag(R^{T}G_{i}R)\| &\leq& \|G_{i}R\| \cdot \|diag(R^{T}G_{i}R)\| \nonumber \\
					 &=& \|G_{i}\| \cdot \|diag(R^{T}G_{i}R)\| \nonumber \\
					 &=& \|U^{T}U - nu_{i}u_{i}^{T}\| \cdot \|diag(R^{T}U^{T}UR) - diag(nR^{T}u_{i}u_{i}^{T}R)\| \nonumber \\
					 &\leq& (\|U^{T}U\| + n\|u_{i}u_{i}^{T}\|) \times \nonumber \\
					 && \hspace{.2in}(\|diag(R^{T}U^{T}UR)\| + n\|diag(R^{T}u_{i}u_{i}^{T}R)\|).
\end{eqnarray}

\noindent Basic matrix algebra implies
$$ \|U^{T}U\| \leq \|U^{T}\|\cdot \|U\| = \|U\|^{2} \leq \|U\|^{2}_{F}, \quad \|u_{i}u_{i}^{T}\| \leq \|u_{i}\|^{2}. $$

\noindent Notice that the $i$-th diagonal element of $R^{T}U^{T}UR$ is
\begin{eqnarray*}
\sum_{t=1}^{k}[(\sum_{s=1}^{k}U_{is}R_{st})^{2}] &\leq&  \sum_{t=1}^{k}[(\sum_{s=1}^{k}U_{is}^{2})(\sum_{s=1}^{k}R_{st}^{2})] \\
								           &=&  k\sum_{s=1}^{k}U_{is}^{2}
\end{eqnarray*}
$\Rightarrow$
\begin{equation} \label{eq:reg1}
\|diag(R^{T}U^{T}UR)\| \leq \sum_{i=1}^{n} k(\sum_{s=1}^{k}U_{is}^{2}) = k\|U\|_{F}^{2}. 
\end{equation}

\noindent Similarly, 
\begin{equation} \label{eq:reg2}
\|diag(R^{T}u_{i}u_{i}^{T}R)\| \leq k\|u_{i}\|^{2}. 
\end{equation}

\noindent Plugging Equation \eqref{eq:reg1}, \eqref{eq:reg2} into Equation \eqref{eq:upperboundreg} yields
\begin{eqnarray*}
\|G_{i}Rdiag(R^{T}G_{i}R)\| &\leq& (\|U\|_{F}^{2} + n\|u_{i}\|^{2})(k\|U\|_{F}^{2} + nk\|u_{i}\|^{2}) \\
					 &=& k\|U\|_{F}^{4} + 2kn\|u_{i}\|^{2}\|U\|_{F}^{2} + kn^{2}\|u_{i}\|^{4}. 
\end{eqnarray*}

\noindent Getting back to Equation \eqref{eq:differential}, we have
\begin{eqnarray*}
\| \frac{\partial v(R, U)}{\partial R} \| &\leq& \frac{4}{n^{3}}\sum_{i=1}^{n}\|G_{i}Rdiag(R^{T}G_{i}R)\| \\
						      &\leq& \frac{4}{n^{3}}\sum_{i=1}^{n}(k\|U\|_{F}^{4} + 2kn\|u_{i}\|^{2}\|U\|_{F}^{2} + kn^{2}\|u_{i}\|^{4}) \\
						      &:=& F_{n}. 
\end{eqnarray*}

\noindent The $F_{n}$ is a random variable (depends on norms of random matrix and vectors). It will be sufficient to give constant bound to the expectation of each term of $F_n$. Notice 
\begin{eqnarray*}
\mathbb{E}\|U\|_{F}^{4} &=& \mathbb{E}[(\sum_{ij}U_{ij}^{2})^{2}] \\
				    &=& n\sum_{j=1}^{k}\mu_{4}^{(j)} + {n \choose 2}\sum_{j=1}^{k}\mu_{2}^{(j)2} + n^{2}\sum_{1\leq \ell\neq j \leq k}\mu_{2}^{(\ell)}\mu_{2}^{(j)},   \\
\mathbb{E}(\|u_{i}\|^{2}\|U\|_{F}^{2}) &=& \mathbb{E}[(\sum_{j=1}^{k}U_{ij}^{2})(\sum_{i,j}U_{ij}^{2})] \\
						      &=& \sum_{j=1}^{k}\mu_{4}^{(j)} + (n-1)\sum_{j=1}^{k}\mu_{2}^{(j)2} + n\sum_{1\leq \ell\neq j \leq k}\mu_{2}^{(\ell)}\mu_{2}^{(j)}, \\
\mathbb{E}(\|u_{i}\|^{4}) &=& \mathbb{E} [(\sum_{j=1}^{k}U_{ij}^{2})^{2}] \sum_{j=1}^{k}\mu_{4}^{(j)} + \sum_{1\leq \ell\neq j \leq k}\mu_{2}^{(\ell)}\mu_{2}^{(j)}.
\end{eqnarray*}

\noindent Therefore 
\begin{eqnarray*}
\mathbb{E}(F_{n}) &=& \frac{4k}{n^{2}}\mathbb{E}\|U\|_{F}^{4} + \frac{8k}{n^{2}}\sum_{i=1}^{n}\mathbb{E}(\|u_{i}\|^{2}\|U\|_{F}^{2}) + \frac{4k}{n}\sum_{i=1}^{n}\mathbb{E}(\|u_{i}\|^{4})\\
		    &=& (4k+\frac{12k}{n})M_{1} + \frac{17k(n-1)}{n}M_{2} + 16kM_{3} \\
		    &\leq& 5kM_1 + 17kM_2+16kM_3 \\
		    &<& \infty.
\end{eqnarray*}

\noindent Here $M_{1} = \sum_{j=1}^{k}\mu_{4}^{(j)}$, $M_{2} = \sum_{j=1}^{k}\mu_{2}^{(j)2}$, $M_{3} = \sum_{1\leq \ell\neq j \leq k}\mu_{2}^{(\ell)}\mu_{2}^{(j)}$ are all constants in our settings. Then DCT accomplishes our proof.

\end{proof}

\subsubsection{Proof of Theorem \ref{thm:soc}}
The following Lemma is useful in the proof of Theorem \ref{thm:soc}.
\begin{lemma} \label{lemma:tr}
For any symmetric matrix $S = vv^{T}$ where $v$ is a k-dimension vector. Any $k\times k$ matrix P. We have
\begin{equation}
\langle Sdiag(SP),P\rangle = \langle SPdiag(S),P\rangle .
\end{equation}
\end{lemma}
\begin{proof} Let $(S)_{i,j} = v_{i}v_{j}$, $(P)_{i,j}= P_{ij}$. We only need to prove $Sdiag(SP) = SPdiag(S)$. \\
For $diag(SP)$. Its j-th diagonal element is $v_{j}\sum_{k}v_{k}P_{kj}$. Multiplying a diagonal matrix on right side is equal to multiplying i-th diagonal element to i-th column. Thus $(Sdiag(SP))_{ij} = v_{i}v_{j}^{2}\sum_{k}v_{k}P_{kj}$. \\

For $SPdiag(S)$ we have $(SP)_{ij} = v_{i}\sum_{k}v_{k}P_{kj} \Rightarrow (SPdiag(S))_{ij} = v_{i}v_{j}^{2}\sum_{k}v_{k}P_{kj}$.     

\end{proof}

\vspace{.2in}

Now we return to the proof of Theorem \ref{thm:soc}
\begin{proof}   
By Lemma \ref{lemma:regularization} and Slutsky Theorem,
\begin{eqnarray*}
\mathbb{E} [ R_U^{T}E_{i}R_Udiag(R_U^{T}E_{i}R_U)] &=&  n^2\mathbb{E}[(I-O)diag(I-O)], \\
\mathbb{E}[ R_U^{T}E_{i}R_UKdiag(R_U^{T}E_{i}R_U)]  &=& n^2\mathbb{E}[(I-O)Kdiag(I-O)], \\ 
\mathbb{E}[ R_U^{T}E_{i}R_Udiag(R_U^{T}E_{i}R_UK)] &=& n^2\mathbb{E}[(I-O)diag(K-OK)]. 
\end{eqnarray*}

The expectation of the first term of Equation \eqref{equation:soc} equals to
\begin{eqnarray} \label{eq:soc1}
n^2\mathbb{E}\langle((I-O)diag(I-O),K^{2}\rangle &=& n^2\mathbb{E}\langle I - O - diagO + OdiagO,K^{2}\rangle \nonumber \\
&=& n^2(\mathbb{E}\langle Odiag(O),K^{2}\rangle - \langle I,K^{2}\rangle).
\end{eqnarray}

Similarly, the expectation of the second term of Equation \eqref{equation:soc} is
\begin{equation} \label{eq:soc2}
n^2\mathbb{E}\langle(I-O)Kdiag(I-O), K\rangle = n^2(\mathbb{E}\langle OKdiag(O),K\rangle - \langle K,K\rangle), 
\end{equation}

and the expectation of the third term of Equation \eqref{equation:soc} is
\begin{equation} \label{eq:soc3}
n^2\mathbb{E}\langle(I-O)diag((I-O)K), K\rangle = n^2(\mathbb{E}\langle Odiag(OK),K\rangle - \langle diag(K),K\rangle). 
\end{equation}

Notice $K$ is skew-symmetric, there are
\begin{equation} \label{eq:traceK}
\langle I, K^2 \rangle = tr(K^2)  = -tr(KK^T) = -\langle K,K\rangle, \langle diag(K), K \rangle = 0. 
\end{equation}

By Lemma \ref{lemma:tr},
\begin{equation} \label{eq:lemmatr}
\langle Odiag(OK),K\rangle = \langle OKdiag(O),K\rangle.
\end{equation}

Properties of trace operator indicate that $tr(Ydiag(Y)) = tr((diag(Y)^{2})$ for any square matrix $Y$. Then with Equation \eqref{eq:soc1}, \eqref{eq:soc2}, \eqref{eq:soc3}, \eqref{eq:traceK}, \eqref{eq:lemmatr}, the second order condition \eqref{equation:soc} boils down to 
\begin{equation}
3\mathbb{E}[tr(diag(OK)^{2})] < (\leq) \mathbb{E}\langle OdiagO,KK^{T}\rangle,
\end{equation}

for any non-zero skewed symmetric matrix $K$.
\end{proof}

\section{Proofs of Corollaries \ref{corollary:dcsbm} and \ref{corollary:pois}}\label{sec:dcsbm}

As we pointed out in Section \ref{sec:block}, independent columns assumption does not hold in the degree-corrected stochastic block model (DC-SBM). We could still make use of the first and second order condition to show that vsp could estimate $Z$ correctly. Similar to Proposition \ref{prop:svd} we have
\begin{equation} \label{eq:tai0}
U = \frac{1}{\sqrt{n}}Z(Z^{T}Z/n)^{-\frac{1}{2}}\widetilde{R}_{U}, \quad \widetilde{R}_{U} \in \oo(k).
\end{equation}
The proof of Corollary \ref{corollary:dcsbm}, \ref{corollary:pois} will focus on validating some key conditions and assumptions.

\subsection{Proof of Corollary \ref{corollary:dcsbm}}

\begin{proof} 
To borrow the conclusion from Theorem \ref{thm:main}, it will be sufficient to check some results. Notice we don't have the centering step and symmetrized adjacency matrices in SBM, such difference only simplifies the proof without introducing extra layers of perturbation. The only things we have to check (because of the dependency of $Z$'s columns) are conclusions of Theorem \ref{thm:varimax1}, Theorem \ref{thm:foc}, \ref{thm:soc}, Lemma \ref{lemma:Chung}, Assumption \ref{assumption:atail} and the arguments in the proof of Lemma \ref{lemma:SOC} that shows the existence of moment generating function of linear inner-product of $Z_i$ ($Z$'s $i$th row). 

\vspace{.2in}

\subsubsection{Theorem \ref{thm:varimax1} Under DC-SBM}

Recall that in current setting, 
$$ \vv_{\widetilde{R}_{U}}(Q) = \sum_{j}(\mathbb{E}([Z\widetilde{R}_{U}Q]_{j}^{4}) - \mathbb{E}([Z\widetilde{R}_{U}Q]_{j}^{2})^{2}). $$

We want to show 
$$ \underset{{Q \in \oo(k)}}{\arg \max} \mbox{ }\vv_{\widetilde U}(Q) = \{\widetilde{R}_{U}^{T}P | P \in \mathcal{P}(k) \}. $$

Let $X = Z_1 - \mathbb{E}(Z_1)$, $\mathbb{E}(X_{i}^{j}) = \mu_{j}^{(i)}$. Since there is exactly one non-zero entry in $X$'s elements, we have

$$ \sum_{j=1}^{k}\mathbb{E}([XQ]^{2}_{j})^{2} = \sum_{j=1}^{k}\mathbb{E}(\sum_{i=1}^{k}X_{i}^{2}Q_{ij}^{2})^{2} = \sum_{j=1}^{k}\sum_{i=1}^{k}\mu_{2}^{(i)2}Q_{ij}^{4}, $$

and 

$$ \sum_{j=1}^{k}\mathbb{E}([XQ]_{j}^{4}) = \sum_{j=1}^{k}\sum_{i=1}^{k}\mathbb{E}(X_{i}^{4}Q_{ij}^{4}) = \sum_{i=1}^{k}\mu_{4}^{(i)}Q_{ij}^{4}. $$

Therefore

\begin{eqnarray*} 
\vv_{I}(Q) &=& \sum_{j}(\mathbb{E}([XQ]_{j}^{4}) - \mathbb{E}([XQ]_{j}^{2})^{2}) \\
		&=& \sum_{i=1}^{k}(\mu_{4}^{(i)} - \mu_{2}^{(i)2})\|Q_{i\cdot}\|_{4}^{4}.
\end{eqnarray*}

Jensen Inequality indicates that $\mu_{4}^{(i)} - \mu_{2}^{(i)2} > 0$ for $\forall i \in [k]$. The remaining part follows the same approach as in the proof of Theorem \ref{thm:varimax1}.

\vspace{.2in}

%
%
%
%
%
\subsubsection{Theorem \ref{thm:foc} Under DC-SBM}

For $\forall i \in [n]$ and $j,\ell \in [k], j\neq \ell$, we have 
$$ Z_{ij}Z_{i\ell} = 0, Z_{ij}^3 Z_{i\ell} = 0 \Rightarrow  \E [Z_{ij}Z_{i\ell}] = 0, \E[Z_{ij}^3 Z_{i\ell}] = 0.$$ 

\subsubsection{Theorem \ref{thm:soc} Under DC-SBM}

To show SOC. Let $E_i = U^T U - n u_i^T u_i, u_i$ is the $i$th row of $U$. It is sufficient to show
\begin{equation} \label{eq:dcsbmSOC}
\begin{aligned}
\mathbb{E}[\langle \widetilde{R}_{U}E_{i}\widetilde{R}_{U}^{T}diag(\widetilde{R}_{U}E_{i}\widetilde{R}_{U}^{T}),K^{2}\rangle+\langle \widetilde{R}_{U}E_{i}\widetilde{R}_{U}^{T}Kdiag(\widetilde{R}_{U}E_{i}\widetilde{R}_{U}^{T}),K\rangle \\
+2\langle \widetilde{R}_{U}E_{i}\widetilde{R}_{U}^{T}diag(\widetilde{R}_{U}E_{i}\widetilde{R}_{U}^{T}K),K\rangle] \leq 0,
\end{aligned}
\end{equation}

Let $E^{i,j}$ be $k$ by $k$ matrix with 1 in $(i,j)$ entry and 0's elsewhere. Then
\begin{equation}
\widetilde{R}_{U}E_{i}\widetilde{R}_{U}^{T} = \widetilde{R}_{U}\widetilde{R}_{U}^T - n\widetilde{R}_{U}u_i^T u_i \widetilde{R}_{U}^T = I - nZ_i^T Z_i E^{z(i),z(i)}.
\end{equation}

Let 
$K=\left(\begin{array}{cccc}
0 & K_{1,2} & ... & K_{1,k}  \\
K_{2,1} & 0 & ... & K_{2,k}  \\
... & ... & ... & ...\\
K_{k,1} & K_{k,2} & ... & 0
\end{array}\right)$ with $K_{i,j} = - K_{j, i}$ for $i > j$. Then
$diag(K^{2})= diag(-\sum_{i\neq 1}K^{2}_{1i}, -\sum_{i\neq 2}K^{2}_{2i}, ...-\sum_{i\neq k}K^{2}_{ki})$. Now we examine each term of \eqref{eq:dcsbmSOC},

\begin{eqnarray}
&& \langle \widetilde{R}_{U}E_{i}\widetilde{R}_{U}^{T}diag(\widetilde{R}_{U}E_{i}\widetilde{R}_{U}^{T}),K^{2}\rangle \nonumber \\
&=& \langle \langle diag(1, 1, ... , (1-n\theta_{i,z(i)}^{2})^{2}, ... , 1), K^{2} \rangle \nonumber \\
&=& \sum_{\ell,j}K_{\ell j}^{2} - [(n\theta_{i,z(i)}^{2})^{2} - 2n\theta^{2}_{i,z(i)}]\sum_{\ell=1}^{k}K_{z(i)\ell}^{2},  \label{eq:dcsbm1} \\
\mbox{and} && \\
&& \langle \widetilde{R}_{U}E_{i}\widetilde{R}_{U}^{T}Kdiag(\widetilde{R}_{U}E_{i}\widetilde{R}_{U}^{T}),K\rangle  \nonumber \\
&=& \langle diag(1, 1, ... ,1-n\theta_{i,z(i)}^{2}, ... , 1)Kdiag(1, 1, ... ,1-n\theta_{i,z(i)}^{2}, ... , 1), K\rangle \nonumber \\
&=& \sum_{\ell,j}K_{\ell j}^{2} - 2n\theta_{i,z(i)}^{2}\sum_{\ell}^{k}K_{z(i)\ell}^{2}.   \label{eq:dcsbm2}
\end{eqnarray}

Since $K$'s diagonal elements are all zero's. $diag(\widetilde{R}_{U}E_{i}\widetilde{R}_{U}^{T}K)$ will be zero matrix.                                                                                                                                                                                                                                                                                                                                                                                                                                                                                                
\begin{equation}  \label{eq:dcsbm3}
\langle \widetilde{R}_{U}E_{i}\widetilde{R}_{U}^{T}diag(\widetilde{R}_{U}E_{i}\widetilde{R}_{U}^{T}K),K\rangle = 0
\end{equation}

From Equations \eqref{eq:dcsbm1}, \eqref{eq:dcsbm2}, \eqref{eq:dcsbm3}, to prove Equation \eqref{eq:dcsbmSOC} it will be suffice to show:
$$\sum_{i}( \sum_{\ell j}^{k}K_{\ell j}^{2} + [(n\theta_{i,z(i)}^{2})^{2} - 2n\theta^{2}_{i,z(i)}]\sum_{\ell=1}^{k}K_{z(i)\ell}^{2}) \geq \sum_{i}(\sum_{\ell, j}K_{\ell j}^{2} - 2n\theta_{i}^{2}\sum_{j}K_{z(i)\ell}^{2}), $$
$\Leftrightarrow$
$$ \sum_{i}([(n\theta_{i,z(i)}^{2})^{2} - 2n\theta^{2}_{i,z(i)}]\sum_{\ell=1}^{k}K_{z(i)\ell}^{2} +  2n\theta_{i,z(i)}^{2}\sum_{j}K_{z(i)\ell}^{2}) \geq 0,  $$
$\Leftrightarrow$
\begin{equation}
\sum_{i}((n\theta_{i,z(i)}^{2})^{2}\sum_{\ell=1}^{k}K_{z(i)\ell}^{2}) \geq 0.
\end{equation}

The last inequality is strict as long as $K$ is not zero matrix and $\theta_i$'s are all positive. We conclude that \eqref{eq:dcsbmSOC} is true. 

\vspace{.2in}
\subsubsection{Lemma \ref{lemma:Chung} Under DC-SBM}

Under DC-SBM, elements of $A$ are sub-gaussian variables. Thus we could utilize a simpler concentration matrix inequality than Lemma \ref{lemma:Chung2}. We apply the following lemma to show the bound for perturbation between $A$ and $\A$. 
\begin{lemma}[(Matrix Bernstein Inequality, \cite{tropp2012user})] \label{lemma:Chung3}
Let $X_1, X_2,..., X_m$ be independent random $N\times N$ symmetric matrix. Assume $\|X_i - \mathbb{E}(X_i)\| \leq M, \forall i$. Write $v^2 = \|\sum_i var(X_i)\|, X = \sum_i X_i$. Then for any $a > 0$, 
$$ \pr(\|X-\mathbb{E}(X)\| \geq a) \leq 2N\exp(-\frac{a^2}{2v^2 + 2Ma/3}).  $$
\end{lemma}

Let $E^{ij}$ be a $n$ by $n$ matrix with 1 in the $(i,j)$ and $(j,i)$ entries and 0 elsewhere. Write $p_{ij} = \A_{ij}$. Then we could express $A - \A$ as sum of matrices, 
$$ Y_{ij} = (A_{ij} - p_{ij})E^{ij}, i < j. $$

Notice that 
$$ \|A - \A\| = \|\sum_{1 \leq i < j \leq n} Y_{ij} \|, $$

and 
$$ \|Y_{ij}\| \leq \|E^{ij}\| = 1. $$

Moreover,
$$\mathbb{E}(Y_ij) = 0 \mbox{  and  }  \mathbb{E}(Y_{ij}^2) = (p_{ij}-p_{ij}^2)(E^{ii}+E^{jj}),  \forall i < j. $$

Then we could get an upper bound for $v^2$,

\begin{eqnarray*} 
v^2 &=& \|\sum_{1 \leq i < j \leq n} \mathbb{E}[Y_{ij}^2]\| \\
       &=& \|\sum_{1 \leq i < j \leq n} (p_{ij}-p_{ij}^2)(E^{ii}+E^{jj}) \| \\
       &=& \frac{1}{2} \|\sum_{1 \leq i,j \leq n} (p_{ij}-p_{ij}^2)(E^{ii}+E^{jj}) \| \\
       &=& \frac{1}{2} \|\sum_{i=1} \sum_{j\neq i} (p_{ij}-p_{ij}^2)E^{ii} \| \\
       &\leq&  \frac{1}{2} \underset{1\leq i \leq n}{\max}(\sum_{j\neq i}(p_{ij}-p_{ij}^2)) \\
       &\leq&  \frac{1}{2} \underset{1\leq i \leq n}{\max}(\sum_{j\neq i}p_{ij})  \\
       &\leq& \frac{n}{2}.
\end{eqnarray*}

From Lemma \ref{lemma:Chung3} we obtain, 
\begin{equation}
\pr(\|A - \A\| > a) \leq 2N\exp(-\frac{a^2}{n+ 2a/3}).
\end{equation}

\subsubsection{Assumption \ref{assumption:atail} with Bernoulli Random Variables}\label{sec:assumption3bern}
%


Suppose $A_{ij}$ has $m$-th central moment being $\mu_{m,ij}$ and $m$-th moment being $\mu_{m,ij}'$. Since $A_{ij}\sim Bernoulli(\A_{ij})$, then $\bar\rho_n \leq 1$. For any $m$,
\begin{equation} \label{eqref:bernoulliB}
 \mathbb{E} [(A_{ij} - \A_{ij})^m ]  = \mu_{m,ij} \leq |\mu_{m,ij}'| = |\A_{ij}| \leq \bar\rho_n, \forall i,j.  
 \end{equation}


\vspace{.2in}
\subsubsection{Arguments of Lemma \ref{lemma:SOC} under DC-SBM}

Since each $Z_i, i\in[k]$ has only one non-zero entry, for $\forall r \in \R^k, i\in[k]$, we have
$$ \E\exp(t\langle Z_i, r\rangle) = \E \exp(t\sum_{j=1}^k Z_{ij}r_j) = \E \exp(tZ_{i,z(i)}r_{z(i)}) = \Pi_{j=1}^{k} \E \exp(tZ_{ij}r_j). $$

\end{proof}

\subsection{Proofs for Corollary \ref{corollary:pois}} \label{sec:LDAproof}
\begin{proof}
In current LDA settings, we need Assumption \ref{assumption:ztail} in Theorem \ref{thm:main} on $\widetilde{Z}_{*}=Z_{*} - \mathbb{E}(Z_{*})$, which is already implied in Corollary \ref{corollary:pois} setup. Recall that:
\begin{equation} \label{eq:ldafactor}
\mathbb{E}(\breve{A}|\Xi, Z) = \widetilde{Z}_{*}(\sqrt{n}\Sigma^{1/2}) (n^{-1/2}\beta^{T}).
\end{equation}

Compared with the semi-parametric factor model in Definition \ref{def:model},  $\sqrt{n}\Sigma^{1/2}$ plays the role of block matrix $B$ and satisfies all the conditions in Theorem \ref{thm:main}. Other than that Assumption \ref{assumption:atail} needs to be checked to prove Equation \eqref{eq:docmember} and the $2\to\infty$ norm of $Y$ (in Equation \eqref{eq:ldafactor}, this is $(n\rho_n)^{-1}\Sigma^{1/2}\beta^{T}$) needs to be bounded by $O(\log d) \asymp O(\log n)$. After that we will show the error bound for topics estimation.

\vspace{.1in}

Notice that $s$ controls the scaling of the term-document matrix, the following inference reflects its relation to $\Delta_n$. Recall that
\[ \rho_n = \frac{1}{nd}\sum_{i,j}\A_{ij}, \quad \Delta = n\rho_n. \]
And,
\begin{equation} \label{eq:Xi}
\1_n^T \A \1_d = \1_n^T \Xi Z\beta^T \1_d = \1_n^T \Xi Z\1_k =   \1_n^T \Xi \1_n.
\end{equation}
Equation \eqref{eq:Xi} implies that
\begin{equation}
\Delta_n = n\rho_n = \frac{n}{nd} \1_n^T \A \1_d = \frac{1}{d}  \1_n^T \Xi \1_n \asymp s.
\end{equation}

%
%

\subsubsection{Assumption \ref{assumption:atail} Under Poisson Random Variables} \label{sec:assumption3pois}


Suppose $A_{ij}$ has $m$-th central moment being $\mu_{m,ij}$. Since $A_{ij}\sim Poisson(\A_{ij})$, recall the recurrence relation of poisson distribution (\cite{riordan1937moment}),
$$ \mu_{m+1,ij} = \A_{ij}(\frac{d\mu_{m,ij}}{d\A_{ij}} + m\mu_{m-1,ij}), \quad \mu_1 = 0, \mu_2 = \A_{ij}, \forall i,j. $$

It could be shown by induction that
$$ \mu_{m,ij} \leq (m-1)! \times \max\{\A_{ij}^{[\frac{m}{2}]}, \A_{ij}\}. $$

Thus,
\begin{equation} \label{eqref:poissonB}
 \mathbb{E} [(A_{ij} - \A_{ij})^m ]  \leq \max\{(m-1)!\bar \rho_n^{[\frac{m}{2}]}, \bar \rho_n\}  \leq \max\{(m-1)!\bar \rho_n^{\frac{m}{2}}, \bar \rho_n\} .
 \end{equation}


\vspace{.2in}
\subsubsection{Upper Bound for $\| n^{-1/2}  \|_{2\to\infty}$}

Notice for arbitrary $j$-th row of $n^{-1/2}\beta^T$, it has $\ell_2$-norm

$$ n^{-1/2} \sqrt{\sum_{\ell=1}^k \beta_{\ell j}^2} \leq n^{-1/2} \sqrt{\sum_{\ell=1}^k \beta_{\ell j}} = n^{-1/2}. $$

Therefore, $\| n^{-1/2}\beta^{T} \|_{2\to\infty} = O(n^{-\frac{1}{2}})$, which is much smaller than $O(\log n)$. 

\vspace{.2in}

\subsubsection{Topics Estimation}

For technical convenience, this proof uses an equivalent construction of $\widehat \beta$.  Define  $\Omega = (\widehat{Z}^{T}\widehat{Z})^{-1}\widehat{Z}^{T}\breve{A} = \Phi/n$ and $\widehat{\beta} = (\Lambda_o^{-1} \Omega)^T\in \mathbb{R}^{d\times k}$, where $\Lambda_o$ is a diagonal matrix with $i$th diagonal element equals to $\ell_1$-norm of $i$th row of $\Omega$.

For the topic estimation $\widehat{\beta}$, from Equation \eqref{eq:ldafactor} there is, 
$$ \breve{\A}^{T}\breve{\A} = n\beta\Sigma^{1/2}(\widetilde{Z}_{*}^{T}\widetilde{Z}_{*}/n)\Sigma^{1/2}\beta^T.$$

By LLN we have $(\widetilde{Z}_{*}^{T}\widetilde{Z}_{*}/n)[i,j] = \mathbbm{1}\{i=j\} + O(1/\sqrt{n})$. Notice that the $j$th diagonal element of $\Sigma_{jj} = \alpha_j s^2 \succeq n\rho_n$. Also $\sigma_{\min}(\beta) > c_1 > 0$, and
$$ \sigma_{\max}(\beta) = \|\beta\| < \|\beta\|_F = (\sum_{ij} \beta_{ij}^2)^{\frac{1}{2}} \leq (\sum_{ij} \beta_{ij})^{\frac{1}{2}} = k^{\frac{1}{2}} $$  
is upper bounded. Therefore 
$$ \sigma_{\min}(\breve{\A}) \asymp \sigma_{\min}((n\beta\Sigma\beta^T)^{\frac{1}{2}}) \asymp \sqrt{n}s \succeq n\rho_n. $$

%


With conclusions of Proposition \ref{prop:clt}, Lemma \ref{lemma:twonorm}, \ref{lemma:Chung}, Davis-Kahan $\sin\Theta$ Theorem, Equation \eqref{eq:docmember} and triangle inequality, there exists $P_n \in \mathcal{P}(k)$ (similar to the $P_n$ in Equation \eqref{eq:docmember}) s.t. for any $\delta, \epsilon > 0$, 

\begin{eqnarray*}
\|\Omega - P_n^T\Sigma^{1/2}\beta^{T}\|_{2\to\infty} &\leq& \frac{1}{n}\left[ \|(\widehat{Z}^{T}(\breve{A}-\breve{\A})\|_{2\to\infty} + \|(\widehat{Z}^{T}\widetilde{Z}_{*} - P_n^T)\Sigma^{1/2}\beta^{T}\|_{2\to\infty} \right] \\
&\leq& \frac{1}{n}\left[\|\widehat{Z}^{T}\|\|\breve{A}-\breve{\A}\| + \|\widehat{Z}^{T}\| \|\widetilde{Z}_{*} - \widehat{Z}P_n^T\| \|\Sigma^{1/2}\beta^{T}\| \right] \\
&=&\frac{1}{n}\left[ \|\widehat{Z}^{T}\|\|\breve{A}-\breve{\A}\| + \|\widehat{Z}^{T}\| \|\widehat{Z} - \widetilde{Z}_{*}P_n\| \|\Sigma^{1/2}\beta^{T}\|\right] \\
&\leq& \frac{1}{n}\left[\|\widehat{Z}^{T}\|\|\breve{A}-\breve{\A}\| + \sqrt{n} \|\widehat{Z}^{T}\| \|\widehat{Z} - \widetilde{Z}_{*}P_n\|_{2\to\infty} \|\Sigma^{1/2}\beta^{T}\| \right]\\
&=& O_{p}( \frac{\Delta_n^{1/2}\log^{5/2} n }{n}) + O_{p}( \frac{\Delta_n ^{3/4+\delta/2}\log^{15/4} n}{\sqrt{n}}) \\
&=& O_{p}(\frac{\Delta_n^{3/4+\delta/2}\log^{15/4} n}{\sqrt{n}}).
\end{eqnarray*}

Let $\Omega_\ell$ be the $\ell$th row of $\Omega$, $\zeta_{\ell}$ be the $\ell$th row of $P_n^T\Sigma^{1/2}\beta^T$. Then for $\forall \ell \in [k]$ there exists $\varepsilon_n = O_{p}(( \Delta_n ^{3/4+\delta/2}\log^{15/4} n)/\sqrt{n})$ s.t. with high probability,
\begin{equation} \label{eq:topic0}
\|\Omega_{\ell} - \zeta_{\ell}\| \leq \varepsilon_n \Rightarrow \|\Omega_{\ell} - \zeta_{\ell}\|_1 \leq \sqrt{d}\varepsilon_n.
\end{equation}

Notice any $\ell$-th column of $\beta$ has unit norm: $\|\beta_\ell\|_1 = 1, \forall \ell \in [k]$. Denotes $\alpha_{\min} = \underset{j}{\min\mbox{ }}\alpha_j, \alpha_{\max} = \underset{j}{\max\mbox{ }}\alpha_j$, then RHS of Equation \eqref{eq:topic0} reflects
\begin{equation} \label{eq:topic}
 s\sqrt{\alpha_{\min}}  - \sqrt{d}\varepsilon_n \leq \|\Omega_{\ell}\|_{1} \leq s\sqrt{\alpha_{\max}}  + \sqrt{d}\varepsilon_n.
\end{equation} 

With Equation \eqref{eq:topic} and notice the $j$-th diagonal element of $\Sigma^{1/2}$ is $s\sqrt{\alpha_j}$, we also have
\begin{equation} \label{eq:topic2}
\underset{j,\ell\in [k]}{\max}\mbox{ } |\Sigma^{1/2}_{jj} - \|\Omega_{\ell} \|_{1} | \leq \sqrt{d} \varepsilon_n.
\end{equation}

Since LHS of \eqref{eq:topic} is greater than 0 with large $n$. Let $[X]_{\ell}$ represents the $\ell$-th row of matrix $X$. Then for $\forall \ell \in [k]$, 
\begin{eqnarray*}
\|\widehat{\beta}_{\ell}^T - [P_n^T\beta^T]_{\ell}\|_{1} &=& \|\frac{\Omega_\ell }{\|\Omega_{\ell}\|_{1}}- [P_n^T\beta^T]_{\ell}\|_1 \\
						 &\leq& \frac{1}{\|\Omega_{\ell}\|_1}\|\Omega_{\ell} -  \zeta_{\ell}\|_1 + \| [P_n^T(\frac{\Sigma^{1/2}}{\|\Omega_{\ell}\|_1}-1)\beta^T]_{\ell}| \\
						 &\leq& \frac{\sqrt{d}\varepsilon_n}{s\sqrt{\alpha_{\min}}  - \sqrt{d}\varepsilon_n} + \frac{1}{\|\Omega_{\ell} \|_1}\times\underset{j\in [k]}{\max}\mbox{ } |\Sigma^{1/2}_{jj} - \|\Omega_{\ell} \|_{1}|\\
						 &\leq& \frac{2\sqrt{d}\varepsilon_n}{s\sqrt{\alpha_{\min}}  - \sqrt{d}\varepsilon_n}  \\
						 &=& O_p(\sqrt{d}\varepsilon_n/s) \\
						 &=& O_{p}( \Delta_n ^{-1/4+\delta/2}\log^{15/4} n).
\end{eqnarray*}

\end{proof}

\end{document}